\documentclass[11pt]{article}

\newif\ifcomments
\commentstrue

\usepackage[margin=1in]{geometry}                		%
\usepackage[utf8]{inputenc}

\usepackage{graphicx}
\usepackage{amsmath,amsthm,amssymb}
\usepackage{algorithm}
\usepackage{multirow}
\usepackage{makecell}
\usepackage{bbm}

\usepackage[noend]{algpseudocode}
\usepackage{algorithmicx}

\algdef{SxNE}[STRUCT]{Struct}{EndStruct}[1]{\textbf{struct} \textsc{#1}}{end struct}

\algdef{SxNE}[CLASS]{Class}{EndClass}[1]{\textbf{class} \textsc{#1}}{end class}

\algdef{SxNE}[DATA]{Data}{EndData}[0]{\textbf{data members}}{end data members}

\makeatletter
\ifthenelse{\equal{\ALG@noend}{t}}%
  {\algtext*{EndStruct}}
  {}%
\ifthenelse{\equal{\ALG@noend}{t}}%
  {\algtext*{EndClass}}
  {}%
\ifthenelse{\equal{\ALG@noend}{t}}%
  {\algtext*{EndData}}
  {}%
\makeatother

\algnewcommand{\LeftComment}[1]{\State {\color{blue}\(\triangleright\) #1}}

\renewcommand*\Call[2]{\textproc{#1}(#2)}

\usepackage{tikz}

\usepackage{thm-restate,color,xspace}
\usepackage{comment}
\usepackage{thmtools}
\usepackage{xcolor}

\usepackage{nameref}
\usepackage{array}
\usepackage{dsfont}

\usepackage{tikz}
\usetikzlibrary{shapes.geometric, arrows, positioning, calc}

\tikzset{
  prefix after node/.style={prefix after command=\pgfextra{#1}},
  /semifill/ang/.initial=45,
  /semifill/upper/.initial=none,
  /semifill/lower/.initial=none,
  semifill/.style={
    circle, draw,
    prefix after node={
      \pgfqkeys{/semifill}{#1}
      \path let \p1 = (\tikzlastnode.north), \p2 = (\tikzlastnode.center),
                \n1 = {\y1-\y2} in [radius=\n1]
            (\tikzlastnode.\pgfkeysvalueof{/semifill/ang}) 
            edge[
              draw=none,
              fill=\pgfkeysvalueof{/semifill/upper},
              to path={
                arc[start angle=\pgfkeysvalueof{/semifill/ang}, delta angle=180]
                -- cycle}] ()
            (\tikzlastnode.\pgfkeysvalueof{/semifill/ang}) 
            edge[
              draw=none,
              fill=\pgfkeysvalueof{/semifill/lower},
              to path={
                arc[start angle=\pgfkeysvalueof{/semifill/ang}, delta angle=-180]
                -- cycle}] ();}}}

\definecolor{ForestGreen}{rgb}{0.1333,0.5451,0.1333}
\definecolor{DarkRed}{rgb}{0.65,0,0}
\definecolor{Red}{rgb}{1,0,0}
\usepackage[linktocpage=true,
pagebackref=true,colorlinks,
linkcolor=DarkRed,citecolor=ForestGreen,
bookmarks,bookmarksopen,bookmarksnumbered]
{hyperref}
\usepackage{cleveref}
\usepackage[T1]{fontenc}

\definecolor{verylightgray}{rgb}{0.95, 0.95, 0.95} 
\definecolor{slightlylightgray}{rgb}{0.9, 0.9, 0.9}

\declaretheorem[numberwithin=section]{theorem}
\declaretheorem[numberlike=theorem,name=Theorem]{thm}
\declaretheorem[numberlike=theorem]{lemma}

\declaretheorem[numberlike=theorem]{fact}

\declaretheorem[numberlike=theorem]{corollary}
\declaretheorem[numberlike=theorem,name=Corollary]{cor}

\declaretheorem[numberlike=theorem]{claim}

\declaretheorem[numberlike=theorem]{observation}

\declaretheorem[numberlike=theorem,style=definition]{definition}

\newcommand{\calS}{\mathcal{S}}
\newcommand{\calN}{\mathcal{N}}
\newcommand{\hdiam}{h_\mathrm{diam}}

\newcommand{\hcov}{h_\mathrm{cov}}

\newcommand{\hed}{h_\mathrm{ed}}

\newcommand{\sldd}{s_\mathrm{nc}}
\newcommand{\sed}{s_\mathrm{ed}}

\newcommand{\id}{\mathrm{id}}
\newcommand{\vlabel}{\mathrm{VLabel}}
\newcommand{\elabel}{\mathrm{ELabel}}
\newcommand{\fp}{\mathrm{Fingerprint}}
\newcommand{\ita}{\mathrm{start}}
\newcommand{\itb}{\mathrm{end}}
\newcommand{\pos}{\mathrm{pos}}
\newcommand{\disc}{G^{\mathrm{disc}}}
\newcommand{\packdisc}{G^{\mathrm{pack-disc}}}
\newcommand{\qea}{p}
\newcommand{\qeb}{q}
\newcommand{\croot}{\mathrm{root}}
\newcommand{\clu}{\Gamma}

\newcommand{\theavy}{\tau_{\mathrm{heavy}}}
\newcommand{\thit}{\tau_{\mathrm{hit}}}

\DeclareMathOperator{\dist}{dist}
\DeclareMathOperator{\ball}{ball}
\DeclareMathOperator{\supp}{supp}
\DeclareMathOperator{\sep}{sep}
\DeclareMathOperator{\spars}{spars}
\DeclareMathOperator{\load}{load}
\DeclareMathOperator{\len}{len}

\global\long\def\poly{\mathrm{poly}}
\global\long\def\fail{\mathrm{fail}}

\global\long\def\low{\mathrm{low}}
\global\long\def\high{\mathrm{high}}
\global\long\def\constr{\mathrm{constr}}
\global\long\def\eexist{\mathrm{exist}}
\global\long\def\efp{\mathrm{efp}}
\global\long\def\vfp{\mathrm{vfp}}
\global\long\def\elb{\mathrm{elb}}
\global\long\def\vlb{\mathrm{vlb}}
\global\long\def\cov{\mathrm{cov}}
\global\long\def\diam{\mathrm{diam}}
\global\long\def\Conn{\textsc{Conn}}
\global\long\def\leng{\mathrm{leng}}
\global\long\def\ed{\mathrm{ed}}

\title{A Constant-Approximation Distance Labeling Scheme\\
under Polynomially Many Edge Failures}
\date{} %

\author{Bernhard Haeupler\thanks{
        INSAIT, Sofia University ``St.~Kliment Ohridski'' and ETH Zürich,
        \texttt{bernhard.haeupler@insait.ai}.
        Partially funded by the Ministry of Education and Science of Bulgaria's support for INSAIT as part of the Bulgarian National Roadmap for Research Infrastructure and through the European Research Council (ERC) under the European Union's Horizon 2020 research and innovation program (ERC grant agreement 949272).} \and Yaowei Long\thanks{University of Michigan, \texttt{yaoweil@umich.edu}. Part of this work was done while at INSAIT, Sofia University "St. Kliment Ohridski", Bulgaria. Partially funded by the Ministry of Education and Science of Bulgaria's support for INSAIT, Sofia University ``St.~Kliment Ohridski'' as part of the Bulgarian National Roadmap for Research Infrastructure.} \and 
Antti Roeyskoe\thanks{ETH Zürich, \texttt{antti.roeyskoe@inf.ethz.ch}. Part of this work was done while at INSAIT, Sofia University "St. Kliment Ohridski", Bulgaria.} \and  Thatchaphol Saranurak\thanks{
        University of Michigan,
        \texttt{thsa@umich.edu}.
        Supported by NSF Grant CCF-2238138. Partially funded by the Ministry of Education and Science of Bulgaria's support for INSAIT, Sofia University ``St.~Kliment Ohridski'' as part of the Bulgarian National Roadmap for Research Infrastructure.
    }   }

\begin{document}

\maketitle

\pagenumbering{gobble}
\begin{abstract}

A \emph{fault-tolerant distance labeling scheme} assigns a label to each vertex and edge of an undirected weighted graph $G$ with $n$ vertices so that, for any edge set $F$ of size $|F| \le f$, one can approximate the distance between $\qea$ and $\qeb$ in $G \setminus F$ by reading only the labels of $F \cup \{\qea,\qeb\}$.

For any $k$, we present a deterministic polynomial-time scheme with $O(k^{4})$ approximation and $\tilde{O}(f^{4}n^{1/k})$ label size.
This is the first scheme to achieve a constant approximation while handling any number of edge faults $f$, resolving the open problem posed by Dory and Parter \cite{DoryP21}. All previous schemes provided only a \emph{linear-in-$f$} approximation \cite{DoryP21,long2025connectivity}.

Our labeling scheme directly improves the state of the art in the simpler setting of  \emph{distance sensitivity oracles}. 
Even for just $f = \Theta(\log n)$ faults, all previous oracles either have super-linear query time, linear-in-$f$ approximation \cite{ChechikLPR12}, or exponentially worse $2^{\poly(k)}$ approximation dependency in $k$ \cite{haeupler2024dynamic}.

\end{abstract}

\newpage
\tableofcontents
\newpage

\pagenumbering{arabic}

\section{Introduction}

A \emph{labeling scheme} for a graph problem is a distributed data structure that assigns short \emph{labels} to vertices and/or edges, with the capability to answer queries given \emph{only} the assigned labels of the query arguments (and no other access to the underlying graph). Early work on labeling schemes focused on schemes for adjacency \cite{Breuer66,BreuerF67,KannanNR92}, with connections to induced universal graphs \cite{AlstrupDK17,AlstrupKTZ19}. Other schemes have been developed for graph connectivity \cite{KatzKKP04,HsuL09,IzsakN12,PettieSY22} and basic queries on trees, such as ancestry and lowest common ancestors \cite{AAKMR06,AlstrupHL14,AlstrupDK17}.

\paragraph{Distance Labeling Schemes: Well-Understood.}

One of the most extensively studied problems in this area is the \emph{distance labeling scheme} formalized by Peleg \cite{Peleg00b}. Given an $n$-vertex graph $G=(V,E)$, the task is to assign a label $L(u)\in\{0,1\}^{s}$ to each vertex $u$ so that, for any $\qea,\qeb\in V$, the distance between $\qea$ and $\qeb$ in $G$ can be \emph{decoded} from only the labels $L(\qea)$ and $L(\qeb)$. A scheme with optimal label size of $s = \Theta(n)$ bits was established following the resolution of the squashed cube conjecture \cite{graham1972embedding} by Winkler \cite{Winkler83}, and the bound was further refined in \cite{gavoille2004distance,weimann2011note,AlstrupGHP16b,gawrychowski2016sublinear}. Allowing for approximation, Matousek's $\ell_{\infty}$-embedding \cite{Matousek96} implies labeling schemes with sublinear size. For any $k\ge1$, the scheme achieves a $2k-1$ approximation using $\tilde{O}(n^{1/k})$ label size.\footnote{$\tilde{O}(\cdot)$ hides logarithmic factors.} Thorup and Zwick \cite{TZ05} later improved the decoding time to $O(k)$. As argued in \cite{TZ05}, this approximation-size trade-off is tight under the Erd\H{o}s girth conjecture, and an unconditional lower bound of $n^{\Omega(1/k)}$ holds. Thus, the optimal trade-off is essentially settled.

\paragraph{Fault-Tolerant Distance Labeling Schemes.} %

The problem becomes significantly more challenging in the \textit{fault-tolerant} labeling setting introduced in \cite{CourcelleT07}. Here, we assign a label to each vertex and edge so that, for any $\qea,\qeb\in V$ and edge set $F\subseteq E$ of size $|F|\le f$, the distance between $\qea$ and $\qeb$ in $G\setminus F$, the graph where edges in $F$ fail, can be decoded from the labels $L(\qea)$, $L(\qeb)$, and $L(e)$ for all $e\in F$.

In the exact case, the best-known label size is $\tilde{O}(n^{2-1/2^{f}})$ \cite{bodwin2023restorable}, which is trivial when $f=\Omega(\log\log n)$ as one can store the entire graph in labels of size $\Theta(n^2)$. When allowing approximation, previous schemes \cite{DoryP21} can only guarantee an approximation ratio linear in the maximum number $f$ of faults. %
Specifically, in \cite{DoryP21}, Dory and Parter showed a reduction to fault-tolerant connectivity labeling schemes, where the task is to check connectivity instead of estimating distances. Plugging in known connectivity labeling schemes \cite{DoryP21,long2025connectivity}, their reduction implies, for any $k\ge1$, a labeling scheme with $(8k-2)(f+1)=O(kf)$ approximation and label size $\tilde{O}(n^{1/k})$ for a randomized scheme or $\tilde{O}(\sqrt{f}n^{1/k})$ for a deterministic one. Note that the approximation quality worsens as $f$ increases, and that notably, the scheme is unable to provide constant approximation under any nonconstant number of failures. %

Thus, Dory and Parter posed an open problem in \cite{DoryP21}: ``\emph{Finally, it will be also important to provide fault-tolerant distance approximate labeling schemes whose stretch bound is independent of the number of faults $f$}''.

\subsection{Main Result: Constant-Approximate FT Labeling Scheme}

We affirmatively answer the open problem posed by Dory and Parter: we present the first constant-approximate fault-tolerant distance labeling scheme that can handle \emph{any number} of faults.

\begin{thm}

\label{thm:main label} For every $k \ge 1$, there is a fault-tolerant distance labeling scheme for an undirected $n$-vertex graph with edge lengths in $[1, \poly(n)]$ undergoing $f$ edge faults, with $O(k^{4})$ approximation and $\tilde{O}(f^{4}n^{1/k})$ label size. The approximate distance can be decoded from the input labels in near-linear $\tilde{O}(f^{5}n^{1/k})$ time. The labels can be computed deterministically in polynomial time.

\end{thm}

\Cref{thm:main label} completely removes the dependency on $f$ from the approximation ratio of previous schemes \cite{DoryP21}, and has only a \emph{polynomial} label size dependency on $f$. In other words, we can tolerate even a polynomial number of failures $f=O(n^{1/k})$, while still guaranteeing a constant approximation and small $n^{O(1 / k)}$ label size.

The improvement comes at the cost of a worse $O(k^4)$-approximation dependency on $k$. We note that this dependency can be improved to $O(k^{3})$ if exponential construction time is allowed. See \Cref{thm:labeling} for the more detailed statement of \Cref{thm:main label}.

\subsection{Consequence: Improved Sensitivity Oracles}

\Cref{thm:main label} also significantly improves the state-of-the-art of one of the most extensively studied graph data structure problems: the \emph{distance sensitivity oracle} problem.

In this problem, we must build a data structure on a graph $G=(V,E)$ so that, given any $\qea,\qeb\in V$ and edge set $F$ of size $|F|\le f$, one can quickly approximate the distance between $\qea$ and $\qeb$ in $G\setminus F$. Observe that an oracle is a strict relaxation of a labeling scheme because, given the input $\{\qea,\qeb\}\cup F$, the algorithm can \emph{adaptively} read any part of the data structure, instead of non-adaptively reading only the labels of $\{\qea,\qeb\}\cup F$. Additionally, a labeling scheme must \emph{evenly distribute} useful information among vertices and edges, whereas an oracle faces no such restriction.

\paragraph{Previous Oracles: Few Faults or Inefficient.}

Despite extensive research, almost all previous oracles either can handle at most $f=O(\log n)$ faults or require super-linear query time or super-polynomial space. Specifically, most exact oracles \cite{demetrescu2003oracles,demetrescu2008oracles,bernstein2008improved,bernstein2009nearly,duan2009dual,williams2011faster,baswana2013approximate,DuanZ17b,chechik2020distance,gu2021constructing,bilo2021near,ren2022improved,dey2024near} handle only $f\le2$ faults. More recently, oracles for multiple faults have been discovered \cite{duan2022maintaining,dey2024nearly}%
, but their query time is super-exponential in $f$ and, hence, super-linear when $f=\omega(\log n)$. The query time of \cite{weimann2013replacement,BrandS19,karczmarz2023sensitivity} is also always super-linear. Even when approximation is allowed, every oracle from \cite{chechik20171+,bilo2023compact,bilo2023approximate,bilo2024improved} requires exponential-in-$f$ space and, thus, super-polynomial when $f=\omega(\log n)$. An approach based on fault-tolerant spanners \cite{levcopoulos1998efficient,chechik2009fault,bodwin2021optimal,bodwin2022partially} also requires directly computing distance on top of the spanners, taking super-linear query time.

\paragraph{The Exceptions.}

To our knowledge, only two known oracles do not suffer from the above limitations. The first oracle, by Chechik, Langberg, Peleg, and Roditty~\cite{ChechikLPR12}, has $(8k-2)(f+1)=O(kf)$ approximation, $\tilde{O}(fkn^{1+1/k})$ space, and $\tilde{O}(f)$ query time for any $k\ge1$. However, its approximation is still linear-in-$f$. 

The second oracle follows from the recent dynamic distance oracle with worst-case update time by Haeupler, Long, and Saranurak \cite{haeupler2024dynamic}. It provides an oracle with $2^{\poly(k)}$ approximation, $O(mn^{1/k})$ space, and $O(fn^{1/k})$ query time for any $k\le\log^{o(1)}n$.

\Cref{thm:main label} implies a new oracle that handles any number of faults and significantly improves upon both \cite{ChechikLPR12,haeupler2024dynamic}.

\begin{cor}

\label{thm:main oracle} For every $k \ge 1$, there is a deterministic distance sensitivity oracle for undirected $n$-vertex graphs with edge lengths in $[1, \poly(n)]$ undergoing $f$ edge faults, with $O(k^{4})$ approximation, $\tilde{O}(f^{4}n^{1+1/k})$ space, and $\tilde{O}(f^{5}n^{1/k})$ query time. The oracle can be constructed in polynomial time.

\end{cor}

The important qualitative improvement of \Cref{thm:main oracle} upon \cite{ChechikLPR12,haeupler2024dynamic} is at the approximation guarantee.
Our approximation exponentially improves upon \cite{haeupler2024dynamic} from $2^{\poly(k)}$ to $\poly(k)$ and does not degrade as the number of faults increases, unlike in \cite{ChechikLPR12}. 

However, \Cref{thm:main oracle} pays a $\poly(f)$-increase in space compared to \cite{ChechikLPR12,haeupler2024dynamic} and $n^{1 / k}$ in query time compared to \cite{ChechikLPR12}. %

In fact, we can significantly improve the query time to  $O(k + \log f)$ in the following natural \emph{two-stage} setting.
First, $F$ is given, and we update the data structure in  $\poly(f) \cdot n^{1 / k}$  time. Then, vertices $p,q$ are given, and we can approximate the distance between $p$ and $q$ in $O(k + \log f)$ time. But our approximation  slightly degrades to $O(k^5)$.
See \Cref{thm:oracle} for the more detailed statement of \Cref{thm:main oracle}, and \Cref{thm:oracle-fastquery} for the statement of the two-stage sensitivity oracle.

\subsection{Our Technique}

We bridge the theory of \emph{length-constrained expanders} \cite{OrigLCED22,ImprovedLCED24,haeupler2024low,haeupler2024dynamic} to the area of labeling schemes. Our technique strengthens two recent developments as follows:

\begin{enumerate}

\item We lift the \emph{centralized} technique from the dynamic distance oracle of \cite{haeupler2024dynamic}, which also uses length-constrained expanders, to the \emph{distributed} setting of labeling schemes. Additionally, we improve their approximation ratio from $2^{\poly(k)}$ to $\poly(k)$.

\item We strengthen the technique from the connectivity labeling scheme of \cite{long2025connectivity}, which uses \emph{standard expanders}, to work with length-constrained expanders, enabling distance computation instead of just connectivity.

\end{enumerate}

We briefly explain why the techniques in \cite{haeupler2024dynamic,long2025connectivity} are insufficient for our result. Firstly, we must avoid the two key primitives in \cite{haeupler2024dynamic}, the \emph{dynamic router} and \emph{local flow} subroutines, as they are highly centralized. Furthermore, \cite{haeupler2024dynamic} inherently pays at least a constant factor in approximation for each of the $k$ levels of their data structure, resulting in the $2^{\poly(k)}$ approximation.\footnote{They also pay $2^{\poly(k)}$ due to the close-to-linear time algorithm \cite{ImprovedLCED24} for length-constrained expander decomposition. We easily bypass this since we allow polynomial preprocessing time.} Lastly, length-constrained expanders are inherently more complex than the standard expanders used in \cite{long2025connectivity}. In particular, it is unclear how to define a single spanning tree based on the expander hierarchy as in \cite{long2025connectivity}, and length-constrained cuts are \emph{fractional} cuts, which introduces complications.

The key step to overcome these obstacles is by introducing a stronger expander hierarchy called the \emph{nested length-constrained expander hierarchy}.

\paragraph{Expander Hierarchies.}

Roughly speaking, an expander hierarchy is a sequence $\{(A_{1},C_{1})\}_{i=1}^d$ of node-weightings and cuts where $A_{i-1}$ is \emph{expanding} in $G-C_{i}$ and $\deg_{C_{i}} \preceq A_{i}$ for all $i$. An expander hierarchy is \emph{nested} if $A_{i} \preceq A_{i-1}$ for each $i$. Most expander hierarchies in the literature \cite{Racke02,bienkowski2003practical,harrelson2003polynomial,goranci2021expander} are nested, as this seems essential for strong applications such as oblivious routing and tree flow sparsifiers.\footnote{The hierarchies in \cite{racke2014computing,li2025congestion} that are not quite nested imply tree flow sparsifiers. The hierarchies of \cite{PatrascuT07,long2025connectivity} are not nested but are still useful for connectivity oracles and labeling schemes.}

Recently, Haeupler el al.~\cite{haeupler2024low} constructed a \emph{length-constrained} expander hierarchy, where the notion of expansion is replaced by length-constrained expansion. They used it to create a low-step flow emulator, leading to faster multi-commodity-flow algorithms. However, their hierarchy is not nested and appears insufficient for our purposes.

\paragraph{Our Technical Contribution: Nested Length-Constrained Expander Hierarchies.}

Our technical contribution is two-fold. First, we show how to utilize a \emph{nested} length-constrained expander hierarchy to obtain our distance labeling scheme, which we will outline in \Cref{sec:overview}. Second, we present an approach to construct a nested expander hierarchy that generalizes to the length-constrained setting. Indeed, previous approaches \cite{Racke02,bienkowski2003practical,harrelson2003polynomial,racke2014computing,goranci2021expander,li2025congestion} do not seem to generalize to length-constrained expansion. In contrast, our approach is simple, generic, and can extend to other notions of expansion, such as directed expansion.

We believe that both the concept of a nested expander hierarchy in the length-constrained setting and our generic construction will likely find broader applications.

\subsection{Other Related Results}

Exact distance labeling schemes have been extensively studied in special graph classes, including trees \cite{Peleg00b,AlstrupBR05,AlstrupGHP16a}, planar graphs \cite{gavoille2001approximate,gavoille2004distance,Tho04,GawrychowskiU23}, and more \cite{katz2000distance,gavoille2003distance,gavoille2003optimal,bazzaro2005distance,gavoille2005distance}.

In planar graphs \cite{AbrahamCG12} and graphs with bounded doubling dimension \cite{AbrahamCGP16}, Abraham~et~al.~gave near-optimal \emph{fault-tolerant} distance labeling schemes with $(1+\epsilon)$-approximation and $\poly(\log(n)/\epsilon)$ label size. In the single fault setting where $f=1$, there are labeling schemes for $(1+\epsilon)$-approximation of \emph{single-source} distances \cite{BaswanaCHR20}, exact distances in planar graphs \cite{Bar-NatanCGMW22}, and $(1+\epsilon)$-approximation of distance in directed planar graphs \cite{boneh2025optimal}.

\section{Overview} \label{sec:overview}

In this section, we will provide an overview of our fault-tolerant approximate distance labeling scheme. However, for ease of presentation, our discussion is mostly under the oracle setting, where the data structure is centralized and can be accessed without restrictions. We note that considering the simpler oracle setting does not substantially simplify our algorithm, and thus it is sufficient to illustrate the main intuition of our approach.

In fact, regardless of whether we work under the oracle or labeling scheme setting, our high-level strategy is essentially the same: we aim to limit the number of vertices that need to be touched/accessed when answering a query. More concretely, given a distance query between vertices $p$ and $q$ under edge failures $F$, we hope to construct a small graph $\disc$ called the \emph{discovered graph} (which includes the touched original vertices and edges, possibly augmented with additional artificial ones), such that $\disc$ approximately preserves the distance between $p$ and $q$ after failures, i.e., $\dist_{\disc}(p,q)\approx \dist_{G\setminus F}(p,q)$, and thus the query can be answered by running a shortest path algorithm on $\disc$. To achieve a fast query time in the oracle setting, it is necessary that the discovered graph $\disc$ has a small size, bounded by $\poly(f,n^{1/k})$. Moreover, we construct $\disc$ by extracting local information from vertices $p,q$ and the failed edges $F$, which makes it possible to store the precomputed data structure in a distributed manner as a labeling scheme.

Our approach to constructing the discovered graph is \emph{expander-based}, following a line of work on expander-based oracles/labeling schemes for edge- and vertex-fault-tolerant connectivity problems (e.g. \cite{PatrascuT07,long2025connectivity}). However, in order to solve the distance problem, we leverage the more powerful notion of \emph{length-constrained expanders}, which enables us to further capture distance-related information. We point out that, when assuming the input graph is a good expander, all the aforementioned expander-based algorithms (including ours) collapse into a simple and unified framework (although the previous algorithms for connectivity do not need to construct the discovered graph explicitly). We will discuss this in \Cref{sect:OverviewExpanderCase} for better understanding, and it already demonstrates the power of length-constrained expanders in solving fault-tolerant distance problems.

However, extending the simple algorithm from expanders to general graphs entails not only significant technical complications but also deeper fundamental challenges, resulting in an algorithm that is very different from the counterparts for connectivity \cite{PatrascuT07,long2025connectivity}. Interestingly, such challenges become easier to identify from a labeling-scheme perspective, and this is also the reason why we obtain a labeling-scheme result first. To overcome these challenges, one of our technical contributions is the introduction of \emph{nested} length-constrained expander hierarchy, which is crucial for achieving a $\poly(k)$ vs. $\poly(f,n^{1/k})$ tradeoff between approximation and label size. We will elaborate on this in \Cref{sect:OverviewGeneral}.

Lastly, we note that we compare our approach with: (1) previous $O(kf)$-approximation fault-tolerant distance algorithms \cite{ChechikLPR12,DoryP21} in \Cref{sect:OverviewExpanderCase}, and (2) previous expander-based fault-tolerant connectivity algorithms \cite{PatrascuT07,long2025connectivity} at the end of \Cref{sect:OverviewGeneral}.

\subsection{The Expander Case}
\label{sect:OverviewExpanderCase}

In this warm-up case, we assume the input graph $G = (V,E)$ is a length-constrained expander. The precise definition of a length-constrained expander is not very important, and we will only exploit one of its key properties, which we will mention soon. Also, for simplicity, we only consider the single-distance decision version of the problem. That is, we additionally receive a target distance $h$ at the beginning, and for each query $\langle p,q, F\rangle$, we want to correctly certify either (FAR) $\dist_{G\setminus F}(p,q) > h$ or (CLOSE) $\dist_{G\setminus F}(p,q)\leq h \cdot \alpha$ for some approximation factor $\alpha \geq 1$. To solve the original problem, it is sufficient to check different length scales $h$ (for example, taking different powers of two as $h$). We first focus on the oracle setting, and will slightly discuss the labeling scheme setting in the last paragraph.

\paragraph{Previous $O(kf)$-Approximate Oracles.} To motivate our algorithm, we first quickly outline the previous $O(kf)$-approximate fault-tolerant distance oracles \cite{ChechikLPR12}, and see where the linear-in-$f$ dependency comes from. The oracle is constructed in two steps.

\medskip

\noindent\underline{Step 1.} First, we construct a \emph{sparse neighborhood cover} ${\cal N}$ of $G$ with covering radius $\hcov = h$, diameter $\hdiam = kh$, and width $\omega = \tilde{O}(n^{1/k})$ \cite{NHCover98}. 

Recall that such a neighborhood cover is a collection of $\omega$ many \emph{clusterings} ${\cal S}\in{\cal N}$, each of which is a collection of disjoint \emph{clusters} $S\subseteq V$. The covering radius is $\hcov$, meaning that for each vertex $v\in V$, its radius-$\hcov$ neighborhood $B(v,h_{\cov}) = \{v'\mid \dist_{G}(v,v')\leq h_{\cov}\}$ is contained by some cluster $S\in{\cal S}\in{\cal N}$ (we write $S\in{\cal N}$ for short). The (strong) diameter is $h_{\diam}$, meaning that for each cluster $S\in{\cal N}$, the induced subgraph $G[S]$ has diameter $h_{\diam}$.

\medskip

\noindent\underline{Step 2.} For each cluster $S\in{\cal N}$, construct a \emph{fault-tolerant connectivity oracle} $\Conn(S)$ for the subgraph $G[S]$. To simplify the analysis, we use a tree-based fault-tolerant connectivity oracle here (see, e.g., \cite{DuanP20}). The connectivity oracle $\Conn(S)$ fixes a shortest path tree $T_{S}$ of $G[S]$ (rooted at an arbitrary vertex) and precomputes certain additional data structures that we omit here. Whenever a connectivity query $\langle p,q, F\rangle$ comes, the failed tree edges $F\cap T_{S}$ will break the tree $T_{S}$ into at most $f+1$ components, and the connectivity oracle simply connects these components via the non-failed, non-tree edges (with the aid of the additional data structures).

\begin{figure}
    \centering
    \begin{tikzpicture}[main/.style={draw, circle, minimum size=0.2cm}, minor/.style={draw, fill=black, circle, minimum size=0.1cm, inner sep=0pt}, edges/.style={draw}, light/.style = {fill=none}, heavy/.style = {fill=slightlylightgray}, txt/.style = {draw=none}, every loop/.style={}]

        \node[ellipse, draw, color=blue!30, line width=0.02cm, minimum width=12cm, minimum height=7cm] (S) at (0, 0) {};

        \node[ellipse, draw, dotted, color=blue!20, line width=0.04cm, minimum width=3.5cm, minimum height=5.3cm] (G1) at (-3.6, 0) {};
        \node[ellipse, draw, dotted, color=blue!20, line width=0.04cm, minimum width=3.5cm, minimum height=7cm] (G2) at (0, 0) {};
        \node[ellipse, draw, dotted, color=blue!20, line width=0.04cm, minimum width=2.9cm, minimum height=5.7cm] (G3) at (3.3, 0) {};

\node[main] (v0) at (3.40, -2.38) {};
\node[main] (v1) at (-0.23, -2.5) {};
\node[main] (v2) at (-1.2, -1.7) {};
\node[main] (v3) at (1.18, -1.73) {};
\node[main] (v4) at (-4.30, -1.62) {};
\node[main] (v5) at (-2.73, -1.38) {};
\node[main] (v6) at (2.35, -0.97) {};
\node[main] (v7) at (-1.18, -0.53) {};
\node[main] (v8) at (-2.80, -0.25) {};
\node[main] (v9) at (4.00, -0.15) {};
\node[main] (v10) at (0.00, 0.00) {};
\node[main] (v11) at (-4.30, 0.23) {};
\node[main] (v12) at (2.52, 0.28) {};
\node[main] (v13) at (-2.17, 0.72) {};
\node[main] (v14) at (-3.25, 1.25) {};
\node[main] (v15) at (-0.57, 1.30) {};
\node[main] (v16) at (3.05, 1.95) {};
\node[main] (v17) at (0.88, 2.02) {};
\node[main] (v18) at (-4, 2.2) {};
\node[main] (v19) at (-0.97, 2.48) {};

        \draw[-, edges] (v18) to (v14);
        \draw[-, edges] (v14) to (v11);
        \draw[-, edges] (v14) to (v13);
        \draw[-, edges] (v13) to (v8);
        \draw[-, edges] (v8) to (v5);
        \draw[-, edges] (v5) to (v4);
        \draw[-, edges] (v7) to (v10);
        \draw[-, edges] (v10) to (v15);
        \draw[-, edges] (v15) to (v17);
        \draw[-, edges] (v19) to (v17);
        \draw[-, edges] (v12) to (v16);
        \draw[-, edges] (v12) to (v6);
        \draw[-, edges] (v6) to (v9);
        \draw[-, edges] (v6) to (v0);
        \draw[-, edges] (v10) to (v3);
        \draw[-, edges] (v2) to (v1);
        \draw[-, edges] (v1) to (v3);

        \draw[-, color=red, edges] (v8) to (v7);
        \draw[-, color=red, edges] (v10) to (v12);
        
        \draw[-, dotted, edges] (v5) to (v2);
        \draw[-, dotted, edges] (v3) to (v6);
        \draw[-, dotted, edges] (v9) to (v16);
        \draw[-, dotted, edges] (v16) to (v17);
        \draw[-, dotted, edges] (v11) to (v18);
        \draw[-, dotted, edges] (v13) to (v19);
        
    \end{tikzpicture}
    \caption{A cluster $S$ and its tree $T_S$ (solid tree edges vs. dotted non-tree edges). The two red tree edges failing breaks the cluster into the three very light blue, dotted components.} %
    \label{fig:cluster-components-example}
\end{figure}
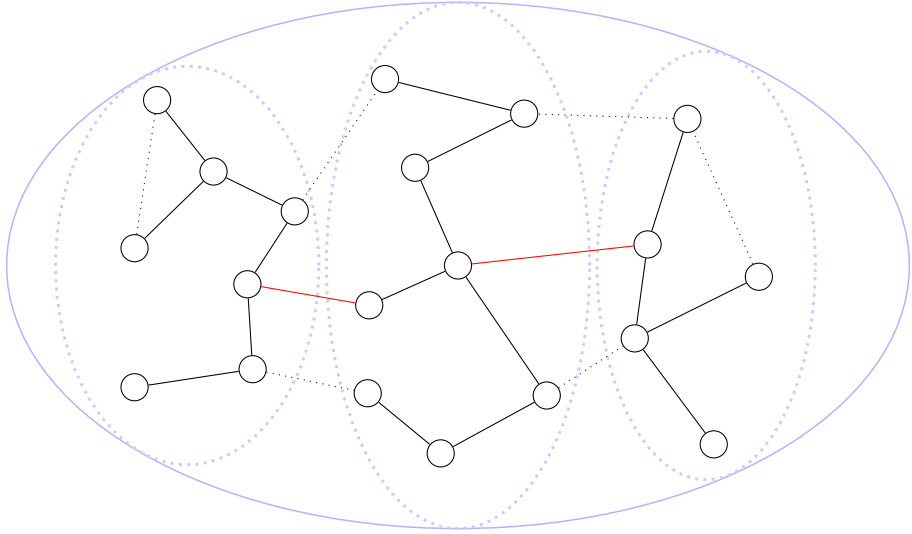

\medskip

To answer a distance query $\langle p, q, F\rangle$, we look at the cluster $S_{p}$ containing the neighborhood $B(p,h_{\cov})$. If its corresponding connectivity oracle $\Conn(S_{p})$ indicates that $p$ and $q$ are connected in $G[S_{p}]\setminus F$, return (CLOSE), otherwise return (FAR). 

The correctness of the output (FAR) is easy to see: if $\dist_{G\setminus F}(p,q)\leq h$, then the underlying shortest path must fall inside $G[S_{p}]$ (this is guaranteed by the covering radius), and thus $p$ and $q$ must be connected in $G[S_{p}]\setminus F$. The approximation $O(kf)$ comes from the analysis of the output (CLOSE): when $p$ and $q$ are connected in $G[S_{p}]\setminus F$, we can only certify $\dist_{G\setminus F}(p,q)\leq O(kfh)$, since the underlying path could go through all the components in $T_{S}\setminus F$ (there could be $f+1$ of them), each of which has a diameter $2h_{\diam} = 2kh$.

In other words, the above approach has an approximation factor linear in $f$, since it is essentially still solving a connectivity problem, which loses an additive factor of $2kh$ at each component.

\paragraph{The Key Property of Length-Constrained Expanders.} Therefore, to get rid of the $f$ factor in the approximation, we will exploit the following key property in \Cref{lemma:KeyProperty}, which holds when the input graph $G$ is a length-constrained expander. For a vertex subset $\Gamma\subseteq V$, we let $\deg_{G}(\Gamma)$ denote the total degree of vertices in $\Gamma$.

\begin{lemma}[The key property]
\label{lemma:KeyProperty} 
When $G$ is a length-constrained expander, there is a value $\theavy = \tilde{O}(f\cdot n^{1/k})$ satisfying the following. For any two vertex subsets $\Gamma,\Gamma'\subseteq V$, if they satisfy that
\begin{enumerate}
\item $\deg_{G}(\Gamma),\deg_{G}(\Gamma')\geq \theavy$, and
\item $\dist_{G}(u,u')\leq h_{\diam}$ for all pairs $u,u'\in \Gamma\cup\Gamma'$,
\end{enumerate}
then regardless of the set $F$ of up to $f$ edge failures, $\dist_{G \setminus F}(\clu, \clu') \leq h_{\diam}\cdot k$, i.e. there exists a path of length at most $h_{\diam}\cdot k$ in $G \setminus F$ from some $u \in \clu$ to $u' \in \clu'$.
\end{lemma}

We now give a high-level explanation of how exploiting \Cref{lemma:KeyProperty} allows us to avoid an additive loss at each component. %
For this, classify the components of $T_{S_{p}}\setminus F$ into \emph{heavy components}  and \emph{light components}, where a component $\Gamma$ is heavy if $\deg_{G}(\Gamma)\geq \theavy$ and light otherwise. Then, we have the following win-win scenario:

For each light component $\Gamma$, it is affordable to touch \emph{all} its incident edges in the query phase (note that there are $\deg_{G}(\Gamma)<\theavy = \tilde{O}(f\cdot n^{1/k})$ incident edges), so intuitively there is no additive loss at all when going through a light component.

For the heavy components, \Cref{lemma:KeyProperty} shows that \emph{any pair} of heavy components $\Gamma$ and $\Gamma'$ have $\dist_{G\setminus F}(\Gamma,\Gamma')\leq h_{\diam}\cdot k$; the components $\Gamma$ and $\Gamma'$ satisfy condition 1 by definition, and satisfy condition 2 as they are contained in the same cluster $S_{p}$ of diameter $h_{\diam}$. Therefore, intuitively, when we want to find an approximate shortest path after failures, there is no need to go through more than 2 heavy components: we can just jump from the first heavy component to the last heavy component by paying an extra additive loss of $h_{\diam}\cdot k$. Each of the two heavy components still incurs an additive loss of $2hk$ to traverse within, so the total additive loss is $4hk + h_{\diam}\cdot k = O(hk^{2})$.

In summary, we expect a total additive loss of $O(hk^{2})$ %
, and should thus be able to achieve an $O(k^{2})$-approximation. %

\paragraph{Our Oracle.} We are ready to provide a formal description of our $O(k^{2})$-approximate fault-tolerant distance oracle (for length-constrained expander graphs). 

\medskip

\noindent\underline{Proprocessing.} The construction of the oracle is basically the same, which includes the neighborhood cover ${\cal N}$ and, for each cluster $S\in{\cal N}$, the shortest path tree $T_{S}$. 

\medskip

\noindent\underline{The Query Algorithm.} The main difference lies in the query algorithm. Denote again by $S_{p}\in{\cal N}$ the cluster containing the neighborhood $B(p,h_{\cov})$ of the vertex $p$. We first define the discovered graph $\disc$. The vertices $V(\disc)$ of $\disc$ include (1) for each light component $\Gamma$ in $T_{S_{p}}\setminus F$, the original vertices inside or adjacent to $\Gamma$, i.e. $\Gamma \cup N(\Gamma)$, (2) for each heavy component $\Gamma$, an artificial vertex called \emph{a component vertex}, denoted by $\pi(\Gamma)$, and (3) the two query endpoint vertices $p$ and $q$.
There are three types of edges in $\disc$.
\begin{enumerate}
\item (original edges) First, we add all the original non-failed edges incident to light components.
\item (between original vertices and heavy components) For each original vertex $v\in V(\disc)\cap V(G)$, if the component $\Gamma_{v}$ containing $v$ is heavy, add an artificial edge between $v$ and $\pi(\Gamma_{v})$ with length $2h_{\diam} = 2hk$. Note that $2h_{\diam}$ upper bounds the diameter of components.
\item (between heavy components) For each pair of heavy components $\Gamma$ and $\Gamma'$, add an artificial edge between $\pi(\Gamma)$ and $\pi(\Gamma')$ with length $h_{\diam} \cdot k + 4 \hdiam = hk^2 + 4hk$. Note that $\hdiam \cdot k$ upper bounds the distance between any two heavy components in $G \setminus F$, and $h_{\diam} \cdot k + 4 \hdiam$ the distance between their furthest two vertices. %
\end{enumerate}
Providing the discovered graph $\disc$, we return (CLOSE) if $\dist_{\disc}(p,q)\leq 10hk^{2}$; otherwise, we return (FAR). Note that we can compute $\dist_{\disc}(p,q)$ by running any exact shortest path algorithm on $\disc$.

\medskip

\noindent\underline{Correctness.} The correctness of the output (CLOSE) is trivial since the discovered graph $\disc$ will never underestimate the distances between vertices in $V \cap V(\disc)$. Thus, we focus on showing the correctness of the output (FAR). Concretely, we want to prove that, if $\dist_{G\setminus F}(p,q)\leq h$, then $\dist_{\disc}(p,q)\leq 10hk^{2}$ and the algorithm will not return (FAR).

Let $P$ be the shortest path in $G\setminus F$ between $p$ and $q$, which has length at most $h$. Recall that $P$ must be entirely inside the cluster $S_{p}$ by the definition of covering radius. If every $P$-vertex is contained by a light component, then $P$ even appears in $\disc$ because of type-1 edges. Otherwise, let $w_{p}$ and $w_{q}$ be the first and last $P$-vertices falling in heavy components $\Gamma_{w_{p}}$ and $\Gamma_{w_{q}}$ respectively. Note that all $P$-edges before $w_{p}$ or after $w_{q}$ are incident to light components, and thus they appear in $\disc$ because of type-1 edges. Furthermore, observe that $w_{p}$ (resp. $w_{q}$) is connected to $\pi(\Gamma_{w_{p}})$ (resp. $\pi(\Gamma_{w_{q}})$) via a type-2 edge of length $2hk$, and $\pi(\Gamma_{w_{p}})$ and $\pi(\Gamma_{w_{q}})$ are connected via a type-3 edge of length $hk^{2} + 4hk$. In summary, we have
\[
\dist_{\disc}(p,q)\leq \leng_{G\setminus F}(P) + 4hk + hk^{2} + 4hk \leq h + 8hk + hk^{2}\leq 10hk^{2}
\]
as desired.

\medskip

\noindent{\underline{Query Time.}} We can easily observe that query time is nearly linear in the size of the discovered graph $\disc$, and the size of $\disc$ is proportional to the number of type-1 edges, i.e., the number of edges incident to light components. Since each light component has at most $\theavy = \tilde{O}(f\cdot n^{1/k})$ incident edges, and there are at most $f+1$ light components, the final bound is $\tilde{O}(f^{2}\cdot n^{1/k})$.

\paragraph{The Labeling Scheme.} Lastly, we briefly discuss how to transform our oracle into a labeling scheme. By our discussion above, it is sufficient if, by accessing the labels of $F$, we can extract the following information for each cluster $S$ and each component $\Gamma$ of $T_{S}\setminus F$: its degree $\deg_{G}(\Gamma)$, and all its incident edges when $\deg_{G}(\Gamma)\leq \theavy$. First, note that we can design labels for each cluster $S$ independently, because each vertex appears in at most $\omega = \tilde{O}(n^{1/k})$ clusters (by the definition of the width $\omega$) which only incurs a multiplicative $\tilde{O}(n^{1/k})$ factor to the final label size.

Fixing a cluster $S$, a labeling scheme that can recover the above information by accessing the labels of $T_{S}\setminus F$ has already been shown in \cite{long2025connectivity}. The high-level idea is to consider the \emph{Euler tour order} of $T_{S}$, and exploit that after failures, components will correspond to unions of intervals of the order, where each such interval is preceded and succeeded by a failed edge. This makes everything straightforward: (1) we can recover the total degree of an interval if the two failed edges delimiting it store their prefix degree-sums in the order, and (2) we can detect up to $\theavy$ of the incident edges to an interval if the two failed edges store up to $\theavy$ of the nearest edges in the order.

\subsection{The General Setting}
\label{sect:OverviewGeneral}

In this subsection, we consider a general input graph $G = (V,E)$. For simplicity, we still consider the single-distance decision version with a target distance $h$. Furthermore, we note that although our goal in this overview is to present an oracle result, some of our discussion is from a labeling-scheme perspective, as this viewpoint can sometimes offer cleaner intuition. %

\paragraph{Length-Constrained Expander Hierarchies.} We want to extend the previous algorithm from expander graphs to general graphs. A standard approach of going about this that we too use is \textit{expander hierarchies}. We therefore start with the definition of length-constrained expander hierarchies along with some related notation.

First, a node weighting is simply a function $A : V \rightarrow \mathbb{R}_{\geq 0}$ assigning a weight $A(v)$ to each vertex $v \in V$. We write $A(\clu) := \sum_{v \in \clu} A(v)$ for the weight of a vertex set $\clu \subseteq V$. We write $A \preceq A'$ for two node weightings $A, A'$ when $A(v) \leq A'(v)$ for all $v$. For an edge set $C\subseteq E$, we write $\deg_{C}$ for the node weighting for which $\deg_{C}(v)$ equals the total number of $C$-edges incident to $v$. 

We introduce the concept that \emph{$A$ is $(h_{\ed},s_{\ed})$-length $\phi$-expanding in $\tilde{G}$}, for some graph $\tilde{G}$, node-weighting $A$ and parameters $h_{\ed},s_{\ed}$ and $\phi$. Again, the precise definition is not important here, and we only exploit the following key property in \Cref{lemma:GeneralKeyProperty}, a generalized version of \Cref{lemma:KeyProperty} \footnote{Observe that \Cref{lemma:KeyProperty} is a special case of \Cref{lemma:GeneralKeyProperty} when $\tilde{G} = G$, $A = \deg_{G}$, $h_{\ed} = h_{\diam}$, $s_{\ed} = k$ and $\phi = 1/\tilde{O}(n^{1/k})$}.

\begin{lemma}[The general key property]
\label{lemma:GeneralKeyProperty}
For a graph $\tilde{G} = (V(\tilde{G}),E(\tilde{G}))$ with a node-weighting $A$ and parameters $h_{\ed},s_{\ed},\phi$, if $A$ is $(h_{\ed},s_{\ed})$-length $\phi$-expanding in $\tilde{G}$, then there is a value $\theavy = \tilde{O}(f/\phi)$ satisfying the following: for any two vertex subsets $\Gamma,\Gamma'\subseteq V(\tilde{G})$ for which %
\begin{enumerate}
\item $A(\Gamma),A(\Gamma')\geq \theavy$, and
\item $\dist_{\tilde{G}}(u,u')\leq h_{\ed}$ for all pairs $u,u'\in \Gamma\cup\Gamma'$,
\end{enumerate}
then regardless of the set $F$ of up to $f$ edge failures, $\dist_{\tilde{G} \setminus F}(\clu, \clu') \leq h_{\ed}\cdot s_{\ed}$, i.e. there exists a path of length at most $h_{\ed}\cdot s_{\ed}$ in $\tilde{G} \setminus F$ from some $u \in \clu$ to $u' \in \clu'$.
\end{lemma}

Now, we are ready to define a length-constrained expander hierarchy. A $(h_{\ed},s_{\ed})$-length $\phi$-expander hierarchy with $d$ levels is a collection of node weightings $A_{i}$ and cuts $C_{i+1}\subseteq E$, denoted by ${\cal H} = \{A_{i},C_{i+1}\mid 0\leq i\leq d\}$, satisfying the following.
\begin{enumerate}
\item For each $0\leq i\leq d$, $A_{i}$ is $(h_{\ed},s_{\ed})$-length $\phi$-expanding in $G - C_{i + 1}$.
\item For each $1\leq i\leq d$, $A_{i} = \deg_{C_{i}}$. In the boundary cases, $A_{0} = \deg_{G}$ and $C_{d+1} = 0$.
\end{enumerate}
We point out that this is an informal definition. In fact, when defining concepts related to length-constrained expanders, the precise definitions always involve \emph{fractional cuts}, but in this subsection, we assume for simplicity that all the cuts $C_{i}$ are \emph{integral cuts} (i.e., each edge is either in $C_{i}$ or not in $C_{i}$). Furthermore, as mentioned earlier and as we will later see in this overview, a basic hierarchy is insufficient for our purposes, and we need to strengthen it to ensure a \emph{nestedness} property\footnote{We note that to ensure nestedness, we may need to relax property 2 by allowing $A_{0}\succeq \deg_{G}$ and $A_{i}\succeq \deg_{C_{i}}$ for each $1\leq i\leq d$. However, this relaxation will not hurt (so it can be ignored) since a larger $A_{i}$ gives a stronger \Cref{lemma:GeneralKeyProperty}.}.
\begin{enumerate}
\item[3.] (nestedness) for each $1\leq i\leq d$, $A_{i}\preceq A_{i-1}$.
\end{enumerate}
One of our main technical contributions is showing the existence of nested length-constrained expander hierarchies (given a graph $G$ and a length parameter $h_{\ed}$) with 
\[
\text{levels}~d = O(k),\qquad\text{length slack}~s_{\ed} = k\qquad\text{and expansion~} \phi = 1/\tilde{O}(n^{1/k}).
\] 
Similar to \Cref{sect:OverviewExpanderCase}, but with a slight difference, we set $h_{\ed} = 2h_{\diam}$. The reason of adding a constant factor $2$ will be clear from the discussion below.

\paragraph{A Top-Down Approach.} Our formal argument in the main body proceeds by a bottom-up induction, but conceptually, our approach is easier to understand in a top-down recursive manner. Indeed, at the top level of the hierarchy, we have that $A_{d}$ is expanding in $G$, which is quite similar to the expander case in \Cref{sect:OverviewExpanderCase}.

In the preprocessing phase, we first compute a hierarchy, and then define the sequence of graphs $G_d := G$ and $G_{j} := G_{j + 1} - C_{j + 1}$ for $0 \leq j < d$: $G_{j}$ is the graph with all higher-level cuts applied, thus (1) \Cref{lemma:GeneralKeyProperty} can be applied to two $A_j$-heavy clusters close to each other in $G_j$, as they are then at least as close in $G - C_{j + 1}$ as well, and (2) distances in $G_j$ are never shorter than distances in $G_{j + 1}$. 

We then construct for each level of the hierarchy a neighborhood cover $\mathcal{N}_j$ with covering radius $h_{\cov} = h$, diameter $h_{\diam} = hk$ and width $\omega = \tilde{O}(n^{1/k})$ on the level's graph $G_j$, and then fix the shortest path trees for all clusters $\{T_{S}\mid S\in \mathcal{N}_j\}$. %

Consider now the query phase of a query $\langle p,q, F\rangle$, and consider a cluster $S$ at the top level of the hierarchy. The cluster's tree $T_{S}$ may be broken into components because of edge failures, and we again classify the components into heavy ones and light ones. However, we now only have that $A_{d}$ (rather than $\deg_{G}$) is expanding in $G$, so naturally, we call a component $\Gamma$ of $T_{S}\setminus F$ \textit{$A_{d}$-heavy} if $A_{d}(\Gamma)\geq \theavy$, and \textit{$A_{d}$-light} otherwise. Analogously to the earlier win-win scenario, we now have the following.
\begin{itemize}
\item For each $A_{d}$-light component, there are at most $\theavy$ incident edges \emph{in $C_{d}$} (recall that $A_{d} = \deg_{C_{d}}$), and we \emph{discover} all of these edges and their endpoints (from a labeling-scheme perspective).
\item Any two $A_{d}$-heavy components $\Gamma,\Gamma'$ of $T_{S}\setminus F$ are close in $G \setminus F$, specifically we have $\dist_{G\setminus F}(\Gamma,\Gamma')\leq h_{\ed}\cdot s_{\ed}= 2h_{\diam}\cdot k =  2hk^{2}$ (by \Cref{lemma:GeneralKeyProperty}).
\end{itemize}

Furthermore, let us try to add vertices and edges to the discovered graph $\disc$ based on the information discovered at level $d$. Similar to \Cref{sect:OverviewExpanderCase}, we will include two types of vertices: (1) all discovered original vertices, and (2) a component vertex $\pi(\Gamma)$ for components in those clusters $S$ with $T_{S}$ intersecting $F$. Note that a cluster $S$ with $T_{S}$ disjoint from $F$ has a unique component that is $S$ itself, but we view this component as undiscovered for now, and thus we will not create a component vertex for it. We will add the following edges to $\disc$.
\begin{enumerate}
\item (original edges) First, we add all discovered edges.
\item (between original vertices and component vertices) For each discovered vertex $v$ and each component vertex $\pi(\Gamma)$, if $v\in \Gamma$, add an artificial edge between them with length $2h_{\diam}$.
\item (between heavy components) For two components $\Gamma,\Gamma'$, if they belong to the same cluster $S$ and both are $A_{d}$-heavy, add an artificial edge between $\pi(\Gamma)$ and $\pi(\Gamma')$ of length $h_{\ed}\cdot s_{\ed} + 4\hdiam = 2h_{\diam}\cdot k + 4\hdiam = 2 h k^2 + 4 h k$.
\end{enumerate}

Generally, for each level of the hierarchy, we will add similar vertices and edges to form each level of the discovered graph $\disc$. We point out that the discovery status of a vertex is global, meaning that once a vertex is discovered at some level, we will view it as a discovered vertex when constructing every level of $\disc$. So different levels of $\disc$ are connected by these discovered vertices (which we call \emph{waypoints} in the main body).
After constructing $\disc$, let us try to show the correctness and see what may go wrong.

\paragraph{An Attempt to Prove Correctness.} We want to prove that, if $\dist_{G\setminus F}(p,q)\leq h$, then there is a short path in $\disc$ between $p$ and $q$, ideally of length $\poly(k)\cdot h$. Fix a shortest path witness $P$ between $p$ and $q$ in $G\setminus F$. Again, we consider the cluster $S_{P}\in{\cal N}$ that entirely contains $P$.

\medskip

\noindent{\underline{A Simple Case.}} A relatively simple scenario is when all the $C_{d}$-edges on $P$ are discovered. Although $P$ will have edges not in $C_{d}$, $P$ does not entirely show up in the discovered graph. However, let us focus on a subpath $\check{P}$ of $P$ between two consecutive edges in $C_{d}$. The key observation here is that as $\check{P}$ contains no edges in the cut $C_{d}$, $\check{P}$ is a path in $G - C_{d}$, exactly the level-$(d-1)$ graph $G_{d - 1}$ in the hierarchy. More concretely, letting $w$ and $w'$ be the endpoints of $\check{P}$, we have the following recursive scenario: %
\begin{itemize}
\item (Recursive Scenario 1) We have two discovered vertices $w$ and $w'$, and there is a path $\check{P}$ between them in $G_{d - 1}$, such that $\check{P}$ is a subpath of the original witness path $P$.

The goal for this recursive scenario is that the distance between $w$ and $w'$ in lower levels of $\disc$ is at most $\poly(k)\cdot \leng(\check{P})$.
\end{itemize}
Note that this recursive scenario is the same as our initial situation but one level lower in the hierarchy, as $\qea$ and $\qeb$ are discovered vertices (as we received their vertex labels) between which there is a path $P$ in $G_d$ ($= G$) that is a subpath of the original witnessed path (in fact equal to the original witnessed path).

If we assume now that the goal of each recursive scenario is achieved, then, by concatenating the $C_{j}$-edges on $P$ (which appear in $\disc$ due to being discovered) and the lower-level shortest paths in $\disc$ (returned by the recursive scenarios), we find a $p$-$q$ path in $\disc$ of length $\poly(k)\cdot h$ as desired.

We make some remarks on the argument above. In this recursive scenario, it is quite crucial that $w$ and $w'$
are discovered, so they appear in $\disc$, and thus the recursive problem is well-defined. Later, we will encounter a harder recursive scenario that involves undiscovered endpoints. Furthermore, we point out that, this recursive problem should be with respect to the length scale around $\leng(\check{P})$ (instead of $h$) to obtain a $\poly(k)$-approximation of $\leng(\check{P})$. So strictly speaking, $\check{P}$ should be a subpath disjoint from the cut $C_{d}$ of \emph{another hierarchy}, built with regard to the length scale around $\leng(\check{P})$. A possible fix is to define the node-weighting $A_{d}$ as the sum of $C_{d}$ across hierarchies of all length scales (doing so will even ensure that $\check{P}$ is disjoint from all $C_{d}$). However, this definition of $A_{d}$ entangles hierarchies of different length scales, in particular causing issues due to the fractionality of length-constrained cuts, a detail avoided for simplicity in this section. Therefore, in the main body, we adopt an alternative approach that keeps the hierarchies fully separate, resulting in an actual algorithm that slightly differs from this overview. We omit the details and ignore the subtlety of transferring between length scales in the discussion below.

\begin{figure}
    \centering
    \begin{tikzpicture}[minor/.style={draw, fill=black, circle, minimum size=0.1cm, inner sep=0pt}, nodraw/.style={}, edges/.style={draw}, txt/.style = {draw=none}, every loop/.style={}]

        \node[minor, label=below:{$\qea$}] (p) at (-4.5, -1) {};
        \coordinate (wp1) at (-3, -0.5) {};
        \coordinate (w2) at (-2.5, 0) {};
        \coordinate (wp2) at (-1, 1) {};
        \coordinate (w3) at (-0.5, 0.5) {};
        \coordinate (wp3) at (0.5, -0.5) {};
        \coordinate (w4) at (1, -1) {};
        \coordinate (wp4) at (2.5, -1.5) {};
        \coordinate (w5) at (3, -1) {};
        \node[minor, label={$\qeb$}] (q) at (4, -0.5) {};

        \node[ellipse, draw, color=blue!30, line width=0.02cm, minimum width=12cm, minimum height=7cm, label=left:{\color{blue}$S_P$}] (S) at (0, 0) {};

        \node[ellipse, draw, dotted, color=blue!30, line width=0.05cm, minimum width=3.5cm, minimum height=5.3cm] (G1) at (-3.6, 0) {};
        \node[ellipse, draw, dotted, color=blue!30, line width=0.05cm, minimum width=3.5cm, minimum height=7cm] (G2) at (0, 0) {};
        \node[ellipse, draw, dotted, color=blue!30, line width=0.05cm, minimum width=2.9cm, minimum height=5.7cm] (G3) at (3.3, 0) {};

        \node[txt] (GT1) at (-3.6, 2) {\color{blue}$\Gamma_1$};
        \node[txt] (GT2) at (0, 2.5) {\color{blue}$\Gamma_2$};
        \node[txt] (GT3) at (3.3, 2) {\color{blue}$\Gamma_3$};
        
        \draw[-, edges, color=blue] (p) to[out=0,in=-135] node[left, above, yshift=0.1cm]{\color{blue} $\check{P}_1$} (wp1);
        \draw[-, edges, line width=0.1cm] (wp1) to (w2);
        \draw[-, edges, color=blue] (w2) to[out=45,in=135] node[left, above, yshift=0.2cm]{\color{blue} $\check{P}_2$} (wp2);
        \draw[-, edges, line width=0.1cm] (wp2) to (w3);
        \draw[-, edges, color=blue] (w3) to[out=-45,in=135] node[left, above, xshift=0.2cm]{\color{blue} $\check{P}_3$} (wp3);
        \draw[-, edges, line width=0.1cm] (wp3) to (w4);
        \draw[-, edges, color=blue] (w4) to[out=-45,in=225] node[left, below, yshift=-0.1cm]{\color{blue} $\check{P}_4$} (wp4);
        \draw[-, edges, line width=0.1cm] (wp4) to (w5);
        \draw[-, edges, color=blue] (w5) to[out=45,in=180] node[left, above, yshift=0.1cm, xshift=-0.1cm]{\color{blue} $\check{P}_5$} (q);
        
    \end{tikzpicture}
    \caption{A witness shortest $(\qea, \qeb)$-path $P$ in $G \setminus F$, contained in a cluster $S_P$ of $\mathcal{N}_d$ with components $\Gamma_1, \Gamma_2$, and $\Gamma_3$. The $C_d$-edges on $P$ are drawn thick, and the subpaths between the cut edges are labeled $\check{P}_1$ to $\check{P}_5$. \\
    Suppose that all of the components $\Gamma_1$, $\Gamma_2$ and $\Gamma_3$ are $A_d$-light. Then, each of the $C_d$-edges on $P$ is discovered, and each of the subpaths $\check{P}_1$ to $\check{P}_5$ corresponds to a recursive scenario 1.}
    \label{fig:rec-scenario-example}
\end{figure}
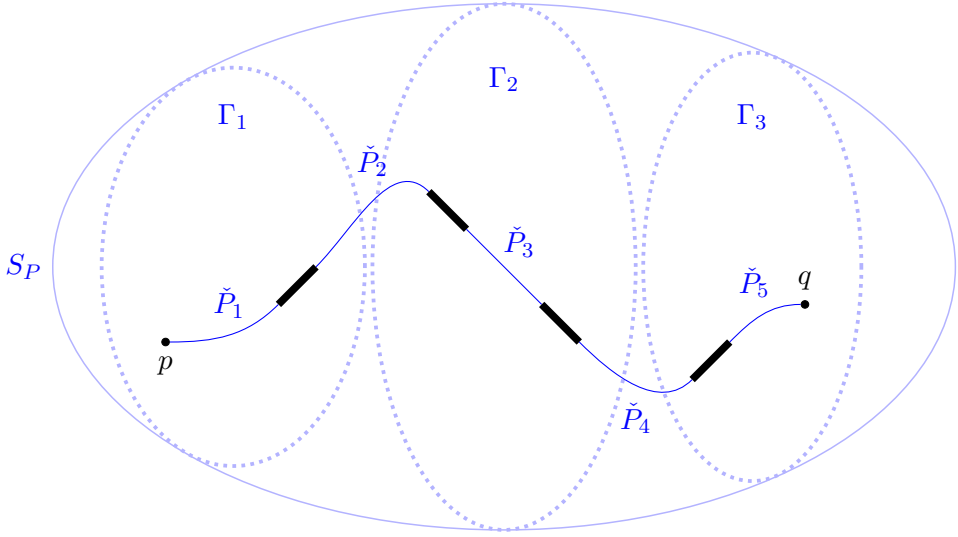

\medskip

\noindent\underline{The General Case.} In general, the witnessed path $P$ may include both discovered and undiscovered edges of the cut $C_{d}$. Recall that an undiscovered $C_{d}$-edge on $P$ will have both its endpoints in $A_{d}$-heavy components of the cluster $S_{P}$. Similar to the idea in \Cref{sect:OverviewExpanderCase}, we want to jump from the first/leftmost undiscovered $C_{d}$-edge on the path to the last/rightmost one by exploiting \Cref{lemma:GeneralKeyProperty}. %

Concretely, let $v$ be the left endpoint of the leftmost undiscovered $C_d$ -edge on $P$, and let $v'$ symmetrically be the right endpoint of the rightmost undiscovered $C_d$ -edge on $P$ (so that the path now consists of subpaths from $\qea$ to $v$, from $v$ to $v'$, and from $v'$ to $\qeb$, where only the middle subpath contains undiscovered edges of $C_d$). To enable the jump in $\disc$, we first need to find two proxy $\disc$-vertices of $v$ and $v'$. Naturally, we let the proxies be the component vertices $\pi(\Gamma_{v})$ and $\pi(\Gamma_{v'})$ of the $A_{d}$-heavy components $\Gamma_{v} \ni v$ and $\Gamma_{v'} \ni v'$. Indeed, there is a type-3 edge in $\disc$ connecting $\pi(\Gamma_{v})$ and $\pi(\Gamma_{v'})$ of length $h_{\diam} \cdot s_{\ed}$, simulating the jump.

Now, discovered $C_d$-edges appear in the discovered graph, the subpaths between two consecutive discovered $C_{d}$-edges can be approximated via recursive scenario 1, and the subpath between $v$ and $v'$ can be approximated by the jump. Let $w$ be the last discovered vertex on $P$ before $v$, and let $w'$ symmetrically be the first discovered vertex on $P$ after $v'$. Then, the only missing parts are the subpath between $w$ and $v$, and also symmetrically the subpath between $v'$ and $w'$. These missing subpaths give us the following recursive scenario:
\begin{itemize}
\item (Recursive Scenario 2) First, we have a discovered vertex $w$. Next, we have a (possibly undiscovered) vertex $v$ that is inside an $A_{d}$-heavy (level-$d$) component $\Gamma_{v}$ and is incident to an edge of $C_{d}$. Also, there is a path $\check{P}$ between $w$ and $v$ in $G_{d - 1}$ that is a subpath of the original witness path $P$.

The goal for this recursive scenario is that the distance from $w$ to $\pi(\Gamma_{v})$ in the discovered graph is at most $\dist_{\disc}(w, \pi(\Gamma_{v}))\leq \poly(k)\cdot h$ \footnote{We can allow the looser upper bound of $\poly(k)\cdot h$ instead of $\poly(k)\cdot \leng(\check{P})$ since recursive scenario 2 appears only twice.}.
\end{itemize}
Note that the recursive scenario 2 is very different from the original scenario: one endpoint $v$ of $\check{P}$ is undiscovered, and its proxy $\pi(\Gamma_{v})$ does not show up in lower levels of the discovered graph. This scenario is the main challenge in following the argument of the expander case. In what follows, we focus on solving this scenario. 

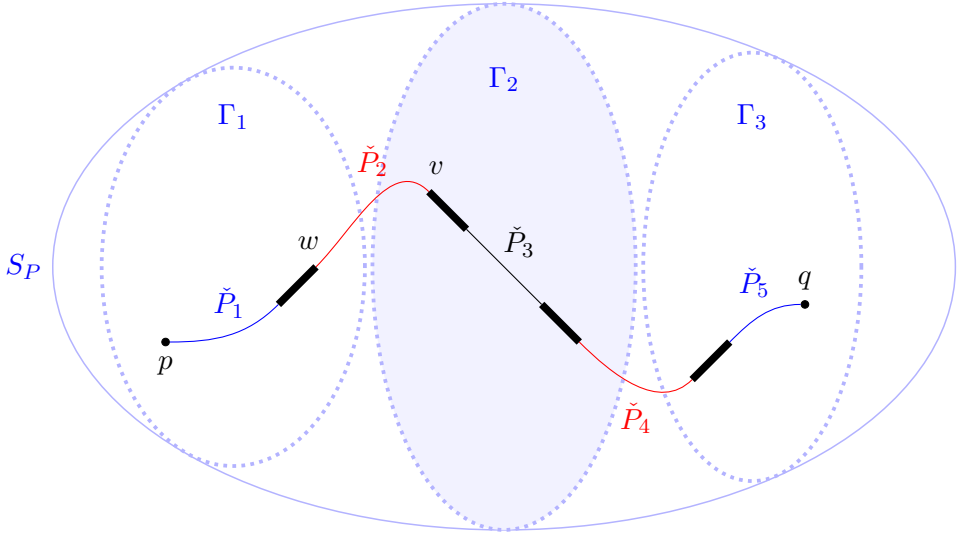
\begin{figure}
    \centering
    \begin{tikzpicture}[minor/.style={draw, fill=black, circle, minimum size=0.1cm, inner sep=0pt}, nodraw/.style={}, edges/.style={draw}, txt/.style = {draw=none}, every loop/.style={}]

        \node[ellipse, draw, color=blue!30, line width=0.02cm, minimum width=12cm, minimum height=7cm, label=left:{\color{blue}$S_P$}] (S) at (0, 0) {};

        \node[ellipse, draw, dotted, color=blue!30, line width=0.05cm, minimum width=3.5cm, minimum height=5.3cm] (G1) at (-3.6, 0) {};
        \node[ellipse, draw, dotted, color=blue!30, fill=blue!5, line width=0.05cm, minimum width=3.5cm, minimum height=7cm] (G2) at (0, 0) {};
        \node[ellipse, draw, dotted, color=blue!30, line width=0.05cm, minimum width=2.9cm, minimum height=5.7cm] (G3) at (3.3, 0) {};

        \node[minor, label=below:{$\qea$}] (p) at (-4.5, -1) {};
        \coordinate (wp1) at (-3, -0.5) {};
        \coordinate[label={[above, yshift=0.1cm, xshift=-0.1cm]:{$w$}}] (w2) at (-2.5, 0) {};
        \coordinate[label={[above, yshift=0.1cm, xshift=0.1cm]:{$v$}}] (wp2) at (-1, 1) {};
        \coordinate (w3) at (-0.5, 0.5) {};
        \coordinate (wp3) at (0.5, -0.5) {};
        \coordinate (w4) at (1, -1) {};
        \coordinate (wp4) at (2.5, -1.5) {};
        \coordinate (w5) at (3, -1) {};
        \node[minor, label={$\qeb$}] (q) at (4, -0.5) {};

        \node[txt] (GT1) at (-3.6, 2) {\color{blue}$\Gamma_1$};
        \node[txt] (GT2) at (0, 2.5) {\color{blue}$\Gamma_2$};
        \node[txt] (GT3) at (3.3, 2) {\color{blue}$\Gamma_3$};
        
        \draw[-, edges, color=blue] (p) to[out=0,in=-135] node[left, above, yshift=0.1cm]{\color{blue} $\check{P}_1$} (wp1);
        \draw[-, edges, line width=0.1cm] (wp1) to (w2);
        \draw[-, edges, color=red] (w2) to[out=45,in=135] node[left, above, yshift=0.2cm]{\color{red} $\check{P}_2$} (wp2);
        \draw[-, edges, line width=0.1cm] (wp2) to (w3);
        \draw[-, edges] (w3) to[out=-45,in=135] node[left, above, xshift=0.2cm]{$\check{P}_3$} (wp3);
        \draw[-, edges, line width=0.1cm] (wp3) to (w4);
        \draw[-, edges, color=red] (w4) to[out=-45,in=225] node[left, below, yshift=-0.1cm]{\color{red} $\check{P}_4$} (wp4);
        \draw[-, edges, line width=0.1cm] (wp4) to (w5);
        \draw[-, edges, color=blue] (w5) to[out=45,in=180] node[left, above, yshift=0.1cm, xshift=-0.1cm]{\color{blue} $\check{P}_5$} (q);
        
    \end{tikzpicture}
    \caption{Suppose that instead the component $\Gamma_2$ is $A_d$-heavy, and the components $\Gamma_1$ and $\Gamma_3$ are $A_d$-light. Then, the subpaths $\check{P}_1$ and $\check{P}_5$ (blue) correspond to recursive scenario 1, and the subpaths $\check{P}_2$ and $\check{P}_4$ (red) correspond to recursive scenario 2. The subpath $\check{P}_3$ does not correspond to any recursive scenario, as we jump directly from the first undiscovered $C_d$-edge on $P$ to the last.}
    \label{fig:rec-scenario-example-2}
\end{figure}

\begin{figure}
    \centering
    \begin{tikzpicture}[minor/.style={draw, fill=black, circle, minimum size=0.1cm, inner sep=0pt}, nodraw/.style={}, edges/.style={draw}, txt/.style = {draw=none}, every loop/.style={}]

        \node[ellipse, draw, color=blue!30, line width=0.02cm, minimum width=12cm, minimum height=7cm, label=left:{\color{blue}$S_P$}] (S) at (0, 0) {};

        \node[ellipse, draw, dotted, color=blue!30, line width=0.05cm, minimum width=3.5cm, minimum height=5.3cm] (G1) at (-3.6, 0) {};
        \node[ellipse, draw, dotted, color=blue!30, fill=blue!5, line width=0.05cm, minimum width=3.5cm, minimum height=7cm] (G2) at (0, 0) {};
        \node[ellipse, draw, dotted, color=blue!30, fill=blue!5, line width=0.05cm, minimum width=2.9cm, minimum height=5.7cm] (G3) at (3.3, 0) {};

        \node[minor] (p) at (-4.5, -1) {};
        \coordinate (wp1) at (-3, -0.5) {};
        \coordinate (w2) at (-2.5, 0) {};
        \coordinate (wp2) at (-1, 1) {};
        \coordinate (w3) at (-0.5, 0.5) {};
        \coordinate (wp3) at (0.5, -0.5) {};
        \coordinate (w4) at (1, -1) {};
        \coordinate (wp4) at (2.5, -1.5) {};
        \coordinate (w5) at (3, -1) {};
        \node[minor] (q) at (4, -0.5) {};

        \node[txt] (GT1) at (-3.6, 2) {\color{blue}$\Gamma_1$};
        \node[txt] (GT2) at (0, 2.5) {\color{blue}$\Gamma_2$};
        \node[txt] (GT3) at (3.3, 2) {\color{blue}$\Gamma_3$};
        
        \draw[-, edges, color=blue] (p) to[out=0,in=-135] node[left, above, yshift=0.1cm]{\color{blue} $\check{P}_1$} (wp1);
        \draw[-, edges, line width=0.1cm] (wp1) to (w2);
        \draw[-, edges, color=red] (w2) to[out=45,in=135] node[left, above, yshift=0.2cm]{\color{red} $\check{P}_2$} (wp2);
        \draw[-, edges, line width=0.1cm] (wp2) to (w3);
        \draw[-, edges] (w3) to[out=-45,in=135] node[left, above, xshift=0.2cm]{$\check{P}_3$} (wp3);
        \draw[-, edges, line width=0.1cm] (wp3) to (w4);
        \draw[-, edges] (w4) to[out=-45,in=225] node[left, below, yshift=-0.1cm]{$\check{P}_4$} (wp4);
        \draw[-, edges, line width=0.1cm] (wp4) to (w5);
        \draw[-, edges] (w5) to[out=45,in=180] node[left, above, yshift=0.1cm, xshift=-0.1cm]{$\check{P}_5$} (q);

    \end{tikzpicture}
    \caption{If instead the components $\Gamma_2$ and $\Gamma_3$ were both $A_d$-heavy, and the component $\Gamma_1$ was $A_d$-light, only the subpath $\check{P}_1$ (blue) corresponds to recursive scenario 1 and the subpath $\check{P}_2$ (red) to recursive scenario 2. Since $\qeb$ itself appears in an $A_d$-heavy component $\Gamma_3$, we can directly jump from $\pi(\Gamma_2)$ to $\pi(\Gamma_3)$ and then take the type-2 edge from $\pi(\Gamma_3)$ to $\qeb$. } %
    \label{fig:rec-scenario-example-3}
\end{figure}
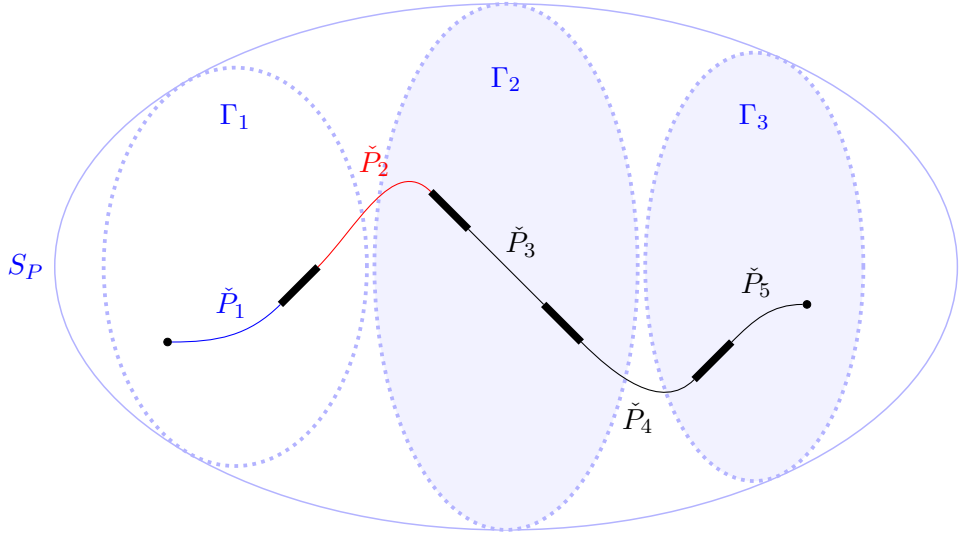

\paragraph{Solving Recursive Scenario 2.} The first idea is that, given that we did not discover $v$ at level $d$, let us try to discover it at level $d-1$. To make this formal, let $\check{S}_{\check{P}}$ be a cluster in the level-$(d - 1)$ neighborhood cover $\mathcal{N}_{d - 1}$ containing $\check{P}$, and let $\check{\Gamma}_{v}$ be the component of $\check{S}_{\check{P}}$ containing $v$.

Note that although $\check{\Gamma}_{v}$ is a component at level $d-1$, we can additionally recover level-$d$ information with it. Using this information, if $\check{\Gamma}_{v}$ is \emph{$A_{d}$-light}, we can discover all $C_{d}$-edges incident to $\check{\Gamma}_{v}$, and in particular we discover the edge of $C_d$ incident to $v$ along with $v$ as one of its two endpoints. This is a good case, as when $v$ is discovered, we can apply recursive scenario 1, which guarantees a good $\disc$-distance between $w$ and $v$. Combining this with the type-2 edge connecting $v$ and $\pi(\Gamma_{v})$, we obtain a good upper bound on the $\disc$-distance between $w$ and $\pi(\Gamma_{v})$ as desired.

Therefore, the hard case is when $\check{\Gamma}_{v}$ is $A_{d}$-heavy. To handle this case, we first need to add some extra connections between different levels of the discovered graph $\disc$ (by further exploiting the key property \Cref{lemma:GeneralKeyProperty}), which we will be able to use to reduce recursive scenario 2 to a purely level-$(d - 1)$ scenario.

\medskip

\noindent\underline{More Connections Between Levels in $\disc$.} To transfer from level $d$ to level $d-1$ in this case, the key observation is that we can apply \Cref{lemma:GeneralKeyProperty} on $\Gamma_{v}$ and $\check{\Gamma}_{v}$: they are both $A_{d}$-heavy, and more importantly, for any two vertices $x \in \Gamma_v$ and $y \in \Check{\Gamma}_v$, we have
\[
\dist_{G}(x, y) \leq \dist_G(x, w) + \dist_G(w, y) \leq \diam(S_{P}) + \diam(\check{S}_{\check{P}}) \leq %
2h_{\diam},
\]
because $\Gamma_{v}$ and $\check{\Gamma}_{v}$ are components of the level-$d$ cluster $S_{P}$ and level-$(d-1)$ cluster $\check{S}_{\check{P}}$ respectively, and $S_{P}$ and $\check{S}_{\check{P}}$ share a common vertex $w$. Recalling that $A_{d}$ is $(h_{\ed}=2h_{\diam},s_{\ed})$-length $\phi$-expanding in $G$, we obtain $\dist_{G\setminus F}(\Gamma_{v},\check{\Gamma}_{v})\leq h_{\ed}\cdot s_{\ed}$. This hints us to add more type-3 edges to $\disc$, even between components in different levels.
\begin{enumerate}
\item[3'.] (more edges between heavy components) For two components $\Gamma,\Gamma'$ belonging to clusters $S$ and $S'$ (which can be at different levels), if both $\Gamma$ and $\Gamma'$ are $A_{j}$-heavy for the larger level $j$ between the two, and $S$ and $S'$ share a common discovered vertex, add an artificial edge between $\pi(\Gamma)$ and $\pi(\Gamma')$ of length $h_{\ed} \cdot s_{\ed} + 4 h_\diam = 2h_{\diam}\cdot k + 4 h_\diam = 2hk^{2} + 4hk$.
\end{enumerate}
After going down levels from $\pi(\Gamma_{v})$ to $\pi(\check{\Gamma}_{v})$ via the extra type-3' edge between them in the discovered graph, now it suffices to show a good upper bound on the $\disc$-distance between $w$ and $\pi(\check{\Gamma}_{v})$.

However, before stating the reduced recursive scenario 2, there is one remaining issue, which we need to exploit nestedness to solve: we need to show that the component $\check{\Gamma}_v$ is \emph{$A_{d - 1}$-heavy}. %

\medskip

\noindent\underline{Exploiting Nestedness.} Indeed, nestedness immediately implies that because $\check{\Gamma}_{v}$ is $A_{d}$-heavy, it is also $A_{d-1}$-heavy as $A_{d} \preceq A_{d-1}$. We now formalize the reduced scenario 2 as follows.
\begin{itemize}
\item (Recursive Scenario 2') First, we have a discovered vertex $w$. Next, we have a (possibly undiscovered) vertex $v$ that is inside an $A_{d - 1}$-heavy (level-$(d-1)$) component $\check{\Gamma}_{v}$ and is incident to an edge of a higher-level cut. Also, there is a path $\check{P}$ between $w$ and $v$ in $G_{d - 1}$ contained in the cluster $\check{S}_{\check{P}}$ that $\check{\Gamma}_v$ is a component of, that is a subpath of the original witness path $P$.

The goal for this recursive scenario is that the distance from $w$ to $\pi(\check{\Gamma}_{v})$ in the discovered graph is at most $\dist_{\disc}(w, \pi(\check{\Gamma}_{v})) \leq \poly(k)\cdot h$. %

\end{itemize}

To motivate nestedness, consider what might happen if the component $\check{\Gamma}_{v}$ could be $A_{d-1}$-light. Then, the last $C_{d-1}$-edge $e$ on $\check{P}$ before $v$ and both of its endpoints could be undiscovered. We would then encounter the (hypothetical) recursive scenario 3, where the subpath $\check{\check{P}}$ of $\check{P}$ between the right endpoint of $e$ and $v$ has \emph{both} endpoints undiscovered. This is problematic, as we require accessing a discovered endpoint to gather the information needed for the base case of the recursion discussed later and to construct the required type-3' edges in $\disc$. Furthermore, even accounting for a potential solution to this issue, nestedness is required to obtain an approximation that is polynomial rather than exponential in $k$. In particular, if $\check{\Gamma}_{v}$ was $A_{d - 1}$-light, a recursive scenario 2 could split again into two recursive scenarios of type 2, for up to $2^k$ at the bottom of the hierarchy.

On the other hand, with this additional guarantee that $v$ lies inside the $A_{d-1}$-heavy level-$(d-1)$ component $\check{\Gamma}_{v}$, when we continue to break down $\check{P}$ into subpaths in $G_{d - 2} = G - C_{d} - C_{d-1}$, we will again only encounter recursive scenarios 1 and 2, and in particular at most one recursive scenario 2, exactly as in \Cref{fig:rec-scenario-example-3}. %

\paragraph{The Base Cases.} In the discussion above, we intentionally ignore the base cases.

For the first base case, regardless of whether we are in recursive scenario 1 or 2, suppose that the level-$(d-1)$ cluster $\check{S}_{\check{P}}$ entirely containing $\check{P}$ has no failed tree edge, i.e., $T_{\check{S}_{\check{P}}}$ is disjoint from $F$. A good thing about this base case is that $\check{S}_{\check{P}}$ has a unique component $\Gamma$ that is the whole cluster itself. However, we can no longer obtain the information of $\Gamma$ (i.e., heavy/light, incident edges), because (as discussed at the end of \Cref{sect:OverviewExpanderCase}) we obtain component information of a cluster by accessing its failed tree edges.

Fortunately, in both recursive scenarios 1 and 2, we have a discovered endpoint $w\in \check{S}_{\check{P}}$ of $\check{P}$, allowing us to associate the unique-component information with the vertex label of $w$ (that we can store along with $w$ in edge labels that discover $w$). %
For vertices discovered from a vertex label, we only recover a very minimal \emph{fingerprint}. The fingerprint of a vertex $v$ stores only its position in each cluster containing it, so that the appropriate edges of type 2 incident to $v$ in $\disc$ may be added. As no further recursion is needed in this base case, a 'fourth tier' of labels is not necessary.

The second base case is when the recursion reaches level $d = 0$. Then, as the node weighting $A_0$ is the degree weighting of the graph, we can proceed as in the expander case discussed in \Cref{sect:OverviewExpanderCase}.

\paragraph{Comparison to Previous Expander-Based FT-Connectivity Algorithms.} After providing an overview of our algorithm, we now highlight again the challenges we overcame, in comparison with previous expander-based fault-tolerant connectivity oracles and labeling schemes \cite{PatrascuT07,long2025connectivity}. 

When the input graph is general, solving fault-tolerant connectivity problems via (classic) expander hierarchies admits a straightforward bottom-up merging procedure. For example, in \cite{long2025connectivity}, the clusters and their corresponding trees are very well-structured: the clusters (from all levels) form a laminar family, and thus there is a single global tree serving as the skeleton of these clusters (i.e., the tree of a cluster is simply a subtree of the global tree). Consequently, the components generated in a FT-connectivity query are also well-structured, which enables a simple bottom-up merging procedure.

In contrast, since we are dealing with a distance problem, we need to employ tools that capture distance information, such as length-constrained expanders and sparse neighborhood covers, which introduce additional technical complications. In particular, we can no longer rely on a single global tree, and therefore, require a novel top-down approach instead. We also face additional difficulties with the aforementioned \emph{fractional cuts} required for length-constrained expander decompositions. Handling these fractional cuts typically requires randomization, but we employ additional techniques to \emph{derandomize} our whole algorithm.

\section{Preliminaries}

As normal, we denote graphs by $G = (V, E, l, u)$ with vertex set $V$ of size $n := |V|$, edge set $E$ of size $m := |E|$, edge lengths $l(e) \geq 1$, and edge capacities $u(e) \geq 1$. Parallel edges are allowed, but $m$ is assumed to be polynomial in $n$ so that $O(\log n) = O(\log m)$ for notational simplicity. All graphs considered in this paper are undirected.

\subsection{Length-Constrained Objects} 

\paragraph{Moving Cuts.} A \textit{$h$-length} moving cut $C : E \rightarrow \left\{0, \frac{1}{h}, \frac{2}{h}, \dots, 1\right\}$ (or $h$-length cut for short) assigns to each edge $e \in E$ a cut value $C(e)$ which is a multiple of $\frac{1}{h}$ between zero and one. The size of $C$ is defined as $|C| := \sum_e C(e) \cdot u(e)$. The length increase associated with a $h$-length moving cut $C$ is $\ell_{C, h}(e) := h \cdot C(e)$. For a graph $G$ and a $h$-length moving cut $C$, we denote by $G - C$ the graph $G$ with edge lengths $l_{G - C} := l_G + \ell_{C, h}$. This notation is only used when $h$ is clear from context.

\paragraph{Node Weightings.} A \textit{node weighting} $A : V \rightarrow \mathbb{R}_{\geq 0}$ is an assignment of a nonnegative value to each vertex of a graph. For a vertex subset $S \subseteq V$, we write $A(S) := \sum_{v \in S} A(v)$. The size $|A|$ of a node weighting $A$ is defined as $|A| := \sum_{v} A(v) = A(V)$. We write $A \preceq A'$ for two node weightings $A, A'$ if $A(v) \leq A'(v)$ for all $v \in V$. The \textit{degree node weighting} $\deg_G$ of a graph $G$ assigns the weight of each vertex to equal its capacitated degree, i.e. $\deg_G(v) := \sum_{e \in E \text{ incident to $v$}} u(e)$. The \textit{degree node weighting} $\deg_C$ of a cut $C$ assigns the weight of each vertex to equal the capacitated cut-value of its incident edges, i.e. $\deg_C(v) := \sum_{e \in E \text{ incident to $v$}} u(e) \cdot C(e)$. Note that $\deg_C \preceq \deg_G$ and $|\deg_C| = 2|C|$ for all moving cuts $C$ on $G$. %

\paragraph{Demands.} A \textit{demand} $D : V \times V \rightarrow \mathbb{R}_{\geq 0}$ assigns a non-negative demand value $D(u, v)$ to each ordered pair of vertices $u, v \in V$. The size $|D|$ of a demand is defined as $|D| := \sum_{u, v \in V} D(u, v)$. The \textit{load} $\load(D)$ of a demand $D$ is the node weighting that assigns to each vertex a weight equal to the total demand value that vertex is involved in, i.e. $\load(D)(v) := \sum_{u \in V} D(u, v) + D(v, u)$. A demand $D$ is called \textit{$A$-respecting} for a node weighting $A$ if $\load(D) \preceq A$. A demand $D$ is called \textit{$h$-length-constrained} (or $h$-length for short) if it assigns positive demand values only to pairs of vertices that are within distance at most $h$, i.e. for all $(u, v) \in \supp(D)$, $\dist_G(u, v) \leq h$.

\paragraph{$h$-Length Separation.} Let $C$ be a $h$-length moving cut. We say a pair of vertices $u, v \in V$ are \textit{$h$-length separated} by $C$ if their distance in $G - C$ is strictly larger than $h$, i.e. $\dist_{G - C}(u, v) > h$. For a $h$-length demand $D$, the \textit{$h$-length separated demand value} $\sep_h(C, D)$ of $C$ and $D$ is the total demand value between pairs of vertices the cut $h$-separates, i.e.
\begin{equation*}
    \sep_h(C, D) := \sum_{\substack{u, v \in V\\\dist_{G - C}(u, v) > h}} D(u, v).
\end{equation*}

\paragraph{$h$-Length Sparsity.} The \textit{$h$-length sparsity} of a $h$-length moving cut $C$ with respect to a \emph{demand} $D$ is the ratio of $C$'s size to the $h$-length separated demand value of $C$ and $D$, i.e.
\begin{equation*}
    \spars_h(C, D) := |C| / \sep_h(C, D).
\end{equation*}
For a \textit{length slack} $s \geq 1$, the \textit{$(h, s)$-length sparsity} of a $hs$-length moving cut $C$ with respect to a node weighting $A$ is defined as the \emph{minimum} $(hs)$-length sparsity of $C$ with respect to any \emph{$h$}-length $A$-respecting demand $D$, i.e.
\begin{equation*}
    \spars_{h, s}(C, A) := \min_{\text{$A$-respecting, $h$-length $D$}} \spars_{h s}(C, D).
\end{equation*}
We say a moving cut $C$ is $(h, s)$-length $\phi$-sparse with respect to $A$ if $\spars_{h, s}(C, A) \geq \phi$.

Like with regular (non-length-constrained) expanders, a sequence of sparse length-constrained cuts is a sparse length-constrained cut \cite{ImprovedLCED24}. For this paper, the weaker result below bounding only the size of the union of cuts suffices. A proof of \Cref{thm:SmallUnionSparseCuts} is included in \Cref{sect:SmallUnionSparseCut}.

\begin{restatable}{theorem}{SmallUnionSparseCuts}
\label{thm:SmallUnionSparseCuts}
    Let $G_0$ be a graph, $A$ be a node weighting and $(C_1, \dots, C_k)$ a sequence of $hs$-length cuts where each cut $C_i$ is $(h, s)$-length $\phi_i$-sparse in $G_i = G_0 - \sum_{j < i} C_j$ with respect to $A$. Then,
    \begin{equation*}
        \sum_i |C_i| / \phi_i \leq n^{O(1 / s)} \log(n) \cdot |A|.
    \end{equation*}
\end{restatable}

\paragraph{Length-Constrained Expansion.} For a length $h$, a length slack $s \geq 1$ and a sparsity $\phi$, a graph $G$ is $(h, s)$-length $\phi$-expanding for a node weighting $A$ if there exists no $(hs)$-length cut $C$ that has $(h, s)$-length sparsity with respect to $A$ strictly less than $\phi$, i.e. the following is satisfied:
\begin{equation*}
    \phi \leq \min_{\text{$(hs)$-length cut } C} \spars_{h, s}(C, A).
\end{equation*}
In this case, we equivalently say $A$ is $(h, s)$-length $\phi$-expanding in $G$.

\paragraph{Length-Constrained Flows.} A \textit{multicommodity flow} $F$ in $G$ is a function that assigns a \textit{flow value} $F(P) \geq 0$ to each simple path $P$ in $G$. A path $P$ is a \textit{flow path} of $F$ if $F(P) > 0$, i.e. $P \in \supp(F)$. The \textit{value} $|F|$ of the flow is the total flow across all paths, i.e. $|F| = \sum_P F(P)$.

The \textit{congestion} $\mathrm{cong}_F(e)$ of an edge is the ratio of total flow of paths using $e$ to the capacity of $e$, i.e. $\mathrm{cong}_F(e) := \sum_{e \in P} F(P) / u(e)$, and the congestion of the flow is the maximum congestion of any edge. The \textit{length} of the flow is the maximum length $\len_G(P)$ of a flow path $P \in \supp(F)$ of $F$. A flow may be referred to as congestion-$\eta$ and length-$h$ even if it has congestion \emph{at most} $\eta$ and length \emph{at most} $h$.

The \textit{demand routed by the flow} $D_F$ is the demand where $D_F(u, v) = \sum_{P \text{ a $(u, v)$-path}} F(P)$, i.e. the $(u, v)$-demand is the total flow from $u$ to $v$. A demand $D$ is said to be \textit{routable} with congestion $\eta$ and length $h$ if there exists a congestion-$\eta$, length-$h$ flow $F$ routing $D$ (i.e. $D = D_F$).

Like with (regular) expanders, the length-constrained expansion of a graph (with respect to a node weighting) is tightly dependent on the routability of demands (respecting the node weighting) on the graph.

\begin{theorem}[Routing Characterization of Length-Constrained Expanders \cite{OrigLCED22}]\label{thm:lce-routing}
    For any graph $G$, node weighting $A$, length $h \geq 1$, length slack $s \geq 1$ and sparsity $\phi < 1$, the following hold:
    \begin{itemize}
        \item \textbf{Length-Constrained Expanders Have Good Routings.} If $A$ is $(h, s)$-length $\phi$-expanding in $G$, then any $h$-length $A$-respecting demand on $G$ can be routed with congestion $O(\log(n) / \phi)$ and length $hs$.
        \item \textbf{Not Length-Constrained Expanders Have a Hard Demand.} If $A$ is not $(h, s)$-length $\phi$-expanding in $G$, then there exists a $h$-length $A$-respecting demand on $G$ that cannot be routed with congestion $1 / 2\phi$ and length $(hs) / 2$.
    \end{itemize}
\end{theorem}

\subsection{Neighborhood Covers}

\begin{definition}[Clustering] A clustering $\mathcal{S}$ of \textit{diameter} $\hdiam$ of a graph $G$ is a collection of disjoint vertex sets $S_1, S_2, \dots, S_{|\mathcal{S}|}$ each of diameter at most $\hdiam$ in $G$.
\end{definition}

\begin{definition}[Neighborhood Cover] A neighborhood cover $\calN$ of \textit{diameter} $\hdiam$, \textit{covering radius} $\hcov$ and \textit{width} $\omega$ of a graph $G$ is a collection of $\omega$-many diameter-$h_{\mathrm{diam}}$ clusterings $\calS_1, \calS_2, \dots, \calS_\omega$ of $G$, such that for any vertex $v$, there exists a cluster $S \in \calS \in \calN$ containing the $\hcov$-radius neighborhood $\ball(v, \hcov)$ of $v$.
\end{definition}

For simplicity of notation, we assume that for any cluster $S \in \calS \in \calN$, the clustering $\calS$ containing $S$ is unique. This is without loss of generality, as removing a cluster from all but one clustering containing it in a neighborhood cover does not invalidate any of the properties of the neighborhood cover. We write just $S \in \calN$ to mean $S \in \calS \in \calN$ for some $\calS$ when the exact $\calS$ is not significant.

We use the following result for both constructive and existential neighborhood covers.

\begin{lemma}[Constructive Neighborhood Cover \cite{NHCover98}]\label{lem:ldd-alg}
    There is a deterministic algorithm that, given a graph $G$, covering radius $\hcov$, and length slack $\sldd \geq 2$, constructs a neighborhood cover $\mathcal{N}$ of $G$ of diameter $\hdiam = \sldd \cdot \hcov$, covering radius $\hcov$ and width $\omega = O(\sldd n^{1 / \sldd})$. The algorithm has running time $O((m \sldd + n \sldd^2) \cdot n^{2 / \sldd})$.
\end{lemma}

\subsection{Euler Tour Representation of Trees}\label{sec:euler-tour-def}

An Euler tour of a tree $T$ without edge lengths on vertex set $V$ is a sequence of length $2|V| - 1$ of vertices from $V$ produced by writing down the vertices visited in a DFS traversal traveling each edge of the tree in both directions exactly once, starting from some root vertex $r$.
\begin{algorithm}[H]
    \caption{Euler Tour} \label{alg:euler-tour}
    \begin{algorithmic}[1]
        \Function{EulerTour}{$T$, $r$, $p = \bot$}
            \State Let $\texttt{tour}$ be an empty list
            \For{neighbour $v$ of $r$ other than $p$}
                \State Append $r$ to $\texttt{tour}$
                \State Append $\Call{EulerTour}{T, v, r}$ to $\texttt{tour}$
            \EndFor
            \State Append $r$ to $\texttt{tour}$
            \State \Return $\texttt{tour}$
        \EndFunction
    \end{algorithmic}
\end{algorithm}

\noindent In this paper, we represent trees as arbitrary Euler tours, and write $T[t]$ for $t \in \{0, 1, \dots, 2|V| - 2\}$ to denote the vertex in the $t$\textsuperscript{th} index of the tour, and $T[t, t')$ to denote the \emph{set} of vertices in the range of the Euler tour \emph{without multiplicity}, i.e. $T[t, t') := \{v \in V : T[i] = v \text{ for some } i \in [t, t')\}$. The following is a list of basic properties of Euler tours.
\begin{itemize}
    \item Euler tours are cyclic: defining $T[t + 2|V| - 2] = T[t]$ for all $t$, any subarray $T[t, t +  2|V| - 2]$ of length $2|V| - 1$ of the tour is an Euler tour of the tree. Note that the period length is one shorter than the length of the Euler tour.
    \item For any edge $e = \{u, v\} \in T$, for either orientation $(u, v)$ of the edge, there is exactly one index $t \in \{0, 1, \dots, 2|V| - 3\}$ for which $T[t - 1] = u$ and $T[t] = v$. We denote this index by $\pos_T(u, v)$. For any edge $e = \{u, v\} \not\in T$, there is no index satisfying $T[t - 1] = u$ and $T[t] = v$ for either orientation of the edge.
    
    This, in particular, implies that any Euler tour of a tree uniquely determines the tree.
    \item Defining $\ita_T(v)$ as the minimum index $t \geq 0$ and $\itb_T(v)$ as the maximum index $t \leq 2|V| - 2$ satisfying $T[t] = v$, the ranges $[\ita_T(v), \itb_T(v)]$ for vertices $v \in V$ are laminar, i.e. for any two intervals, either they do not intersect or one contains the other. Further, the interval of $v$ contains the interval of $v'$ if and only if $v'$ is in $v$'s subtree (rooting the tree at $r$).
\end{itemize}

\begin{figure}
    \centering
    \begin{tikzpicture}[main/.style={draw, circle, minimum size=0.8cm}, minor/.style={draw, circle, minimum size=0.5cm}, edges/.style={draw}, light/.style = {fill=none}, heavy/.style = {fill=slightlylightgray}, txt/.style = {draw=none}, every loop/.style={}]

        \node[main] (a) at (0, 4) {$a$};
        \node[main, fill=blue!30] (b) at (-1, 3) {$b$};
        \node[main] (j) at (1, 3) {$j$};
        \node[main, fill=blue!30] (c) at (-2, 2) {$c$};
        \node[main, fill=blue!30] (d) at (-1, 1.5) {$d$};
        \node[main, fill=blue!30] (k) at (0, 2) {$k$};
        \node[main, semifill={lower=blue!30, upper=yellow!30, ang=90}] (e) at (-2, 0.5) {$e$};
        \node[main, semifill={lower=blue!30, upper=red!30, ang=90}] (g) at (0, 0.5) {$g$};
        \node[main, semifill={lower=blue!30, upper=yellow!30, ang=90}] (f) at (-2, -0.5) {$f$};
        \node[main, semifill={lower=blue!30, upper=red!30, ang=90}] (h) at (-0.5, -0.5) {$h$};
        \node[main, semifill={lower=blue!30, upper=red!30, ang=90}] (i) at (0.5, -0.5) {$i$};

        \node[txt] (tour) at (0, -1.5) {\texttt{(a, b, c, b, d, e, f, e, d, g, h, g, i, g, d, b, k, b, a, j, a)}};
        \node[txt] (btour) at (0, -2) {\texttt{(a, \colorbox{blue!30}{b, c, b, d, e, f, e, d, g, h, g, i, g, d, b, k, b}, a, j, a)}};
        \node[txt] (etour) at (0, -2.5) {\texttt{(a, b, c, b, d, \colorbox{yellow!30}{e, f, e}, d, \colorbox{red!30}{g, h, g, i, g}, d, b, k, b, a, j, a)}};

        \draw[-, edges] (a) -- (b);
        \draw[-, edges] (b) -- (c);
        \draw[-, edges] (b) -- (d);
        \draw[-, edges] (d) -- (e);
        \draw[-, edges] (e) -- (f);
        \draw[-, edges] (d) -- (g);
        \draw[-, edges] (g) -- (h);
        \draw[-, edges] (g) -- (i);
        \draw[-, edges] (b) -- (k);
        \draw[-, edges] (a) -- (j);
    \end{tikzpicture}
    \caption{A tree and its Euler tour rooted at $a$. The subtrees of $b$, $e$ and $g$ are highlighted in the tour and the tree. Note how the intervals of $e$ and $g$ are disjoint as neither is an ancestor of the other, and both are contained in the interval of their ancestor $b$.} %
    \label{fig:euler-tour-example}

\vspace{2cm}

    \centering
    \begin{tikzpicture}[main/.style={draw, circle, minimum size=0.8cm}, minor/.style={draw, circle, minimum size=0.5cm}, edges/.style={draw}, light/.style = {fill=none}, heavy/.style = {fill=slightlylightgray}, txt/.style = {draw=none}, every loop/.style={}]

        \node[main] (a1) at (-3, 4) {$a$};
        \node[main, fill=black!20] (b1) at (-4, 3) {$b$};
        \node[main] (j1) at (-2, 3) {$j$};
        \node[main, fill=black!20] (c1) at (-5, 2) {$c$};
        \node[main, fill=black!20] (d1) at (-4, 1.5) {$d$};
        \node[main, fill=black!20] (k1) at (-3, 2) {$k$};
        \node[main, fill=black!20] (e1) at (-5, 0.5) {$e$};
        \node[main] (g1) at (-3, 0.5) {$g$};
        \node[main, fill=black!20] (f1) at (-5, -0.5) {$f$};
        \node[main] (h1) at (-3.5, -0.5) {$h$};
        \node[main] (i1) at (-2.5, -0.5) {$i$};

        \draw[-, edges, line width = 0.03cm] (a1) -- (b1) [dotted];
        \draw[-, edges] (b1) -- (c1);
        \draw[-, edges] (b1) -- (d1);
        \draw[-, edges] (d1) -- (e1);
        \draw[-, edges] (e1) -- (f1);
        \draw[-, edges, line width = 0.03cm] (g1) -- (d1) [dotted];
        \draw[-, edges] (g1) -- (h1);
        \draw[-, edges] (g1) -- (i1);
        \draw[-, edges] (b1) -- (k1);
        \draw[-, edges] (a1) -- (j1);

        \node[main, fill=blue!30] (a2) at (3, 4) {$5$};
        \node[main, semifill={upper=red!30, lower=blue!30, ang=90}] (b2) at (2, 3) {$5$};
        \node[main] (j2) at (4, 3) {$2$};
        \node[main, fill=red!30] (c2) at (1, 2) {$3$};
        \node[main, semifill={upper=red!30, lower=blue!30, ang=90}] (d2) at (2, 1.5) {$1$};
        \node[main, fill=blue!30] (k2) at (3, 2) {$0$};
        \node[main, fill=red!30] (e2) at (1, 0.5) {$0$};
        \node[main, fill=red!30] (g2) at (3, 0.5) {$1$};
        \node[main, fill=red!30] (f2) at (1, -0.5) {$2$};
        \node[main] (h2) at (2.5, -0.5) {$3$};
        \node[main] (i2) at (3.5, -0.5) {$4$};

        \draw[->, edges, color=red!50, line width = 0.03cm] (a2) -- (b2) [dotted];
        \draw[-, edges] (b2) -- (c2);
        \draw[-, edges] (b2) -- (d2);
        \draw[-, edges] (d2) -- (e2);
        \draw[-, edges] (e2) -- (f2);
        \draw[->, edges, color=blue!50, line width = 0.03cm] (g2) -- (d2) [dotted];
        \draw[-, edges] (g2) -- (h2);
        \draw[-, edges] (g2) -- (i2);
        \draw[-, edges] (b2) -- (k2);
        \draw[-, edges] (a2) -- (j2);

        \node[txt] (tour) at (0, -1.5) {\texttt{(a, \colorbox{black!20}{b, c, b, d, e, f, e, d}, g, h, g, i, g, \colorbox{black!20}{d, b, k, b}, a, j, a)}};
        \node[txt] (btour) at (0, -2) {\texttt{(a, \colorbox{red!30}{b, c, b, d, e, f, e, d, g}, h, g, i, g, \colorbox{blue!30}{d, b, k, b, a}, j, a)}};
    \end{tikzpicture}
    \caption{A case of \Cref{lem:label-cover-component} with $\theavy = 12$. On the left, vertices are labeled with their labels, the edges in $F$ are dotted, and $C$ is the shaded connected component. On the right, vertices are labeled with the values $A(v)$, a vertex is red if it is contained in the interval following the red edge in $F$, and blue if it is in the interval following the blue edge in $F$. The Euler tours below show with shading the union of intervals corresponding to the component $C$, and in color the maximal intervals. Note that the two intervals cover the union of intervals corresponding to the component.}
    \label{fig:euler-tour-example2}
\end{figure}

For labeling schemes, we want to store information in the edges of $T$, such that after some edges $F \subseteq E_T$ have failed, we can recover information about connected components $C$ of $T \setminus F$ from the information stored on the edges in $F$. For this, we use the same technique as \cite{long2025connectivity}: storing at each tree edge $\{u, v\} \in T$, for both orientations $(u, v)$ of the edge, information about the vertices in some interval of the Euler tour immediately following the position of $(u, v)$ in the tour. Suppose $A$ is a node weighting with $A(v)$ being a measure of how many bits of information needs to be stored about $v$. Then, this interval should be the maximal interval such that sum of $A(v)$ over vertices on the interval is at most some threshold $\theavy$ which bounds the label size. The following Lemma shows that this approach recovers information about vertices in components $C$ with $A(C) \leq \theavy$.
\begin{lemma}\label{lem:label-cover-component}
Let $T$ be a tree on vertex set $S$, $A$ a node weighting on $S$, and $\theavy$ some threshold. Let $F \subseteq E_T$ be some subset of the tree edges, and let $C$ be a connected component in $T \setminus F$ such that $A(C) \leq \theavy$. Then,
\begin{equation*}
    C \subseteq \bigcup_{\{u, v\} \in F, v \in C} T[t_{u, v}, t'_{u, v})
\end{equation*}
where $t_{u, v} = \pos_T(u, v)$ is the Euler tour position of the orientation $(u, v)$ in $T$, and $t_{u, v} \leq t'_{u, v} \leq t_{u, v} + 2(|S| - 1)$ the maximum position in the Euler tour of $T$ for which $\sum_{w \in T[t_{u, v}, t'_{u, v})} A(w) \leq \theavy$.
\end{lemma}

\begin{proof}
    Let $t$ be any position in the Euler tour such that $v = T[t] \in C$ and $u = T[t - 1] \not\in C$. Then, $\{u, v\} \in F$, as each pair of adjacent positions in the Euler tour corresponds to an edge of the tree. Let $t'$ be the first index after $t$ such that $T[t'] \not\in C$. If the interval $[t_{u, v}, t'_{u, v})$ does not contain the interval $[t, t')$, then $\sum_{v \in T[t, t')} A(v) > \theavy$ by definition, but each vertex in $T[t, t')$ is in $C$, thus we would have $A(C) > \theavy$, contradicting the assumption. Thus, the union of these maximal intervals following orientations of edges in $F$ directed towards $C$ covers the union of all intervals in the tour containing vertices of $C$, and as every vertex in $V$ appears in the tour, we are done.
\end{proof}

\subsection{Path Decomposition in $G \setminus F$}

As part of the labeling scheme, we need to cover all shortest paths in $G \setminus F$ with few "piece paths" in $G$. The following result of \cite{PathDecomp02} shows that any shortest path in $G \setminus F$ can be written as the concatenation of at most $2f + 1$ shortest paths in $G$.
\begin{theorem}[Theorem 2 of \cite{PathDecomp02}]\label{lem:path-decomp-nonlexmin}
    For any undirected graph $G = (V, E, l)$ and set $F \subseteq E$ of up to $f$ edge failures, any shortest path in $G \setminus F$ can be written as the concatenation of at most $f + 1$ shortest paths in $G$ interleaved with up to $f$ edges in $G$.
\end{theorem}

However, the graph $G$ could have exponentially many shortest paths. To bound the number of different pieces we need to use, we fix some ordering $e_1, \dots, e_m$ of the edges in the graph, and focus on shortest paths with \textit{lexicographically maximum edge indicator vectors} $(\mathbb{I}[e_1 \in P], \dots, \mathbb{I}[e_m \in P])$, henceforth referred to as lex-max shortest paths. %
Note that as lex-max shortest paths are unique, any graph has exactly $\binom{n}{2}$ nonempty lex-max shortest paths, and that any subpath of a lex-max shortest path is a lex-max shortest path. Now, we have the following:
\begin{lemma}\label{lem:path-decomp}
    Let $G = (V, E, l)$ be an undirected graph and $F \subseteq E$ a set of up to $f$ failed edges. Then, any lex-max shortest path in $G \setminus F$ can be formed by the concatenation of at most $2f + 1$ edges and lex-max shortest paths in $G$.
\end{lemma}
\begin{proof}
    Let $\epsilon > 0$ be the minimum value $\len_G(P) - \len_G(P')$ between the shortest $(u, v)$-path $P$ and the shortest strictly-longer $(u, v)$-path $P'$ in $G$ over $(u, v) \in \binom{n}{2}$. Let $l'(e_i) := l(e_i) + \epsilon \cdot 2^{-i}$ and $G' = (V, E, l')$. Now, $\len_G(P) < \len_{G'}(P) < \len_{G}(P) + \epsilon$, thus the shortest $(u, v)$-path $P$ in $G'$ is a shortest $(u, v)$-path in $G$, and is in fact by the choice of $l'$ the lex-max shortest $(u, v)$-path in $G$. Applying \Cref{lem:path-decomp-nonlexmin} to $G'$, we obtain \Cref{lem:path-decomp}. 
\end{proof}
For reading the paper, note that one could disregard all appearances of the word 'lex-max' by making the simple assumption that all shortest paths in $G$ are unique.

\section{Nested Length-Constrained Expander Hierarchy}

\begin{restatable}{definition}{nestedlcehierarchy}\label{def:hierarchy}(Nested Length-Constrained Expander Hierarchy).
For a graph $G$ and node weighting $A_0$, a $h$-length $\phi$-sparse nested expander hierarchy $\{(A_i, C_i)\}_{i \in [d]}$ with depth $d$ and length slack $s$ consists of $d$ pairs of $hs$-fractional node weightings $A_i$ and $hs$-length cuts $C_i$, such that
\begin{enumerate}
    \item[(1)] $A_d$ is $(h, s)$-length $\phi$-expanding in $G$
    \item[(2)] \makebox[8cm]{$A_{i - 1}$ is $(h, s)$-length $\phi$-expanding in $G - C_i$ \hfill} (for all $i \in [d]$)
    \item[(3)] \makebox[8cm]{$\deg_{C_i} \preceq A_{i} \preceq A_{i - 1}$  \hfill} (for all $i \in [d]$)
\end{enumerate}
\end{restatable}

Due to the restriction $A_i \preceq A_{i - 1}$, length-constrained expander hierarchies do not necessarily exist for all node weightings $A_0$: consider for example any bipartite graph where the node weighting $A_0$ is zero on one half of the partition. Then, no nonzero cut $C_1$ satisfies $\deg_{C_1} \preceq A_0$.

To avoid this issue, instead of constructing a length-constrained expander hierarchy for a node weighting $A$, we want to construct the length-constrained expander hierarchy for some $A_0 \succeq A$ such that $|A_0|$ is not much larger than $|A|$.

In this section, we obtain the following results for constructive and existential length-constrained expander hierarchies.
\begin{theorem}[Constructive and Existential Length-Constrained Expander Hierarchy]\label{thm:hierarchy-result}
    For any graph $G$ with edge lengths $l \geq 1$ and capacities $u \geq 1$, a node weighting $A$, a length bound $h$, a length slack $s \geq 100$, and a desired depth $d \in \mathbb{N}$, a $h$-length expander hierarchy $\{(A_i, C_i)\}_{i \in [d]}$ for some $A_0 \succeq A$, $|A_0| \leq |A|(1 + |A|^{-1/d})$ with length slack $s$, depth $d$, and sparsity $\phi$ with
    \begin{align*}
        \phi^{-1} &= |A|^{1 / d} \cdot \tilde{O}(n^{O(s^{-0.5})}), &\text{can be constructed in poly-time.}\\
        \phi^{-1} &= |A|^{1 / d} \cdot O(n^{O(s^{-1})} \log n), &\text{exists.}\\
    \end{align*}
\end{theorem}

To construct length-constrained expander hierarchies, we work with \textit{partial} length-constrained expander hierarchies.
\begin{definition}[Partial Length-Constrained Expander Hierarchy]\label{def:hierarchy-partial}
    For a graph $G$ and node weighting $A_0$, a $h$-length $\phi$-sparse expander hierarchy $\{(A_i, C_i)\}_{i \in [d]}$ with depth $d$, length slack $s$, and \emph{shrink ratio} $\gamma$ consists of $d$ pairs of $hs$-fractional node weightings $A_i$ and $hs$-length cuts $C_i$, such that %
\begin{enumerate}
    \item[(1)] \makebox[8cm]{$|A_i| \leq \gamma \cdot |A_{i - 1}|$\hfill} (for all $i \in [d]$)
    \item[(2)] \makebox[8cm]{$A_{i - 1}$ is $(h, s)$-length $\phi$-expanding in $G - C_i$ \hfill} (for all $i \in [d]$)
    \item[(3)] \makebox[8cm]{$\deg_{C_i} \preceq A_{i} \preceq A_{i - 1}$  \hfill} (for all $i \in [d]$).
\end{enumerate}
\end{definition}
These replace the condition that $A_d$ is $(h, s)$-length $\phi$-expanding in $G$ with the condition that each successive node weighting multiplicatively shrinks by at least some small value $\gamma$. To motivate the naming, note that if $\gamma^d$ is small enough, $|A_d| \leq 1$, thus $A_d$ is trivially $(h, s)$-length $\phi$-expanding in $G$, and the hierarchy is a (non-partial) length-constrained expander hierarchy.
\begin{observation}\label{ob:depth-completeness-formula}
    Let $G$ be a graph with edge lengths $l \geq 1$ and capacities $u \geq 1$, and $\{(A_i, C_i)\}_{i \in [d]}$ a partial length-constrained expander hierarchy of length $h$, sparsity $\phi \leq \frac{1}{2}$, length slack $s \geq 2$, shrink factor $\gamma < 1$ and depth $d \geq \log_{1 / \gamma} |A_0|$. Then, $A_d$ is $(h, s)$-length $\phi$-expanding in $G$.
\end{observation}
\begin{proof}
By property (1) of the partial length-constrained expander hierarchy, $|A_d| \leq \gamma^d |A_0| \leq 1$. Thus, any $h$-length $A_d$-respecting demand $D$ satisfies $|D| \leq 1$, and can thus be routed in $G$ with congestion at most $1 \leq \frac{1}{2\phi}$ and length $h \leq (hs) / 2$ by just routing along shortest paths, as all capacities are at least $1$. Thus by \Cref{thm:lce-routing}, $A_d$ is $(h, s)$-length $\phi$-expanding in $G$.
\end{proof}

Constructing partial length-constrained expander hierarchies is easier, as they can be ``nested'' to obtain results through induction on $d$. This nesting is however not trivial: suppose you know how to construct a depth-$d$ partial length-constrained expander hierarchy. Take a moving cut $C_1$ so that $A$ is $(h, s)$-length $\phi$-expanding in $G - C_1$ and $|C_1| \leq \gamma |A|$, and take a $h$-length $\phi$-sparse partial expander hierarchy $\{(A_i, C_i)\}_{i \in \{2, 3, \dots, d + 1\}}$ of some $A_1 \succeq \deg_{C_1}$. Is $\{(A_i, C_i)\}_{i \in [d + 1]}$ a partial length-constrained expander hierarchy of $A_0 = A$? Unfortunately, there is one violated constraint: $A_1 \preceq A_0$ does not necessarily hold.

To fix this, it is easier to maintain a partial length-constrained expander hierarchy in an incremental setting. Then, we can update $A_0 \gets A_0 + A_1$, incrementally update $C_1$ so that $A_0$ is still $(h, s)$-length $\phi$-expanding in $G - C_1$, and we update the depth-$d$ hierarchy with the increment to $\deg_{C_1}$. This is a standard \textit{stabilization} approach, similar to what is used in \cite{ImprovedLCED24} to construct linked length-constrained expander decompositions.%

We obtain the following result:

\begin{restatable}{theorem}{constructivelcehierarchy}\label{thm:constructive-incremental-lce-hierarchy}
    Let $G$ be a graph with edge lengths $l \geq 1$ and capacities $u \geq 1$. Let $h$, $\phi$, $d$ and $s \geq 100$ be some fixed parameters. Then, for some
    \begin{align*}
        \gamma &= \phi \cdot n^{O(s^{-0.5})} \cdot \poly \log(n)   &\text{with polynomial time per update},\\
        \gamma &= \phi \cdot n^{O(s^{-1})} \cdot \log(n)                 &\text{existentially},
    \end{align*}
    if $\gamma \leq \frac{1}{2}$, the following holds:
    
    Let $A$ be an initially zero node weighting. The data structure \Cref{alg:inc-lce-hierarchy} maintains an (initially zero) partial $h$-length $\phi$-sparse expander hierarchy $\{(A_i, C_i)\}_{i \in [d]}$ with depth $d$, length slack $s$ and shrink factor $\gamma$ for a node weighting $A_0$ satisfying $A_0 \succeq A$ and $|A_0| \leq |A| / (1 - \gamma)$ under incremental updates to $A$:
    \begin{enumerate}
        \item The update sets $A \gets A'$ for a given node weighting $A' \succeq A$
        \item The algorithm selects $hs$-fractional node weightings $A_i'$ for $i \in \{0, 1, \dots, d\}$ and $hs$-length moving cuts $C_i'$ for $i \in [d]$, and updates $A_i \gets A_i + A_i'$ and $C_i \gets C_i + C_i'$.
    \end{enumerate}
\end{restatable}
Combining \Cref{thm:constructive-incremental-lce-hierarchy} and \Cref{ob:depth-completeness-formula}, we immediately obtain \Cref{thm:hierarchy-result}.
\begin{proof}(of \Cref{thm:hierarchy-result}). Let $\kappa = \tilde{O}(n^{O(s^{-0.5})})$ in the constructive case and $\kappa = n^{O(s^{-1})} \cdot \log n$ in the existential case be so that $\gamma = \phi \cdot \kappa$, and let $\phi = \frac{1}{2} |A|^{-1 / d} \cdot \kappa^{-1}$. Now, $\gamma^{d} \leq \frac{1}{2} |A|^{-1} \leq |A_0|^{-1}$, thus by \Cref{ob:depth-completeness-formula}, $A_d$ is $(h, s)$-length $\phi$-expanding in $G$. It remains to bound $|A_0|$, and we have $|A_0| \leq |A| / (1 - \gamma) \leq |A| + 2\gamma |A| \leq |A|(1 + |A|^{-1 / d})$.
\end{proof}
It remains to obtain \Cref{thm:constructive-incremental-lce-hierarchy}. As a key ingredient, we need to be able to maintain the length-constrained expander decomposition $C_1$ in an incremental setting. As in \cite{ImprovedLCED24}, we do this by computing and subtracting sparse length-constrained cuts until none exists. At this point, since no sparse length-constrained cut exists, the graph must be a length-constrained expander. By \Cref{thm:SmallUnionSparseCuts}, the sum of these sparse cuts is small.

For constructive results, we use \Cref{lem:cut-or-certify}, proven in \Cref{sec:polytime-approxcut}, to either find a sparse length-constrained cut or certify the node weighting is expanding. For existential results, we simply take an arbitrary length-$(h, s)$ sparsity-$\phi$ cut with respect to $A$ if one exists.

\begin{restatable*}{lemma}{cutorcertify}\label{lem:cut-or-certify}
    There is a polynomial-time algorithm $\Call{CutOrCertify}{G, A, h, s, s', \phi}$ that, given a graph $G = (V, E, u, l)$ with edge lengths $l \geq 1$ and capacities $u \geq 1$, a node weighting $A$, a length constraint $h \geq 1$, length slacks $s' \geq 4$, $s \geq 8s'$, and a sparsity parameter $\phi > 0$, either
    \begin{itemize}
        \item certifies that $A$ is $(h, s)$-length $\phi$-expanding in $G$ (returning an empty cut), or
        \item returns a nonempty $hs$-length moving cut $C$ and an $A$-respecting $hs'$-length demand $D$ such that $C$ is $hs$-length $\phi'$-sparse for $D$ for $\phi' = \phi \cdot \tilde{O}(n^{O(1 / s')})$.
    \end{itemize}
\end{restatable*}

\begin{lemma}\label{lem:cut-until-certify-properties}
    Let $G = (V, E, u, l)$ be a graph with capacities $u \geq 1$ and edge lengths $l \geq 1$, $A_0, A_1, \dots, A_t$ be an arbitrary monotonically increasing sequence of node weightings (i.e. satisfying $A_{i - 1} \preceq A_{i}$ for all $i \in [t]$) with $A_0 = \{0\}_V$, and $h, \phi$ and $s \geq 64$ some parameters. Let $C_0, C_1, \dots, C_t$ be a sequence of $hs$-length moving cuts constructed as $C_0 = \{0\}_E$ and $C_i = C_{i - 1} + \Call{CutUntilCertify}{G - C_{i - 1}, A_i, h, s, \phi, \textsc{flag}}$. Then, for all $i \in \{0, 1, \dots, t\}$,
    \begin{enumerate}
        \item $A_i$ is $(h, s)$-length $\phi$-expanding in $G - C_i$
        \item $|C_i| \leq \kappa \phi \cdot |A_i|$ where
        \begin{align*}
            \kappa &= n^{O(s^{-0.5})} \cdot \poly \log(n) &\text{if \textsc{flag} = \textsc{poly}}\\
            \kappa &= n^{O(s^{-1})} \cdot \log(n) &\text{if \textsc{flag} = \textsc{exist}}\\
        \end{align*}
    \end{enumerate}
    Furthermore, if $\textsc{flag} = \textsc{poly}$, each call to $\mathrm{CutUntilCertify}$ takes polynomial time. Regardless, over any sequence of operations, there are at most $\binom{n}{2}$ indices $i$ such that $C_i \neq C_{i - 1}$.
\end{lemma}

\begin{algorithm}[H]
    \caption{Maximal Sequence of Sparse Cuts} \label{alg:maximal-cut}
    \begin{algorithmic}[1]
        \Function{CutUntilCertify}{$G, A, h, s, \phi, \textsc{flag}$}
            \State Let $C \gets \{0\}_E$
            \Loop
                \State Let $C' \gets \left\{\begin{array}{lr}
                    \Call{CutOrCertify}{G - C, A, h, s, \sqrt{s}, \phi} & \textsc{flag} = \textsc{poly}\\
                    \text{an arbitrary $(h, s)$-length $< \phi$-sparse cut w.r.t. $A$ in $G - C$} & \textsc{flag} = \textsc{exist}
                \end{array}\right.$ 
                \If{$C'$ is nonzero} $C \gets C + C'$
                \Else \ break
                \EndIf
            \EndLoop
            \State \Return $C$
        \EndFunction
    \end{algorithmic}
\end{algorithm}

\begin{proof}
    Whenever $C'$ is nonzero, since the cut has non-infinite sparsity, there must exist some vertex pair $v, v'$ such that the distance between $v$ and $v'$ before applying the new cut is at most $hs'$, but the distance between $v$ and $v'$ after applying the cut is at least $hs > hs'$. Thus, at most $\binom{n}{2}$ returned cuts can be nonzero, guaranteeing there are at most $\binom{n}{2}$ indices such that $C_i \neq C_{i - 1}$. This guarantees the function terminates, and as each call to $\mathrm{CutOrCertify}$ takes polynomial time, each call to $\mathrm{CutUntilCertify}$ with $\textsc{flag} = \textsc{poly}$ takes polynomial time.

    The first property immediately follows from the stopping condition, as $\mathrm{CutUntilCertify}$ repeatedly applies cuts until $\mathrm{CutOrCertify}$ certifies $A$ is a $(h, s)$-length $\phi$-expanding in $G - C$, and $A$ is $(h, s)$-length $\phi$-expanding in $G - C$ if and only if there exists no $(h, s)$-length strictly-less-than $\phi$-sparse cut with respect to $A$ in $G$.

    It remains to show the second property, which follows from \Cref{thm:SmallUnionSparseCuts}. Let $C'_1, C'_2, \dots, C'_{t'}$ be the cuts produced by $\mathrm{CutOrCertify}$ and $l_1, \dots, l_t$ indices such that $C_i = \sum_{j \in [l_i]} C'_j$. Additionally, let $A'_j = A_i$ for all $j \in (l_{i - 1}, l_i]$, and $l_0 = 0$, $C'_0 = \{0\}_E$, $A'_0 = \{0\}_V$.
    
    First, consider the polynomial-time case $\textsc{flag} = \textsc{poly}$. By \Cref{lem:cut-or-certify}, each cut $C'_j$ is $(hs', s / s')$-length $\phi'$-sparse in $G - \sum_{j' < j} C'_{j'}$ with respect to $A'_j$ (and all $A'_{j'}$, $j' \geq j$, as $A'_{j} \succeq A'_{j - 1}$), where $\phi' = \tilde{O}(\phi \cdot s' n^{O(1 / s')})$. Thus, by \Cref{thm:SmallUnionSparseCuts}, each cut $C_{i} = \sum_{j \in [l_i]} C'_j$ has size
    \begin{equation*}
        |C_i|   \leq \phi' \cdot n^{O\left(\frac{s'}{s}\right)} \log(n) \cdot |A|
                = \phi \cdot \tilde{O}\left(n^{O\left(\frac{s'}{s}\right) + O\left(\frac{1}{s'}\right)}\right) \cdot |A| = \phi \cdot \tilde{O}\left(n^{O(s^{-0.5})}\right) \cdot |A| = \kappa \phi \cdot |A|
    \end{equation*}
    where $\kappa = \tilde{O}(n^{O(s^{-0.5})})$, as desired. Finally, consider the existential case $\textsc{flag} = \textsc{exist}$. Each cut $C'_j$ is $(h, s)$-length $< \phi$-sparse in $G - \sum_{j' < j} C'_{j'}$ with respect to all $A'_{j'}$, $j' \geq j$. Thus, by \Cref{thm:SmallUnionSparseCuts}, each cut $C_{i} = \sum_{j \in [l_i]} C'_j$ has size $|C_i| < \phi \cdot n^{O(1 / s)} \log(n) \cdot |A| = \kappa \phi \cdot |A|$.
\end{proof}

We are now ready to prove \Cref{thm:constructive-incremental-lce-hierarchy}.

\constructivelcehierarchy*

\begin{algorithm}[h]
    \caption{Incremental Length-Constrained Expander Hierarchy} \label{alg:inc-lce-hierarchy}
    \begin{algorithmic}[1]
        \Class{IncrementalLCEH}
        \Data
            \State Constant $\textsc{flag} \in \{\textsc{poly}, \textsc{exist}\}$
            \State Constant graph $G = (V, E, l, u)$
            \State Constants $h$, $s$, $\phi$, $d$
            \State Node weightings $A$ and $A_0, A_1, \dots, A_d$, all initially $\{0\}_{V}$
            \State Moving cuts $C_1, C_2, \dots, C_d$ of length $hs$, all initially $\{0\}_E$
        \EndData
        \Function{RecUpdate}{$\Delta_A$, $j$}
            \If{$j < d$} \Loop
                \State Let $\Delta_C \gets \Call{CutUntilCertify}{G - C_{j + 1}, A_j + \Delta_A, h, s, \phi, \textsc{flag}}$
                \If{$\Delta_C$ is nonzero}
                    \State $C_{j + 1} \gets C_{j + 1} + \Delta_C$
                    \State $\Delta_A \gets \Delta_A + \Call{RecUpdate}{\deg_{\Delta_C}, j + 1}$
                \Else \ break
                \EndIf
            \EndLoop \EndIf
            \State $A_j \gets A_j + \Delta_A$
            \State \Return $\Delta_A$
        \EndFunction
        \State
        \Function{Update}{$A'$}
            \State Let $\Delta_A \gets A' - A$
            \State $A \gets A'$
            \State \Call{RecUpdate}{$\Delta_A, 0$}
        \EndFunction
        \EndClass
    \end{algorithmic}
\end{algorithm}

\begin{proof}
    Select $\kappa$ to satisfy \Cref{lem:cut-until-certify-properties}, and let $\gamma = 4 \kappa \phi$. We will show that for all $j \in \{0, 1, \dots, d\}$, initially and whenever a call to $\Call{RecUpdate}{\Delta_A, j}$ for any first parameter is returned from, the following are satisfied:
    \begin{itemize}
        \item If $j < d$, then $\{(A_{j'}, C_{j'})\}_{j' \in \{j + 1, j + 2, \dots, d\}}$ is a partial $h$-length $\phi$-sparse expander hierarchy with depth $d - j$, length slack $s$ and shrink factor $\gamma$ for $A_j$.
        \item For all $j' \geq j$, we have
        \begin{align*}
            A_{j'} &= A_{j' + 1} + A                & \text{if } j' = 0\\
            A_{j'} &= A_{j' + 1} + \deg_{C_{j'}}    & \text{if } j' \in [d - 1]\\
            A_{j'} &= \deg_{C_{j'}}                 & \text{if } j' = d
        \end{align*}
    \end{itemize}
    Note that the combination of the properties $A_{j'} = A_{j' + 1} + \deg_{C_{j'}}$ and $|A_{j' + 1}| \leq \gamma |A_{j'}|$ implies $|A_{j'}| \leq |\deg_{C_{j'}}| / (1 - \gamma)$ for $j' > 0$ (with the even tighter bound $|A_d| \leq |\deg_{C_d}|$ for $j' = d$) and the combination of $A_0 = A_1 + A$ and $|A_1| \leq \gamma |A_0|$ implies $|A_0| \leq |A| / (1 - \gamma)$.

    The second property can be immediately observed from the algorithm. $\Call{RecUpdate}{\cdot, j}$ returns the change to the node weighting $A_j$, thus line 14 ensures that whenever a call to $\Call{RecUpdate}{\cdot, j}$ returns, the node weighting $A_j$ equals $A_{j + 1}$ plus the sum of all $\Delta_A$ over calls to $\Call{RecUpdate}{\cdot, j}$ so far. As $\Call{RecUpdate}{\deg_{\Delta_C}, j + 1}$ is immediately called after an update $C_{j + 1} \gets C_{j + 1} + \Delta_C$ to the cut $C_{j + 1}$, this second term equals $\deg_{C_{j}}$ for $j > 0$, and as a call to $\Call{RecUpdate}{\Delta_A, 0}$ is made whenever the node weighting $A$ is incremented by $\Delta_A$, the term equals $A$ for $j = 0$. Thus, property 2 always holds.
    
    Now, we show the first property holds by induction on $j$ from $j = d - 1$ to $j = 0$. Fix some $j$ and suppose the property holds for all $j' > j$. Note that the cut $C_{j + 1}$ is the sum of calls to $\mathrm{CutUntilCertify}$ on the current cut $C_{j + 1}$, monotonically increasing node weightings $A_j + \Delta_A$ (note that within each call $\Delta_A$ monotonically increases, and $\Delta_A$ is added to $A_j$ before returning). Thus, by \Cref{lem:cut-until-certify-properties}, after each iteration of line 11, we have
    \begin{itemize}
        \item $A_j + \Delta_A$ is $(h, s)$-length $\phi$-expanding in $G - C_{j + 1} - \Delta_C$
        \item $|C_{j + 1} + \Delta_C| \leq \kappa \phi \cdot |A_j + \Delta_A| = \frac{\gamma}{4} |A_j + \Delta_A|$
    \end{itemize}
    Thus, in particular when the function returns, we have that $A_j$ is $(h, s)$-length $\phi$-expanding in $G - C_{j + 1}$ and $|\deg_{C_{j + 1}}| = 2|C_{j + 1}| \leq \frac{\gamma}{2} |A_j|$, and from induction we have $|A_{j + 1}| \leq |\deg_{C_{j + 1}}| / (1 - \gamma) \leq \frac{\gamma}{2(1 - \gamma)} |A_j| \leq \gamma |A_j|$, using $\gamma \leq \frac{1}{2}$. Finally, $A_j \succeq A_{j + 1}$ follows from $A_j = A_{j + 1} + A$ or $A_j = A_{j + 1} + \deg_{C_{j'}}$, depending on if $j = 0$, and $\deg_{C_{j + 1}} \preceq A_{j + 1}$ follows from $A_{j + 1} = A_{j + 2} + \deg_{C_{j + 1}}$ or $A_{j + 1} = \deg_{C_{j + 1}}$, depending on if $j + 1 = d$.

    Finally, each update takes polynomial time if $\textsc{flag} = \textsc{poly}$ and terminates if $\textsc{flag} = \textsc{exist}$ as by \Cref{lem:cut-until-certify-properties}, inside calls $\Call{RecUpdate}{\cdot, j}$, there can be at most $\binom{n}{2}$ times that the cut $\Delta_C$ is nonzero after executing line 11, thus the total number of recursive calls to $\Call{RecUpdate}{\cdot, j + 1}$ over any sequence of updates is at most $\binom{n}{2}$. Outside recursive calls, the work inside $\mathrm{RecUpdate}$ consists of at most $\binom{n}{2} + 1$ calls to $\mathrm{CutUntilCertify}$, each of which takes polynomial time when $\textsc{flag} = \textsc{poly}$. 
\end{proof}
\section{Labeling Scheme}

This section presents our main result, the approximate distance labeling scheme under edge failures. In \Cref{thm:labeling}, we say an edge label $\elabel(e)$ is \emph{trivial} if it only stores an $O(\log n)$-bit identifier of $e$, otherwise it is \emph{non-trivial}. We will bound the number of non-trivial edge labels, because when we change to the distance sensitivity oracle setting in \Cref{sect:Oracle}, we want to claim oracle size sublinear in $m$ (when $m$ is much larger than $n$).

\begin{theorem}[Approximate distance labeling scheme under edge failures]\label{thm:labeling}
    There is a labeling scheme that, given a graph $G = (V, E, \ell)$ with edge lengths ${1 \leq \ell(e) \leq L}$, a length slack $s$ at least a fixed large constant, and a bound $f$ on the number of edge failures, assigns a label $\text{ELabel}(e), \text{VLabel}(v)$ to every edge $e$ and vertex $v$, such that there is a deterministic algorithm with work $\tilde{O}(|\vlabel(\qea)| + |\vlabel(\qeb)| + \sum_{e \in F} |\elabel(e)|)$ for the following problem:
    
    \begin{itemize}
        \item 
        given \textbf{only} the vertex labels $\text{VLabel}(\qea), \text{VLabel}(\qeb)$ of query endpoint vertices $\qea, \qeb \in V$ and the edge labels $\{\text{ELabel}(e)\}_{e \in F}$ of a set $F$ of up to $f$ edge failures, return either \emph{UNREACHABLE} or a value $\hat{d}$, such that
        \begin{itemize}
            \item The output is UNREACHABLE if and only if $\qea$ and $\qeb$ are disconnected in $G \setminus F$
            \item Otherwise, the value $\hat{d}$ satisfies $\hat{d} \leq \dist_{G \setminus F}(\qea, \qeb) \leq \hat{d} \cdot s$
        \end{itemize}
    \end{itemize}
    The labels can be constructed deterministically in polynomial time, have size
    \begin{align*}
        |\text{VLabel}(v)| &\leq \tilde{O}(f^2 \cdot n^{O(s^{-1/4})} \cdot \log^2(L)),\\
        |\text{ELabel}(e)| &\leq \tilde{O}(f^4 \cdot n^{O(s^{-1/4})} \cdot \log^3(L)),
    \end{align*}
    and have at most $\tilde{O}(n^{1+O(s^{-1/4})}\log (L))$ non-trivial edge labels.
    If the label construction is not required to take polynomial time, the labels can have size
    \begin{align*}
        |\text{VLabel}(v)| &\leq \tilde{O}(f^2 \cdot n^{O(s^{-1/3})} \cdot \log^2(L)),\\
        |\text{ELabel}(e)| &\leq \tilde{O}(f^4 \cdot n^{O(s^{-1/3})} \cdot \log^3(L)),
    \end{align*}
    and have at most $\tilde{O}(n^{1+O(s^{-1/3})}\log (L))$ non-trivial edge labels.
\end{theorem}

\subsection{The Labels}\label{sec:label-definition}

\noindent We first define the labels of \Cref{thm:labeling}.  Note that these labels are parameterized in terms of three slack parameters $\sldd, \sed, d$ rather than just $s$ in addition to the edge failure bound $f$. In \Cref{sec:label-correctness}, we show how to implement the distance approximation algorithm with length slack $s = O(\sldd \cdot \sed \cdot d)$, and in \Cref{sec:label-construction}, we bound the sizes of both polynomial-time constructible labels and existential labels following the template.

The eventual selection will be $\sldd = O(s^{1/4})$, $\sed = O(s^{1/2})$, $d = O(s^{1/4})$ for constructive labels, and $\sldd = O(s^{1/3})$, $\sed = O(s^{1/3})$, and $d = O(s^{1/3})$ for existential labels. 

\begin{definition}[Label Template]\label{def:label-template}
Let $G = (V, E, l)$ be a graph with edge lengths $1 \leq \ell(e) \leq L$, vertex set $V = [n]$, and a unique identifier $\id(e)$ for each edge, and let $\sldd \geq 2$, $\sed \geq 100$, $d \geq 1$ and $f$ be set parameters.

The labeling scheme involves (almost fully) separate vertex and edge labels $\vlabel_i(v)$, $\elabel_i(e)$ for each power-of-two scale $i \in \{0, 1, \dots, i_\text{max}\}$, $i_\text{max} = \lceil \log(nL) \rceil$, with the vertex and edge labels being the collections of labels of individual scales.
\begin{align*}
    \vlabel(v) &:= \{\vlabel_i(v)\}_{i \in \{0, 1, \dots, i_{\text{max}}\}}\\
    \elabel(e) &:= \{\elabel_i(e)\}_{i \in \{0, 1, \dots, i_{\text{max}}\}}\cup\{\id(e)\}
\end{align*}
Recall that we say an edge label $\elabel(e)$ is \emph{trivial} if it only stores $\id(e)$, which means $\elabel_{i}(e)$ are all empty.

We now describe the labels of length scale $h := 2^i$, starting with the structure the labels use.
\begin{enumerate}
    \item Let $\hcov := 2h$, $\hdiam := \hcov \cdot \sldd$, $\hed := 2\hdiam$.

    \item\label{structure2} Let $\{(A_j, C_j)\}_{j \in [d]}$ be a $\hed$-length expander hierarchy with length slack $\sed$ and depth $d$ for a node weighting $A_0 \succeq \deg(G)$, $|A_0| \leq 2|\deg(G)|$, with sparsity $\phi$ as large as possible.

    \textit{(By \Cref{thm:hierarchy-result}, we can obtain)}
    \begin{align*}
        \phi &= n^{-O(1/d)} \cdot n^{-O(1/\sqrt{\sed})} / \poly \log(n) &\text{constructively}\\
        \phi &= n^{-O(1/d)} \cdot n^{-O(1/\sed)} / \log n &\text{existentially}
    \end{align*}
    \item Let $\theavy$ be a value such that for each $j$, any $\hed$-length $A_j$-respecting demand can be routed in $G - C_{j + 1}$ (letting $C_{d + 1} = 0$) with congestion $\theavy / f$ and length $\hed \sed$.    
    
    \textit{(By \Cref{thm:lce-routing}, we can obtain $\theavy = f \cdot O(\log(n) / \phi)$.)} %
    \item Let $\thit := \frac{h}{\hed \sed (2f + 1) d} = ((2f + 1) \cdot 4 \sldd \sed d)^{-1}$
    \item\label{structure5} For each $j \in \{0, 1, \dots, d\}$,
    \begin{itemize}
        \item Let $G_j$ be the graph formed by applying all cuts $C_{j'}$ with $j' > j$ to $G$, i.e. the graph with edge lengths $l_{G_j}(e) := l_G(e) + \sum_{j' = j + 1}^{d} \hed \sed \cdot C_{j'}(e)$. %
        \item Let $\calN_j$ be a neighborhood cover of $G_j$ of covering radius $\hcov$ and cluster diameter $\hdiam$. Assign a unique identifier $\id_j(S)$ to each cluster $S \in \calN_j$.
        \item Let $L_j \subseteq E$ be an edge set such that
        \begin{itemize}
            \item If $j = 0$, then $L_j = E$.
            \item If $j > 0$, for all paths $P$ that either consist of a single edge or are a lex-max shortest path in $G$, either $C_j(P) \leq \thit$ or $L_j \cap P \neq \emptyset$.

            \vspace{0.1cm} %
            
            \textit{(Think of $L_j$ as sampling each edge with probability $C_j(e) \cdot O(\frac{1}{\thit} \log n)$. The actual deterministic selection described in \Cref{sec:label-construction} is a derandomization of this.)}
        \end{itemize}
    \end{itemize}
    \item\label{structure6} For each cluster $S$ that appears in at least one $\calN_j$\footnote{Here for two clusters $S\in {\cal N}_{j}$ and $S'\in {\cal N}_{j'}$ inside two different covers (i.e. $j\neq j'$), if $S$ and $S'$ correspond to the same vertex subset, we will treat them as the same cluster. That is why a cluster $S$ can appear in more than one covers. We note it would alternatively also be fine to think of them as two different clusters.}, let $T_S$ be a spanning tree of $S$ of radius $\hdiam$ in $G_j$ for the minimum $j$ for which $S \in \calN_j$. %

    \textit{(Note that as edge lengths monotonically decrease in $G_0, G_1, \dots, G_d$, the radius of $T_S$ is $\hdiam$ in each $G_j$ for which $S \in \calN_j$.)}
\end{enumerate}
With this structure in place, the labels $\vlabel_i$, $\elabel_i$ are defined as follows, with \textit{vertex fingerprints} and \textit{edge fingerprints} as helper definitions for basic information (for this scale) about vertices and edges respectively.
\begin{itemize}
    \item The \textbf{vertex fingerprint} $\fp(v)$ of a vertex consists of the value $v \in [n]$, and, for each $j$, for each cluster $S \in \calN_j$ containing $v$,
    \begin{itemize}
        \item The unique identifier $\id_{j}(S)$ and size $|S|$ of the cluster.
        \item The Euler tour indices $\text{start}_{T_{S}}(v)$ and $\text{end}_{T_{S}}(v)$ of $v$ in $T_{S}$.
        \item For each $j'$, the node weight $A_{j'}(S)$ of the cluster $S$, and the node weight $A_{j'}(\{v' \in T_{S}[\text{start}_{T_{S}}(v), \text{end}_{T_{S}}(v)]\})$ of the subtree of $v$ in $T_S$
    \end{itemize}
    \item The \textbf{edge fingerprint} $\fp(e)$ of an edge $e = \{u, v\}$ consists of
    \begin{itemize}
        \item The unique identifier $\id(e)$ and length $\ell(e)$ of the edge.
        \item The fingerprints $\{\fp(u), \fp(v)\}$ of the endpoints of $e$.
    \end{itemize}
    \item The \textbf{vertex label} $\vlabel_i(v)$ of a vertex consists of
    \begin{itemize}
        \item $\fp(v)$.
        \item For each $j'$, for each cluster $S \in \calN_{j'}$ containing $v$, for each $j$ such that $A_{j}(S) \leq \theavy$, $\fp(e)$ of every edge $e \in L_{j}$ incident on a vertex in $S$.
    \end{itemize}
    \item The \textbf{edge label} $\elabel_i(e)$ of an edge $e = \{u, v\}$ consists of, for each $j'$, for each cluster $S \in \calN_{j'}$ such that $e \in T_{S}$,
    \begin{itemize}
        \item $\id_{j'}(S)$
        \item for both orientations of $e$, let $t = \text{pos}_{T_{S}}(u, v)$ (or $\text{pos}_{T_{S}}(v, u)$) be the tour index of the orientation of $e$ in $T_S$. For all $j$, let $t'$ be the maximum index satisfying
        \begin{enumerate}
            \item $t \leq t' \leq t + 2(|S| - 1)$
            \item $\sum_{v \in T_{S}[t, t')} A_{j}(v) \leq \theavy$
        \end{enumerate}
        For each edge $\{u', v'\} = e' \in L_{j}$ incident on a vertex $u' \in T_{S}[t, t')$, the label includes $\fp(e')$ and the labels $\{\vlabel(u'), \vlabel(v')\}$ of the endpoints of $e'$.

        \vspace{0.1cm}
    
        \textit{(Not just this scale's labels $\vlabel_i$ -- this is the only dependence between scales.)}
    \end{itemize}
\end{itemize} 
\end{definition}

\subsection{Distance Approximation}\label{sec:label-correctness}

In this section, we will show how to obtain approximate distances in $G \setminus F$ given the labels of the query endpoints and the failed edges $F$. Specifically, we prove the following:

\begin{restatable}{lemma}{distapproxwithlabels}\label{lem:dist-approx-with-labels}
    Let $\vlabel$, $\elabel$ be vertex and edge labels following \Cref{def:label-template} for a graph $G$ and parameters $\sldd, \sed, d, f$. There is a deterministic algorithm that, for any $\qea, \qeb \in V$ and $F \subseteq E$ with $|F| \leq f$, given the labels $\vlabel(\qea), \vlabel(\qeb)$ and $\{\elabel(e)\}_{e \in F}$, returns either UNREACHABLE or a value $\hat{d}$, such that for $s = O(\sldd \cdot \sed \cdot d)$,
    \begin{equation*}
        \hat{d} \leq \dist_{G \setminus F}(\qea, \qeb) \leq \hat{d} \cdot s.
    \end{equation*}
    Furthermore, the algorithm takes $\tilde{O}(|\vlabel(\qea)| + |\vlabel(\qeb)| + \sum_{e \in F} |\elabel(e)|)$ work.
\end{restatable}

The distance approximation is done by constructing from the labels a \textit{discovered graph} $\disc$ that approximates distances in $G \setminus F$. Specifically, it will have the property that the distance between any two vertices in $\disc$ is at least their distance in $G \setminus F$, and the distance between any two \textit{waypoint vertices} in $\disc$ is at most $s$ times their distance in $G \setminus F$, where a waypoint vertex is a vertex whose vertex label is recovered, either because the vertex is an endpoint of the query, or because the vertex label is stored in the edge label of some failed edge.

\begin{definition}[Waypoints]
For fixed vertex and edge labels $\vlabel, \elabel$ and a set of edge failures $F$, the set of \textit{waypoints} $W = W(\qea, \qeb, F) \subseteq V$ is the set of vertices $v$ for which either $\vlabel(v)$ appears in some edge label $\elabel(e)$, $e \in F$ or $v \in \{\qea, \qeb\}$.
\end{definition}

The discovered graph $\disc$ is the union of scale-$2^i$ discovered graphs $\disc_i$ (defined in \Cref{def:DiscoveredGraph} below) for each $i$, merging vertices that appear in multiple graphs. Notably, the waypoint vertices always appear in each $\disc_i$.

Consider now a specific scale $2^i$. We use the following additional terms:
\begin{itemize}
    \item \textbf{Fingerprinted} vertex/edge: a vertex or edge is \textit{fingerprinted} if its fingerprint appears in the scale-$2^i$ label $\vlabel_{i}(w)$ of some waypoint $w \in W$ or $\elabel_{i}(e)$ of some $e \in F$. 
    
    We emphasize that the notion of fingerprinted vertices is defined w.r.t. the specific $i$. Moreover, each waypoint $w \in W$ (the definition of waypoints is irrespective of $i$) is always a fingerprinted vertex w.r.t. any $i$, as $\vlabel_i(w)$ stores the fingerprint of $w$.

    \item \textbf{Components} of a cluster: the \textit{components} $C \subseteq S$ of a cluster $S \in \calN_j$ are the connected components of $T_{S} \setminus F$.
    
    \textit{(Note that in the case where $G$ has parallel edges, the edge $\{u, v\} \in T_{S}$ does \emph{not} appear in $T_{S} \setminus F$ if and only if there is at least one failed edge $e \in F$ with endpoints $u$ and $v$.)}

    \item \textbf{Heavy} component: a component $C$ is called $A_j$-heavy if $A_j(C) > \theavy$, otherwise it is $A_{j}$-light. The following key observation shows the monotonicity of the heaviness of a component, and we will heavily exploit this observation in the future argument. 

\begin{observation}
\label{ob:MonotonicHeavy}
An $A_{j}$-heavy component $C$ is also $A_{j'}$-heavy for all $j'\leq j$.
\end{observation}
\begin{proof}
This is because $A_{j}\preceq A_{j'}$ from \Cref{def:hierarchy} of the hierarchy.
\end{proof}

\end{itemize}

We now give the definition of the discovered graph.
 
\begin{definition}[Discovered Graph]
\label{def:DiscoveredGraph}
The discovered graph $\disc = \disc(\qea, \qeb, F)$ is the union of discovered graphs $\disc_i = \disc_i(W, F)$ for every scale $2^i$. The vertex and edge set of the graph $\disc_i$ are as follows:

\begin{itemize}
    \item \textbf{Vertices}. The vertex set of the scale-$2^i$ discovered graph $\disc_i$ consists of the following
    \begin{enumerate}
        \item Every fingerprinted vertex $v$
        \item A vertex for every component $C$ of each cluster $S$ with any fingerprinted vertex $v \in S$, denoted by $\pi_{S}(C)$. We often write just $\pi(C)$ when $S$ is clear from context.
    \end{enumerate}
    \item \textbf{Edges}. The discovered graph's edge set consists of the following
    \begin{enumerate}
        \item\label{type:Edge1} Every fingerprinted edge not in $F$ (of the same length, connecting the same vertices)
        \item\label{type:Edge2} For every cluster $S$, component $C \subseteq S$ and fingerprinted vertex $v \in C$, an edge of length $\hdiam$ connecting $v$ and $\pi_S(C)$
        \item\label{type:Edge3} For every pair of components $C \subseteq S \in \calN_{j}$ and $C' \subseteq S' \in \calN_{j'}$ such that there exists a waypoint vertex in $S \cap S'$ and both $C$ and $C'$ are $A_{\max(j, j')}$-heavy, an edge of length $\hed \sed + 2\hdiam$ connecting $\pi_S(C)$ and $\pi_{S'}(C')$
    \end{enumerate}
\end{itemize}
\end{definition}

Now, we show that the discovered graph overestimates distances in \Cref{lem:disc-graph-lowerbound}, that the discovered graph overestimates the distance between any waypoint pair in $G \setminus F$ by at most a multiplicative factor of $O(\sldd \cdot \sed \cdot d)$ in \Cref{lem:disc-graph-upperbound}.

\begin{restatable}{lemma}{discgraphlowerbound}\label{lem:disc-graph-lowerbound}
    For any two vertices $u, v$ with $u, v \in V(\disc) \cap V$,
    \begin{equation*}
        \dist_{G \setminus F}(u, v) \leq \dist_{\disc}(u, v)
    \end{equation*}
\end{restatable}

\begin{proof}
It suffices to focus in a single scale $i$, and prove the inequality $\dist_{G \setminus F}(u, v) \leq \disc_{i}(u, v)$ for any vertices $u, v \in V(\disc_i) \cap V$. Fix now the scale $i$.

We will assign to each vertex $v' = \pi(C) \in V(\disc_i) \setminus V$ a representative vertex $\croot(v') = \croot(C) \in V$, and let $\croot(v) := v$ for $v \in V$. Then, it suffices to show that $\dist_{G \setminus F}(\croot(u), \croot(v)) \leq \len_{\disc_i}(e)$ for any edge $\{u, v\} = e \in \disc_i$.

To select $\croot(v')$ for the vertex $v' \in V(\disc_i)$ corresponding to the component $C$, note that since each tree $T_{S}$ was picked to have radius $\hdiam$ in each $G_j$ for which $S \in \calN_j$, there must exist for each component $C$ of $S$ a root vertex $\croot(C) \in C$ such that $\dist_{G_j \setminus F}(\croot(C), v) \leq \hdiam$ for all $v \in C$ and $j$ such that $S \in \calN_j$ (thus also $\dist_{G \setminus F}(\croot(C), v) \leq \hdiam$).

We now consider each of the three edge types in $\disc_i$.
\begin{itemize}
    \item type-\ref{type:Edge1}: Fingerprinted edges not in $F$ exist in $G \setminus F$, so the claim is trivial.
    \item type-\ref{type:Edge2}: By the definition of $\croot(C)$, $\dist_{G \setminus F}(\croot(C), v)$ for $v \in C$ is at most $\hdiam$.
    \item type-\ref{type:Edge3}: Take some components $C \subseteq S \in \calN_j$ and $C' \subseteq S' \in \calN_{j'}$ for which an edge of the third type exists. We need to show $\dist_{G \setminus F}(\croot(C), \croot(C')) \leq \hed \sed + 2\hdiam$.

    Assume without loss of generality that $j' < j$. By definition of the third type of edges, we have $A_j(C), A_j(C') > \theavy$ and there is a waypoint $w \in S\cap S'$. Thus, 
    \begin{align*}
        \max_{u \in C, v \in C'} \dist_{G_j}(u, w) &\leq \max_{u \in C, v \in C'} \dist_{G_j}(u, w) + \dist_{G_j}(w, v)\\
            &\leq \mathrm{diam}_{G_j}(S) + \mathrm{diam}_{G_j}(S')\\
            &\leq \mathrm{diam}_{G_j}(S) + \mathrm{diam}_{G_{j'}}(S')\\
            &\leq 2\hdiam = \hed.
    \end{align*}
    Therefore, we can pick an $A_j$-respecting demand $D$ from $C$ to $C'$ of value $|D| > \theavy$ that is $\hed$-length in $G_j$. Since $G_j = G - \sum_{j' = j + 1}^{d} C_{j'}$ is a graph with edge lengths at least those in $G - C_{j + 1}$ (defining $C_{d + 1} = 0$), the length of $D$ in $G - C_{j + 1}$ must also be at most $\hed$. Because $A_{j}$ is $(\hed,\sed)$-length $\phi$-expanding in $G - C_{j+1}$, there exists a $\hed \sed$-length flow in $G - C_{j + 1}$ of congestion at most $\theavy / f$ routing $D$ by \Cref{thm:lce-routing} and our choice of $\theavy$. Thus, as $|F| \leq f$ and the congestion of any edge in $F$ in the flow is strictly less than a $\frac{1}{f}$-fraction of the total flow value, there has to exist a flow path $P$ in $G \setminus F$ from some $u \in C$ to $v \in C'$, satisfying $\len_{G}(P) \leq \len_{G - C_{j + 1}}(P) \leq \hed \sed$ as the flow has path length $\hed \sed$. Finally, as $\dist_{G \setminus F}(\croot(C), u), \dist_{G \setminus F}(v, \croot(C')) \leq \hdiam$, we have $\dist_{G \setminus F}(\croot(C), \croot(C')) \leq \hed \sed + 2\hdiam$.
\end{itemize}
\end{proof}

Before proving the upper bound \Cref{lem:disc-graph-upperbound} on distances between waypoints in $\disc$, we state three \Cref{lem:light-comp-known,lem:hit-in-gminusf,lem:cluster-close}, all still in the context of a fixed scale $i$.

The first shows we obtain information about edges in $L_j$ incident to $A_j$-light components. %
\begin{lemma}\label{lem:light-comp-known}
    For any cluster $S \in \calN_{j'}$ containing at least one waypoint vertex $v \in S$, any component $C \subseteq S$, and any $j$ such that $C$ is $A_j$-light,
    \begin{itemize}
        \item Every edge $e \in L_j$ with at least one endpoint in $C$ is fingerprinted.
        \item Further, if $C$ is a strict subset of $S$ (so $T_S \cap F \neq \emptyset$), both endpoints of such $e$ are waypoints.
    \end{itemize}
\end{lemma}
\begin{proof}
    Take any such $S, C, v$. Suppose first $C = S$, i.e. $T_S \cap F = \emptyset$. Since $A_j(S) \leq \theavy$, the vertex label of $v$ stores the fingerprint of every edge $e \in L_j$ incident to a vertex in $S$.

    Now suppose $C$ is a strict subset of $S$. Then, by \Cref{lem:label-cover-component} and the choice of labels, since $A_j(C) \leq \theavy$, the ranges $T_S[t, t')$ of edges $e \in T_S \cap F$ incident to $C$, for $j$ and the orientation towards $e$, cover the union of ranges of the Euler tour of $T_S$ consisting of vertices of $C$. The label $\elabel_i(e)$ stores for every vertex $u' \in T_S[t, t')$, for every edge $e' \in L_j$ incident to $u'$, the fingerprint of $e'$ and the vertex labels of both endpoints of $e'$. Since every vertex $u' \in C$ appears in at least one of these ranges, the two claims hold. 
\end{proof}

Sometimes, we will use a weaker version of the contrapositive of \Cref{lem:light-comp-known} to certify that some component $C$ is $A_{j}$-heavy.

\begin{corollary}
\label{coro:CertifyHeavyComp}
Consider any cluster $S\in{\cal N}_{j'}$ containing at least one waypoint vertex $v\in S$, any component $C\subseteq S$ and any $j$. When $C$ is a strict subset of $S$, if there exists an edge $e\in L_{j}$ with at least one endpoint in $C$ but one of $e$'s endpoints is not a waypoint, then $C$ is $A_{j}$-heavy.
\end{corollary}

The second, \Cref{lem:hit-in-gminusf}, states that each lex-max shortest path in $G \setminus F$ for which $C_j(P)$ is large enough has a sampled edge $e \in L_j$ early on in the path. This follows from the choice of $\thit$ and \Cref{lem:path-decomp}, which you may recall states that each lex-max shortest path in $G \setminus F$ being a concatenation of at most $2f + 1$ lex-max shortest paths and edges in $G$.

\begin{lemma}\label{lem:hit-in-gminusf}
    For any set $F \subseteq E$ of up to $f$ edge failures and lex-max shortest paths $P$ in $G \setminus F$ from a vertex $u$ to $v$, for each $j \in [d]$, there exists a (possibly empty) prefix $P'$ of $P$ from $u$ to some $v'$ such that
    \begin{itemize}
        \item $\len_{G_{j - 1}}(P') \leq \len_{G_j}(P) + \frac{h}{d}$
        \item Either $v' = v$ or $v'$ is followed on $P$ by an edge in $L_{j}$
    \end{itemize}
\end{lemma}

\begin{proof}
    Let $P''$ be the minimum prefix of $P$ satisfying $C_j(P'') > (2f + 1)\thit$. If no such prefix exists or $P'' = P$, the path $P''$ satisfies the desired property by considering the prefix $P'=P$. Otherwise, by \Cref{lem:path-decomp}, the path $P''$ is the concatenation of at most $2f + 1$ lex-max shortest paths (and single-edge paths) in $G$, thus at least one subpath $P'''$ among those satisfies $C_j(P''') > \thit$, and by the property of $L_j$, contains an edge of $L_j$. Thus, in particular, $P''$ contains an edge in $L_j$. Let $P'$ be a maximum prefix of $P''$ not containing an edge in $L_j$. The prefix $P'$ is immediately followed by an edge in $L_j$, and since $P'$ is a strict prefix of $P''$,
    \begin{align*}
        \len_{G_{j - 1}}(P')    &= \len_{G_j}(P') + \hed \sed C_j(P')\\
                                &\leq \len_{G_j}(P') + \hed \sed (2f + 1)\thit\\
                                &\leq \len_{G_j}(P) + \hed \sed (2f + 1)\thit = \len_{G_j}(P) + \frac{h}{d}.
    \end{align*}
\end{proof}

The third, \Cref{lem:cluster-close}, states that for a short path in $G_j \setminus F$ from a waypoint $u$ to a vertex $v'$ satisfying some specific conditions, the vertex $u$ is close in $\disc_i$ to the vertex of every component $C \subseteq S \in \calN_j$ containing $v'$. The specific conditions are in place to prove the Lemma through induction.

\begin{lemma}\label{lem:cluster-close}
Let $u$ be a waypoint and $P$ be a lex-max shortest path in $G \setminus F$ from $u$ to some vertex $v'$, such that the path does not contain any waypoints other than $u$. Then, if for some $j \in [d]$,
\begin{enumerate}
    \item the vertex $v'$ is incident to an edge in $L_{j'}$ for some $j' \geq j$, and
    \item the length of $P$ in $G_{j - 1}$ is at most $h \cdot \left(1 + \frac{d - j + 1}{d}\right)$,
\end{enumerate}
then, for every component $C \subseteq S \in \calN_j$ such that $v' \in C$ and $u \in S$, the distance in $\disc_i$ between $u$ and the vertex $\pi(C)$ of $C$ is at most $(2\hed \sed + 6\hdiam) \cdot j$.
\end{lemma}

\begin{proof}
    We perform induction on increasing $j$. First, note that if $v' = u$, the distance in $\disc_i$ between $u$ and the vertex of any $C$ containing $u$ is $\hdiam$ because of type-\ref{type:Edge2} edges. We thus can assume $v' \neq u$.

    Take a cluster $S' \in \calN_{j - 1}$ containing the entirety of $P$, which must exist as $\len_{G_{j - 1}}(P) \leq 2h = \hcov$, and let $C'$ be the component of $S'$ containing $v'$. In what follows, we will consider multiple cases. For each case, we will either reach the desired statement of the lemma, or establish the following \emph{intermediate claim}: $C'$ is $A_{j}$-heavy, and the distance in $\disc_{i}$ between $u$ and the vertex $\pi(C')$ of $C'$ is at most 
    \[
    (2\hed \sed + 6\hdiam) \cdot (j - 0.5).
    \]
    We then finish by completing the proof given the above intermediate claim.

    \medskip

    \noindent\textbf{Case 1.} 
    First, suppose that $S' = C'$, i.e. that the cluster $S'$ consists of only one component. Then, since $S'$ contains a waypoint vertex in $u$, by \Cref{lem:light-comp-known}, either 
    \begin{itemize}
    \item (Case 1a) $C'$ is $A_{j'}$-heavy (thus also $A_j$-heavy by \Cref{ob:MonotonicHeavy}), or
    \item (Case 1b) the edge $e \in L_{j'}$ incident to $v'$ is fingerprinted, thus so is $v'$.
    \end{itemize}

\noindent{\underline{Case 1b.}} The distance in $\disc_i$ between $u$ and $v'$ is at most $2\hdiam$ because of type-\ref{type:Edge2} edges (as $u$ and $v'$ both appear in the component $C'$ and are both fingerprinted).
Moreover, the distance between $v'$ and the vertex of any component containing $v'$ is at most $\hdiam$ because of type-\ref{type:Edge2} edges. Therefore, for every component $C \subseteq S \in \calN_j$ such that $v' \in C$ and $u \in S$, the distance in $\disc_i$ between $u$ and $\pi(C)$ is as stated in the lemma at most 
\[
2\hdiam + \hdiam\leq (2\hed \sed + 6\hdiam) \cdot j.
\]

\noindent{\underline{Case 1a.}} 
Otherwise, the intermediate claim holds, as $C'$ is $A_{j}$-heavy and the distance in $\disc_{i}$ between $u$ and $\pi(C')$ is at most $\hdiam \leq (2\hed \sed + 6\hdiam) \cdot (j - 0.5)$ because of type-\ref{type:Edge2} edges. %

\medskip

\noindent{\textbf{Case 2.}}
    Now, assume $S'$ consists of multiple components, which means $C'$ is a strict subset of $S'$. Then, the component $C'$ must be $A_{j'}$-heavy by \Cref{coro:CertifyHeavyComp} and the fact that $v'\in C'$ is incident to an edge in $L_{j'}$ but $v'$ is not a waypoint (as $u \neq v'$ and the path $P$ contains no waypoints other than $u$). We also know $C'$ is $A_j$ and $A_{j - 1}$-heavy by \Cref{ob:MonotonicHeavy}, which establishes the first half of the intermediate claim.

\medskip

\noindent{\underline{Case 2a.}}
    Suppose we are in the base case $j=1$. Consider the component $C_{u}\subseteq S'$ with $u\in C_{u}$, and we wish to say $C_{u}$ is $A_{0}$-heavy. Let $e$ be the first edge of $P$, which means $e$ is incident to $u \in C_{u}$ and $e$'s other endpoint is not a waypoint (since $P$ contains no waypoint other than $u$). Also, we have $e\in L_{0} = E$. By \Cref{coro:CertifyHeavyComp}, we know $C_{u}$ is $A_{0}$-heavy.
    
    Thus, there is a type-\ref{type:Edge3} edge of length $\hed \sed + 2\hdiam$ connecting the vertices $\pi(C_{u}),\pi(C')$ of the components $C_u \subseteq S' \in \calN_{0}$ and $C' \subseteq S' \in \calN_{0}$, as both are $A_0$-heavy and $S'$ contains the waypoint $u$. Also, a type-\ref{type:Edge2} edge of length $\hdiam$ connects $u$ and $\pi(C_{u})$. Thus, in $\disc_{i}$, $u$ is within distance $\hdiam + \hed \sed + 2\hdiam \leq (2\hed \sed + 6\hdiam) \cdot (j - 0.5)$ of $\pi(C')$, as required by the intermediate claim.

\medskip

\noindent{\underline{Case 2b.}}
    If now $j > 1$, by induction the claim holds for all smaller $j$. Applying \Cref{lem:hit-in-gminusf} to $P$ and $j - 1$, we get that there must be a prefix $P'$ of $P$ ending in some vertex $v''$ such that 
    \[
    \len_{G_{j - 2}}(P') \leq \len_{G_{j - 1}}(P) + \frac{h}{d} \leq h \cdot \left(1 + \frac{d - j + 2}{d}\right)
    \]
    and $v''$ is either $v'$ or is followed on $P$ by an edge in $L_{j - 1}$, in either case being incident to an edge in $L_{j''}$ for some $j'' \geq j - 1$ (concretely if $v''=v'$, then $j''=j'\geq j$, otherwise $j''=j-1$). 
    Additionally, we know that $P$ is a lex-max shortest path in $G \setminus F$ not containing any waypoints other than its source $u$, and so is $P'$ as a prefix of $P$. Thus, $P'$ meets the induction assumption.

    Let $C'' \subseteq S'$ be the component of $S'$ containing $v''$, which must exist as $S'$ contains the entirety of $P$. By induction, the distance in $\disc_{i}$ from $u$ to the vertex $\pi(C'')$ of the component $C''$ is at most $(2\hed \sed + 6\hdiam) \cdot (j - 1)$.

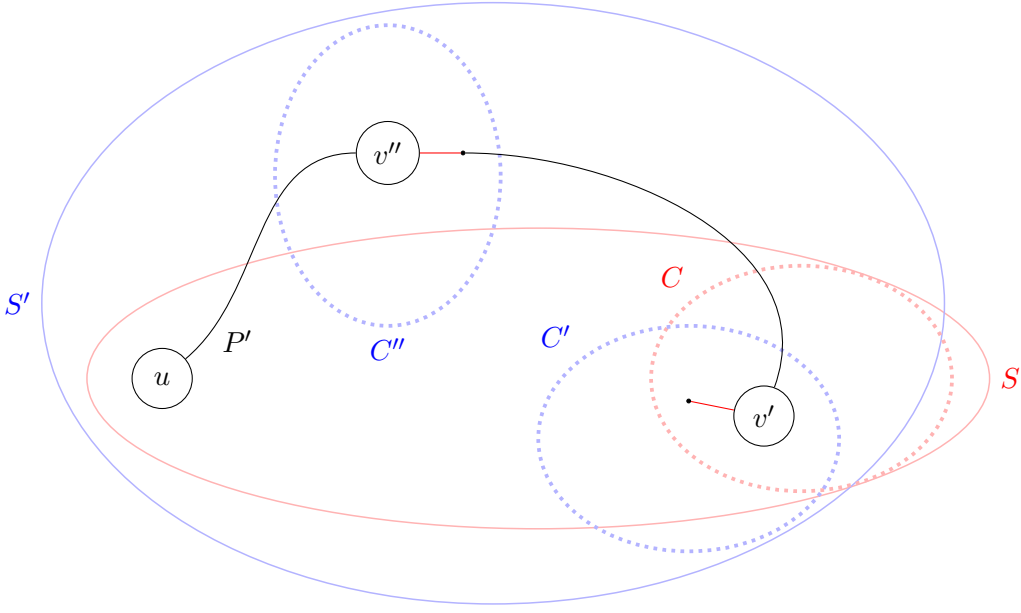
\begin{figure}[h]
    \centering
    \begin{tikzpicture}[main/.style={draw, circle, minimum size=0.8cm}, minor/.style={draw, circle, minimum size=0.05cm, inner sep=0pt}, edges/.style={draw}, light/.style = {fill=none}, heavy/.style = {fill=slightlylightgray}, txt/.style = {draw=none}, every loop/.style={}]

        \node[main] (u) at (-5, 0) {$u$};
        \node[main] (v2) at (-2, 3) {$v''$};
        \node[minor, fill=black] (dummy1) at (-1, 3) {};
        \node[main] (v) at (3, -0.5) {$v'$};
        \node[minor, fill=black] (dummy2) at (2, -0.3) {};

        \node[txt] (lp) at (-4, 0.5) {$P'$};

        \node[ellipse, draw, color=red!30, line width=0.02cm, minimum width=12cm, minimum height=4cm, label=right:${\color{red}S}$] (S) at (0, 0) {};

        \node[ellipse, draw, color=blue!30, line width=0.02cm, minimum width=12cm, minimum height=8cm, label=left:${\color{blue}S'}$] (S2) at (-0.6, 1) {};
        
        \node[ellipse, draw, color=red!30, line width=0.05cm, dotted, minimum width=4cm, minimum height=3cm, label=north west:${\color{red}C}$] (C) at (3.5, 0) {};

        \node[ellipse, draw, color=blue!30, line width=0.05cm, dotted, minimum width=4cm, minimum height=3cm, label=north west:${\color{blue}C'}$] (C2) at (2, -0.8) {};

        \node[ellipse, draw, color=blue!30, line width=0.05cm, dotted, minimum width=3cm, minimum height=4cm, label=below:${\color{blue}C''}$] (C2) at (-2, 2.7) {};

        \draw[-, edges] (u) to[out=40,in=180] (v2);
        \draw[-, edges, color=red] (v2) to (dummy1);
        \draw[-, edges] (dummy1) to[out=0,in=70] (v);
        \draw[-, edges, color=red] (v) to (dummy2);
    \end{tikzpicture}
    \caption[Case 2b of \Cref{lem:cluster-close}]{Case 2b of \Cref{lem:cluster-close}, where $j > 1$ and the components $C''$ and $C'$ are strict subsets of $S'$. In this picture, $u \neq v'' \neq v'$ and $C$ is a strict subset of $S$, but this does not necessarily hold. The solid ellipses are the clusters $S \in \calN_j$ (red) and $S' \in \calN_{j - 1}$ (blue), and the dotted ellipses are the components $C \subseteq S$ (red) and $C', C'' \subseteq S'$ (blue). The $(u, v')$-path (through $v''$) is $P$, and the prefix from $u$ to $v''$ is $P'$. The red edge following $v''$ on $P'$ is in $L_{j - 1}$, and the red edge incident to $v'$ is in $L_{j'}$ for some $j' \geq j$.%

    }
    \label{fig:cluster-close-example}
\end{figure}

    The component $C''$ must be $A_{j-1}$-heavy: it is a strict subset of $S'$ because $S'$ contains multiple components, and
    \begin{itemize}
    \item when $v''=v'\neq u$, the vertex $v''$ is not a waypoint but is incident to an edge $e\in L_{j'}$, thus $C''$ is $A_{j'}$-heavy by \Cref{coro:CertifyHeavyComp} and also $A_{j-1}$-heavy by \Cref{ob:MonotonicHeavy}. 
    \item otherwise, $v''$ is followed on $P$ by an edge in $L_{j-1}$, and the non-$v''$ endpoint of that edge cannot be a waypoint as it is on the path $P$ but is not $u$, thus $C''$ is $A_{j-1}$-heavy by \Cref{coro:CertifyHeavyComp}. %
    \end{itemize}
    
    Since $C''$ is $A_{j - 1}$-heavy, there is a type-\ref{type:Edge3} edge of length $\hed \sed + 2\hdiam$ connecting the vertices $\pi(C'')$ and $\pi(C')$ of the components $C'' \subseteq S' \in \calN_{j - 1}$ and $C' \subseteq S' \in \calN_{j - 1}$, as both components are $A_{j - 1}$-heavy components of the cluster $S'$ containing a waypoint $u \in S'$. Thus, in $\disc_{i}$, $u$ is within distance $(2\hed \sed + 6\hdiam) \cdot (j - 1) + \hed \sed + 2\hdiam \leq (2\hed \sed + 6\hdiam) \cdot (j - 0.5)$ of $\pi(C')$.

\medskip

\noindent\textbf{Finishing the proof.} We are now ready to complete the proof provided the intermediate claim. Let $C \subseteq S \in \calN_j$ be any component such that $v' \in C$ and $u \in S$. If $C = S$, then the distance between $u$ and the vertex of $C$ is $\hdiam$, and we are done. Otherwise, the component $C$ is $A_{j'}$-heavy by \Cref{coro:CertifyHeavyComp} as the non-waypoint $v'\in C$ is incident to an edge $L_{j'}$. Thus, $C$ is also $A_j$-heavy by \Cref{ob:MonotonicHeavy}. 

The intermediate claim gives that $u$ is within distance $(2\hed \sed + 6\hdiam) \cdot (j - 0.5)$ of the vertex $\pi(C')$ of the $A_j$-heavy component $C' \subseteq S' \in \calN_{j - 1}$ with $u \in S'$. Further, there is a type-\ref{type:Edge3} edge of length $\hed \sed + 2\hdiam$ connecting the vertices of the components $C \subseteq S \in \calN_j$ and $C' \subseteq S' \in \calN_{j - 1}$, as both components are $A_j$-heavy and the clusters $S$ and $S'$ both contain the waypoint $u$. Thus, as desired, the distance in $\disc_i$ between $u$ and the vertex $\pi(C)$ of $C$ is at most 
\[
(2\hed \sed + 6\hdiam) \cdot (j - 0.5) + \hed \sed + 2\hdiam\leq (2\sed\hed + 6\hdiam) \cdot j.
\]
\end{proof}

\begin{restatable}{lemma}{discgraphupperbound}\label{lem:disc-graph-upperbound}
    For any waypoint vertices $x, y \in W$,
    \begin{equation*}
        \dist_{\disc}(x, y) \leq O(\sldd \cdot \sed \cdot d) \cdot \dist_{G \setminus F}(x, y)
    \end{equation*}
\end{restatable}

\begin{proof}
    Fix $s = 50 \sldd \sed d = O(\sldd \cdot \sed \cdot d)$. Let $P^{\text{wit}}$ be the lex-max shortest $(x, y)$-path in $G \setminus F$. As $x$ and $y$ are waypoints, the path $P^{\text{wit}}$ can be broken into a concatenation of edge-disjoint lex-max shortest paths $P$ in $G \setminus F$ between waypoint vertices $u$ and $v$ such that $P$ does not contain any waypoint vertices other than $u$ or $v$. Take the minimum $i$ such that $2^i \geq \len_G(P)$. We show that 
    \begin{equation*}
        \dist_{\disc_i}(u, v) \leq 2^i \cdot \frac{s}{2} \leq \len_G(P) \cdot s.
    \end{equation*}
    Thus, by the triangle inequality, $\dist_{\disc}(x, y) \leq \dist_{G \setminus F}(x, y) \cdot s$.

    Now, we prove the claim. Fix some $u, v, P, i$ and consider the graph $\disc_i$. Recall that $h=2^{i}$. Let $j\in\{0,...,d\}$ be the minimum integer such that $\len_{G_j}(P) \leq h \cdot \left(1 + \frac{d - j }{d}\right)$, which must exist because $G_{d}=G$ has $\len_{G}(P)\leq h$. Let $S \in \calN_j$ be a cluster containing the entire path $P$, which must exist as $\len_{G_j}(P) \leq 2h = \hcov$.

    First, consider the case that $S$ contains only one component $C = S$. Then, $\dist_{\disc_i}(u, v) \leq 2\hdiam = 4h$ through two type-\ref{type:Edge2} edges, and we are done. We hereafter assume that $S$ consists of multiple components.
    
    Next, assume $j = 0$. Let $C_u$ be the component of $S$ containing $u$ and $C_v$ the component containing $v$. If either $C_{u}$ or $C_{v}$ is $A_0$-light, then by \Cref{lem:light-comp-known} as $L_0 = E$ and $S$ consists of multiple components, but $P$ contains no waypoints other than $u$ or $v$, the path $P$ consists of a single, fingerprinted edge, and $\dist_{\disc_i}(u, v) \leq \len_G(P)$. Otherwise, as $C_u \subseteq S \in \calN_0$ and $C_v \subseteq S \in \calN_0$ are both $A_0$-heavy components and $S$ contains a waypoint, there is a type-\ref{type:Edge3} edge of length $\hed \sed + 2\hdiam$ in $\disc_i$ connecting the vertices $\pi(C_u)$ and $\pi(C_v)$ of the components $C_u$ and $C_v$, thus $\dist_{\disc_i}(u, v) \leq \hed \sed + 4\hdiam$.
    
    Now, we may assume $j > 0$ and that $S$ consists of multiple components. By \Cref{lem:hit-in-gminusf}, there exists a prefix $P'_u$ of $P$ to some $v'$ such that $\len_{G_{j - 1}}(P'_u) \leq \len_{G_j}(P) + \frac{h}{d}$ and either $v' = v$ or $v'$ is incident to an edge in $L_j$. However, as
    \begin{equation*}
        \len_{G_{j - 1}}(P) > h \cdot \left(1 + \frac{d - j + 1}{d}\right) \geq \len_{G_j}(P) + \frac{h}{d} \geq \len_{G_{j - 1}}(P'_u)
    \end{equation*}
    $P'_u$ must be a true prefix of $P$, and thus $v'$ is followed on $P$ by an edge $e \in L_j$.

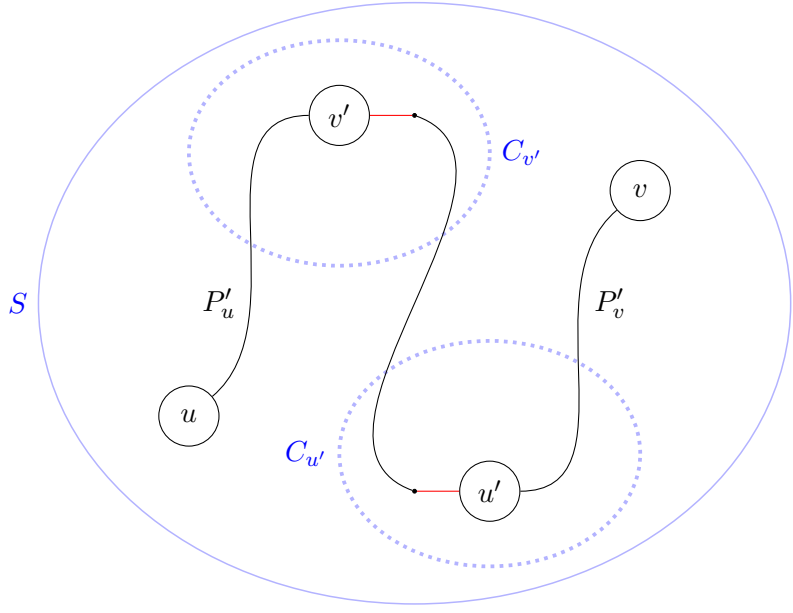
\begin{figure}[h]
    \centering
    \begin{tikzpicture}[main/.style={draw, circle, minimum size=0.8cm}, minor/.style={draw, circle, minimum size=0.05cm, inner sep=0pt}, edges/.style={draw}, light/.style = {fill=none}, heavy/.style = {fill=slightlylightgray}, txt/.style = {draw=none}, every loop/.style={}]

        \node[main] (u) at (-3, -1.5) {$u$};
        \node[main] (v2) at (-1, 2.5) {$v'$};
        \node[minor, fill=black] (dummy1) at (0, 2.5) {};
        \node[main] (v) at (3, 1.5) {$v$};
        \node[main] (u2) at (1, -2.5) {$u'$};
        \node[minor, fill=black] (dummy2) at (0, -2.5) {};

        \node[txt] (lpu) at (-2.6, 0) {$P'_u$};
        \node[txt] (lpv) at (2.6, 0) {$P'_v$};

        \node[ellipse, draw, color=blue!30, line width=0.02cm, minimum width=10cm, minimum height=8cm, label=left:${\color{blue}S}$] (S) at (0, 0) {};
        
        \node[ellipse, draw, color=blue!30, line width=0.05cm, dotted, minimum width=4cm, minimum height=3cm, label=right:${\color{blue}C_{v'}}$] (Cvp) at (-1, 2) {};

        \node[ellipse, draw, color=blue!30, line width=0.05cm, dotted, minimum width=4cm, minimum height=3cm, label=left:${\color{blue}C_{u'}}$] (Cup) at (1, -2) {};

        \draw[-, edges] (u) to[out=40,in=180] (v2);
        \draw[-, edges, color=red] (v2) to (dummy1);
        \draw[-, edges] (v) to[out=220,in=0] (u2); %
        \draw[-, edges, color=red] (u2) to (dummy2);
        \draw[-, edges] (dummy1) to[out=340,in=160] (dummy2);
    \end{tikzpicture}
    \caption[A case of \Cref{lem:disc-graph-upperbound}]{The main case of \Cref{lem:disc-graph-upperbound} where $j > 0$, the components $C_u, C_v$ are true subsets of the cluster $S \in \calN_j$, and the $(u, v)$-path $P$ consists of more than one edge. The path $P'_u$ is the prefix of $P$ from $u$ to $v'$, and the path $P'_v$ is the suffix of $P$ from $v$ to $u'$. The red edges are in $L_j$.} %
    \label{fig:disc-upperbound-example}
\end{figure}
    
    Note that by the above, we additionally have $\len_{G_{j - 1}}(P'_u) \leq h \cdot \left(1 + \frac{d - j + 1}{d}\right)$, thus the conditions of \Cref{lem:cluster-close} are met by $P'_u$ for $j$. Let $C_{v'}$ be the component of $S$ containing $v'$. By \Cref{lem:cluster-close}, the distance in $\disc_i$ between $u$ and the vertex of $C_{v'}$ is at most $(2\hed \sed + 6\hdiam) \cdot d$. Define $C_{u'}$ symmetrically.

    Now, we show that either $C_{v'}$ (and symmetrically $C_{u'}$) is $A_j$-heavy, or the path $P$ consists of only a single fingerprinted edge. 
    \begin{itemize}
        \item If the path $P$ consists of the single edge $e$, then by \Cref{lem:light-comp-known} as $e \in L_j$, $e$ is incident to both $C_{u'}$ and $C_{v'}$, and $S$ contains a waypoint, $e$ is fingerprinted if either $C_{u'}$ or $C_{v'}$ is $A_j$-light.
        \item Otherwise, the path $P$ contains multiple edges. Then, at least one endpoint of $e$ must be neither $u$ or $v$, thus by \Cref{lem:light-comp-known} the component $C_{v'}$ must be $A_j$-heavy, as otherwise both endpoints of $e$ would be waypoints as $e \in L_j$, $e$ is incident to $C_{v'}$, and $S$ consists of multiple components. This is a contradiction as $P$ contains no waypoints other than $u$ or $v$. By symmetry, $C_{u'}$ must also be $A_j$-heavy.
    \end{itemize}
    
    We now complete the proof. If $P$ consists of only a single fingerprinted edge, $\dist_{\disc_i}(u, v) \leq \len_G(e) = \dist_{G \setminus F}(u, v)$ as $e$ appears in $\disc_i$. Otherwise, as both $C_{v'}$ and $C_{u'}$ are $A_j$-heavy and are components of a cluster $S \in \calN_j$ containing a waypoint, there is a type-\ref{type:Edge3} edge connecting the vertices $\pi(C_{v'})$ and $\pi(C_{u'})$ of $C_{v'}$ and $C_{u'}$ of length $\hed \sed + 2\hdiam$. Thus,
    \begin{align*}
        \dist_{\disc_i}(u, v)   &\leq 2d \cdot (2 \hed \sed + 6\hdiam) + \hed \sed + 2\hdiam\\
                                &\leq 2d \cdot (3 \hed \sed + 6\hdiam)\\
                                &= h \cdot \left(\frac{6\hed \sed d}{h} + \frac{12\hdiam d}{h}\right)\\
                                &= 24 h \cdot \left(\sldd \sed d + \sldd d\right)\\
                                &\leq h \cdot 25 \sldd \sed d\\
                                &\leq h \cdot \frac{s}{2} = 2^i \cdot \frac{s}{2} \leq \len_G(P) \cdot s.
    \end{align*}
\end{proof}

We are now ready to prove \Cref{lem:dist-approx-with-labels}. This is done by computing the distance between $\qea$ and $\qeb$ in the discovered graph $\disc$ as discussed before. We note that constructing $\disc$ and running a dijkstra's algorithm on it can be done trivially in $\poly(f,\log n,\log L)$ time (which is already fine for a labeling scheme by convention).

In \Cref{lem:dist-approx-with-labels}, we just want to further optimize the query time. The only complication is that the discovered graph might not have size $\tilde{O}(|\vlabel(\qea)| + |\vlabel(\qeb)| + \sum_{e \in F} |\elabel(e)|)$ due to type-\ref{type:Edge3} edges, of which there might be quadratically many. This is easily solved by constructing a slightly compressed version $\packdisc$ of the discovered graph, where the distance between any two vertices is still equal to the distance between them in $\disc$. 

\distapproxwithlabels*

\begin{proof}
    For a fixed scale $i$, we define a \textit{packed} discovered graph $\packdisc_i$ as follows. Note that the graph contains some \emph{directed} edges, and when we say \emph{each/every cluster $S$} below, it means each cluster $S\in\bigcup_{j'}{\cal N}_{j'}$.
    \begin{itemize}
        \item \textbf{Vertices}. The vertex set of $\packdisc_i$ consists of the following
        \begin{enumerate}
            \item\label{item:packdiscV1} Every fingerprinted vertex $v$
            \item\label{item:packdiscV2} A vertex $\pi(C) = \pi_{S}(C)$ for every component $C$ of each cluster $S$ with a fingerprinted vertex (i.e. there exists a fingerprinted $v \in S$)
            \item\label{item:packdiscV3} Vertices $\pi^{\text{in}}(S, j)$ and $\pi^{\text{out}}(S, j)$ for every $j$ and cluster $S$ with a waypoint vertex
            \item\label{item:packdiscV4} Vertices $\pi(v, j)$ for every waypoint vertex $v$ and $j$
        \end{enumerate}
        \item \textbf{Edges}. The edge set of $\packdisc_i$ consists of the following
        \begin{enumerate}
            \item\label{type:PackEdge1} Every fingerprinted edge not in $F$ (of the same length, connecting the same vertices)
            \item\label{type:PackEdge2} For every cluster $S$, component $C \subseteq S$ and fingerprinted vertex $v \in C$, an edge of length $\hdiam$ connecting $v$ and $\pi_S(C)$
            \item\label{type:PackEdge3} For every cluster $S$, waypoint vertex $v \in S$, and $j$ such that $S \in \calN_{j'}$ for some $j' \leq j$,
            \begin{itemize}
                \item a \emph{directed} edge of length $(\hed \sed + 2\hdiam) / 4$ from $\pi^{\text{out}}(S, j)$ to $\pi(v, j)$
                \item a \emph{directed} edge of length $(\hed \sed + 2\hdiam) / 4$ from $\pi(v, j)$ to $\pi^{\text{in}}(S, j)$
            \end{itemize}
            \item\label{type:PackEdge4} For every cluster $S$ with a waypoint vertex, and $j$ such that $S\in{\cal N}_{j'}$ for some $j'\leq j$, for every $A_j$-heavy component $C$ of $S$,
            \begin{itemize}
                \item a \emph{directed} edge of length $(\hed \sed + 2\hdiam) / 4$ from $\pi_S(C)$ to $\pi^{\text{out}}(S, j)$
                \item a \emph{directed} edge of length $(\hed \sed + 2\hdiam) / 4$ from $\pi^{\text{in}}(S, j)$ to $\pi_S(C)$
            \end{itemize}
        \end{enumerate}
    \end{itemize}
    As with $\disc$, we let $\packdisc$ be the union of the graphs $\packdisc_i$. Concretely, for each original vertex $v\in V$, merge together each $v\in V(\packdisc_{i})\cap V$ (vertices in $V(\packdisc_{i})\cap V$ are all from \Cref{item:packdiscV1}).  Note that the only difference between $\disc_i$ and $\packdisc_i$ is that the type-\ref{type:Edge3} edges of $\disc_i$ are replaced with some additional vertices and directed edges, that together will simulate the type-\ref{type:Edge3} edges of $\disc_i$.

    Let $k$ denote the total size of labels in the query, and $k_i$ the total size of scale-$i$ labels appearing in labels in the query, as below. Note that $\sum_i k_i = O(k)$, as the labels are stored explicitly.
    \begin{align*}
        k = |\vlabel(\qea)| + |\vlabel(\qeb)| &+ \sum_{e \in F} |\elabel(e)|\\
        k_i = \sum_{w \in W} |\vlabel_i(w)| &+ \sum_{e \in F} |\elabel_i(e)| &\text{for $i \in \{0, 1, \dots, i_{\text{max}}\}$}.
    \end{align*}
    We will show the following three claims.
    \begin{enumerate}
        \item\label{claim:packgraphsimulation} For any $u, v \in V(\packdisc_i) \cap V$, $\dist_{\disc_i}(u, v) = \dist_{\packdisc_i}(u, v)$.
        \item\label{claim:packgraphsize} $|V(\packdisc_i)| + |E(\packdisc_i)| = \tilde{O}(k_i)$.
        \item\label{claim:packgraphconstruct} $\packdisc_i$ can be constructed deterministically in time $\tilde{O}(k_i)$ from the scale-$i$ labels of waypoints and failed edges, i.e. $\{\vlabel_i(w)\}_{w \in W}$ and $\{\elabel_i(e)\}_{e \in E}$.
    \end{enumerate}
    then, \Cref{lem:dist-approx-with-labels} follows by constructing $\packdisc$ and returning UNREACHABLE if $\qeb$ is unreachable from $\qea$, and $\hat{d} = \dist_{\packdisc}(\qea, \qeb)$ otherwise. The desired accuracy of the estimate $\hat{d}$ follows from \Cref{lem:disc-graph-lowerbound}, \Cref{lem:disc-graph-upperbound} and Claim \ref{claim:packgraphsimulation}, and the query time bound follows from Claim \ref{claim:packgraphconstruct} and Djikstra taking $\tilde{O}(k)$ time to run on a graph that, by Claim \ref{claim:packgraphsize}, has size $\sum_i \tilde{O}(k_i) = \tilde{O}(k)$.

    We now proceed to prove each of the three claims.

\medskip

\noindent{\underline{Claim \ref{claim:packgraphsimulation}.}}
        Take any type-\ref{type:Edge3} edge connecting some $\pi_S(C)$ and $\pi_{S'}(C')$ in $\disc_i$, with $S \in \calN_j$ and $S' \in \calN_{j'}$, and a waypoint $v \in S \cap S'$. Then, there is a directed path
        \begin{equation*}
            \pi_S(C) \rightarrow \pi^{\text{out}}(S, \max(j, j')) \rightarrow \pi(v, \max(j, j')) \rightarrow \pi^{\text{in}}(S', \max(j, j')) \rightarrow \pi_{S'}(C')
        \end{equation*}
        in $\packdisc$ of equal length. Next, take any directed path
        \begin{equation*}
            \pi_S(C) \rightarrow \pi^{\text{out}}(S, j'') \rightarrow \pi(v, j'') \rightarrow \pi^{\text{in}}(S', j'') \rightarrow \pi_{S'}(C')
        \end{equation*}
        in $\packdisc$. Then, there is a type-\ref{type:Edge3} edge in $\disc_i$ of the same length connecting $\pi_S(C)$ and $\pi_{S'}(C')$: $v$ must be a waypoint and contained in both $S$ and $S'$, $C$ must be a $j''$-heavy component of $S$, and $C'$ must be a $j''$-heavy component of $S'$. Let $j$ and $j'$ be the minimum values such that $S \in \calN_j$ and $S' \in \calN_{j'}$. Then, we must have $j'' \geq \max(j, j')$ by the construction of $\packdisc$, and by \Cref{ob:MonotonicHeavy}, both $C$ and $C'$ must be $A_{\max(j, j')}$-heavy. Thus, the conditions for a type-\ref{type:Edge3} edge are fulfilled.

        The other parts of $\disc_i$ and $\packdisc_i$ are exactly the same, thus Claim \ref{claim:packgraphsimulation} holds.
        
\medskip

\noindent{\underline{Claim \ref{claim:packgraphsize}.}} Note that $d = O(\log n)$ (so the number of different $j$ is at most $d = \tilde{O}(1)$), since $d$ actually controls the level of the hierarchy and there is no reason to consider hierarchies with more than $O(\log n)$ levels.

We first bound the number of vertices. The number of vertices from \Cref{item:packdiscV1,item:packdiscV4} is clearly bounded by $\tilde{O}(k_{i})$. For vertices from \Cref{item:packdiscV2}, their number is at most the sum of (i) the number of pairs $(S,v)$ where $v$ is a fingerprinted vertex inside the cluster $S$, and (ii) $2$ times (the factor $2$ is from the edge orientations) the number of pairs $(S,e)$ where $e\in F$ is a failed edge on the tree $T_{S}$ of $S$. Note that (i) is bounded by $\tilde{O}(k_{i})$ because each fingerprinted $v$ store all $S\ni v$ in its fingerprint. Also (ii) is bounded by $\tilde{O}(k_{i})$ because each failed edge $e$ stores all clusters $S$ with $e\in T_{S}$. Therefore, item-\ref{item:packdiscV2} vertices have number $\tilde{O}(k_{i})$. Similarly, we can show there are at most $\tilde{O}(k_{i})$ item-\ref{item:packdiscV3} vertices.

Consider the number of edges. In fact, similar to the argument for bounding item-\ref{item:packdiscV2} vertices number, we can also bound the number edges by $\tilde{O}(k_{i})$. In particular, for item-\ref{type:PackEdge2} edges, we further use the fact that for each pair $(S,v)$ where $v$ is a fingerprinted vertex inside the cluster $S$, $v$ belongs to exactly one component $C$ of $S$.

\medskip

\noindent{\underline{Claim \ref{claim:packgraphconstruct}.}}     For every unique cluster identifier $\id_j(S)$ such that there is at least one fingerprinted vertex $v \in S$, we will recover for every connected component $C$ of $T_S \setminus F$ the following \emph{intermediate information}:
    \begin{itemize}
        \item The Euler tour intervals of $T_S$ that corresponds to $C$ (an interval is represented by its start and end positions on $T_{S}$).
        \item For every $j'$, the value $A_{j'}(C)$.
    \end{itemize}
    To do this, we recover for every subtree of $T_S \setminus F$, rooted at a vertex $v$ such that either $v$ is the root of the Euler tour of $T_S$ or the edge $\{u, v\}$ from $v$ to its parent $u$ in $T_{S}$ appears in $F$, the values $\ita_{T_S}(v), \itb_{T_S}(v)$ and the node weight $A_{j'}$ of the subtree of $v$ for all $j'$. From this, we immediately obtain the desired information by simply subtracting the child intervals from each interval.

    First, note that since $S$ contains a fingerprinted vertex $v \in S$, the fingerprint of $v$ stores associated with $\id_j(S)$ the value $|S|$ (from which we get $\ita_{T_S}(r) = 0$ and $\itb_{T_S}(r) = 2|S| - 2$ of the Euler tour root $r$ of $T_S$) and the values $A_{j'}(S)$ for all $j'$, equaling the $A_{j'}$-node weight of the root's subtree.
    
    If there is no edge $e \in F$ such that $\id_j(S)$ appears in $\label_i(e)$ (which occurs exactly when $e \in T_S$), then $T_S \cap F = \emptyset$, and the only connected component $C \subseteq S$ of $T_S \cap F$ consists of the entire cluster $S$, and we are done.

    Otherwise, suppose there is at least one edge $e \in F$ such that $\id_j(S)$ appears in $\label_i(e)$. For each such $e$, from $\fp(e)$, take $\fp(u)$ and $\fp(v)$, and from those find the entry for $\id_j(S)$. From this, we get $\ita_{T_S}$, $\itb_{T_S}$ and the subtree sum in $T_S$ of $A_{j'}$ for all $j'$ for both $u$ and $v$. Since $e \in T_S$, out of $u$ and $v$, the one whose interval is contained in the other's interval has its edge to its parent in $F$, and we add its values to the recovered list. Clearly, exactly one entry is added to the list for every edge $e \in T_S \cap F$.

Providing the intermediate information, we can directly construct the graph $\packdisc$ in time nearly-linear to the graph size. The only part not so trivial is that, for a fingerprinted vertex $v$ and a cluster $S\ni v$, how to find the component $C$ containing $v$ (this is need to construct item-\ref{type:PackEdge2} edges). This can be done in $\tilde{O}(1)$ time by getting the unique interval of $T_{S}\setminus F$ containing $v$ (via binary search), and learning from intermediate information which component $C$ owns this interval.
\end{proof}

\subsection{Label Construction}\label{sec:label-construction} %

\begin{lemma}
Let $G = (V, E, l)$ be a graph with edge lengths $1 \leq l(e) \leq L$, vertex set $V = [n]$, and a unique identifier $\id(e)$ for each edge, and let $\sldd \geq 2$, $\sed \geq 100$, $d \geq 1$ and $f$ be set parameters. The labels from \Cref{def:label-template} can be constructed to have size
\begin{align*}
&|\text{VLabel}(v)|\leq \tilde{O}\left(f^{2} \cdot \frac{n^{2/\sldd}}{\phi} \cdot d^{4}\sldd^{3}\sed \cdot \log^{2}(L)\right)\\
&|\text{ELabel}(v)|\leq \tilde{O}\left(f^{4} \cdot \frac{n^{3/\sldd}}{\phi^{2}} \cdot d^{6}\sldd^{5}\sed^{2} \cdot \log^{3}(L)\right)
\end{align*}
and the number of non-trivial edge labels are $\tilde{O}(dn^{1+1/\sldd}\log L)$, 
where $\phi = \phi_{\constr} = n^{-O(1/d)}\cdot n^{-O(1/\sqrt{\sed})}/\poly\log(n)$ for the constructive labels and $\phi = \phi_{\eexist} = n^{-O(1/d)}\cdot n^{-O(1/\sed)}/\log n$ for the existential labels. The constructive labels can be computed in polynomial time.
\end{lemma}

\paragraph{Construction of the Labels.} We first describe how to deterministically construct the structures in \Cref{def:label-template}, and we only discuss the non-trivial parts.

In \Cref{structure2}, by \Cref{thm:hierarchy-result}, the hierarchy can be constructed in $\poly(n)$ time with sparsity $\phi_{\constr} = n^{-O(1/d)}\cdot n^{-O(1/\sqrt{\sed})}/\poly\log(n)$, and the hierarchy exists (i.e. it can be constructed in exponential time) with sparsity $\phi_{\eexist} = n^{-O(1/d)}\cdot n^{-O(1/\sed)}/\log n$. In \Cref{structure5}, each neighborhood cover ${\cal N}_{j}$ can be constructed in $\poly(n)$ time with width
\[
\omega = O(\sldd n^{1/\sldd})
\]
using \Cref{lem:ldd-alg}. In \Cref{structure6}, for each $S$, we just take the single-source shortest path tree rooted at an arbitrary vertex from $S$ as $T_{S}$. Computing all such $T_{S}$ takes $\poly(n)$ time. It remains to construct the $L_{j}$ in \Cref{structure5}, for which we use \Cref{lemma:derandomized} below, proven in \Cref{sec:derandomized}.

\begin{restatable}{lemma}{derandomized}\label{lemma:derandomized}
Consider an element set $E$, a fractional weight function $w:e\to[0,1]$, and two collections of $E$'s subsets ${\cal P} = \{P_{1},...,P_{\bar{i}}\}$ and ${\cal Q} = \{Q_{1},...,Q_{\bar{j}}\}$, such that
\begin{itemize}
\item for each $1\leq i\leq \bar{i}$, $w(P_{i}) = \sum_{e\in P_{i}}w(e)\geq \tau_{\low}$, and
\item for each $1\leq j\leq \bar{j}$, $w(Q_{j}) = \sum_{e\in Q_{j}}w(e)\leq \tau_{\high}$
\end{itemize}
for some known $0<\tau_{\low}\leq 1$ and $\tau_{\high}\geq 1$. There is a deterministic algorithm that computes a subset $S\subseteq E$ such that
\begin{itemize}
\item for each $1\leq i\leq \bar{i}$, $|P_{i}\cap S|\geq 1$, and
\item for each $1\leq j\leq \bar{j}$, $|Q_{j}\cap S|\leq \alpha$,
\end{itemize}
where $\alpha = O(\log m\cdot \tau_{\high}/\tau_{\low})$ and $m = \bar{i} + \bar{j}$. The running time is $\poly(n,m,\alpha)$.
\end{restatable}

\noindent{\underline{Construction of $L_{j}$.}} For each $j\in\{1,...,d\}$, consider constructing $L_{j}$. We first describe our inputs to 
\Cref{lemma:derandomized}. The element set $E$ is the edge set of $G$, the weight function $w$ is the moving cut function $C_{j}$. Let ${\cal P}$ be the collection of (the edge sets corresponding to) all paths $P$ that either consist of a single edge or are a lex-max shortest path in $G$ such that $C_{j}(P)\geq \thit$. The collection ${\cal Q}$ includes two types of subsets: first, the incident edges of each cluster $S$ with $A_{j}(S)\leq \theavy$; second, the incident edges of each $T_{S}[t,t')$ in the definition of $\elabel_{i}(e)$. 

To be precise, regarding the first-type subsets, for each $j'\in{0,...,d}$ and each $S\in{\cal N}_{j'}$ with $A_{j}(S)\leq \theavy$, we add into ${\cal Q}$ a set $Q\subseteq E$ containing all incident edges of $S$. Note that 
\[
C_{j}(Q)\leq A_{j}(S)\leq \theavy,
\]
where the first inequality is because $\deg_{C_{j}}\preceq A_{j}$ from \Cref{def:hierarchy} of the hierarchy.

Regarding the second-type subsets, for each scale $i\in\{0,1,...,i_{\max}\}$, each edge $e\in E$, each $j'\in\{0,...,d\}$, each cluster $S\in{\cal N}_{j'}$ with $e\in T_{S}$, both orientations of $e$ on $T_{S}$, let $t$ be the position of this orientation of $e$ on $T_{S}$, let $t'$ be the maximum index satisfying $t\leq t'\leq t + 2(|S|-1)$ and $\sum_{v\in T_{S}[t,t')}A_{j}(v)\leq \theavy$. Then we add into ${\cal Q}$ a set $Q\subseteq E$ containing all incident edges of $T_{S}[t,t')$. Similarly, we have
\[
C_{j}(Q)\leq \sum_{v\in T_{S}[t,t')}A_{j}(v)\leq \theavy,
\]

Now we feed the above $E,C_{j},{\cal P},{\cal Q}$ to \Cref{lemma:derandomized} with the known $\tau_{\low} = \thit$ and $\tau_{\high} = \theavy$. Observe that both ${\cal P}$ and ${\cal Q}$ have size $\poly(n)$. We set $L_{j}$ to be the output, and by \Cref{lemma:derandomized}, it has the following properties. 
\begin{enumerate}
\item For each $P\in{\cal P}$, $|P\cap L_{j}|\geq 1$. Namely, $L_{j}$ intersects each path in $P$, exactly what we want for $L_{j}$.
\item\label{item:Q2} For each $Q\in{\cal Q}$, $|Q\cap L_{j}|\leq \alpha = O(\log(|{\cal P}| + |{\cal Q}|)\cdot \tau_{\high}/\tau_{\low}) = O((\sldd\sed d)\cdot f^{2}\log^{2}(n)/\phi)$ (recall the definition of $\thit$ and $\theavy$ from \Cref{def:label-template}). We emphasize that this property is useful when we bound the label size below.
\end{enumerate}
The construction time of $L_{j}$ is $\poly(n)$ by \Cref{lemma:derandomized}.

\medskip

\noindent{\underline{Overall Label Construction Time.}} The overall construction time of the structures is $\poly(n)$ when the hierarchy has sparsity $\phi_{\constr}$. Providing the structures, the construction of the labels is straightforward with an additional $\poly(n)$ time.

\paragraph{Label Size.} Fix a scale $i\in\{0,1,...,i_{\max}\}$. We now bound 
$|\vlabel_{i}(v)|$ and $|\elabel_{i}(v)|$ in terms of $\sldd,\sed,d,f$ and $\phi$ (where $\phi$ is the parameter from \Cref{structure2}). Also, recall that $\omega = O(\sldd n^{1/\sldd})$ is the width of each neighborhood cover ${\cal N}_{j}$, and that $\alpha = O((\sldd\sed d)\cdot f^{2}\log^{2}(n)/\phi)$ is from the above Property \ref{item:Q2}.

\medskip

\noindent{\underline{Vertex Fingerprints.}} By the definition, each $\fp(v)$ takes $\eta_{\vfp} = \tilde{O}(d\omega)$ bits, as it stores a constant amount of $O(\log n)$-bit numbers for each $j\in{0,...,d}$ and each $S\in{\cal N}_{j}$ containing $v$.

\medskip

\noindent{\underline{Edge Fingerprints.}} Each $\fp(e)$ takes $\eta_{\efp} = O(\eta_{\vfp}) + O(\log L)$ bits, as it stores two vertex fingerprints and the length of one edge.

\medskip

\noindent{\underline{Vertex Labels.}} Each $\vlabel_{i}(v)$ takes $\eta_{\vlb} = \eta_{\vfp} + \eta_{\efp}\cdot O(d^{2}\omega\alpha)$ bits, because for each $j,j'\in\{0,...,d\}$ and each $S\in{\cal N}_{j'}$ containing $v$ (of which there are at most $\omega$) with $A_{j}(S)\leq \theavy$, we store the fingerprints of $S$'s incident $L_{j}$-edges (of which there are at most $\alpha$ by Property \ref{item:Q2}). 

\medskip

\noindent{\underline{Edge Labels.}} Each $\elabel_{i}(v)$ takes $\eta_{\elb} = \eta_{\vlb}\cdot O(d^{2}\omega\alpha\log L)$ bits. There is an interval $T_{S}[t,t')$ for $e$ for each $j,j'\in\{0,...,d\}$, each $S\in{\cal N}_{j'}$ with $e\in T_{S}$, and each of the two orientations of  $e$, so there are $O(d^{2}\omega)$ intervals $T_{S}[t,t')$ for $e$. For each interval $T_{S}[t,t')$, we store the vertex labels $\vlabel(v)$ (which includes $i_{\max} = O(\log L)$ many $\vlabel_{i}(v)$) of endpoints of incident $L_{j}$-edges of $T_{S}[t,t')$ (of which there are at most $\alpha$ by Property \ref{item:Q2}).

\medskip

\noindent{\underline{Final Label Size.}} Each vertex label $\vlabel(v)$ takes bits
\[
i_{\max}\cdot \eta_{\vlb} = O(\log L)\cdot \tilde{O}(d^{2}\omega\alpha\cdot(d\omega + \log L)) = \tilde{O}(d^{4}\sldd^{3}n^{2/\sldd}\sed f^{2}\log^{2}(L)/\phi).
\]
Each edge label $\elabel(e)$ takes bits
\[
i_{\max}\cdot \eta_{\elb} = O(\log L)\cdot \tilde{O}(d^{4}\omega^{2}\alpha^{2}\log L\cdot(d\omega + \log L)) = \tilde{O}(d^{6}\sldd^{5}n^{3/\sldd}\sed^{2}f^{4}\log^{3}(L)/\phi^{2}).
\]

\paragraph{Number of Non-Trivial Edge Labels.} Recall that an edge $e$ has a non-trivial $\elabel(e)$ if and only if $\elabel_{i}(e)$ is not empty for some $i$. By definition, fixing $i$, $\elabel_{i}(e)$ is not empty only if $e\in T_{S}$ for some cluster $S\in\bigcup_{j}{\cal N}_{j}$. The number of edges on $T_{S}$ is bounded by the number of vertices in $S$. Therefore, the number of non-trivial edge labels is at most
\[
\sum_{\text{scales }i}\sum_{\text{levels }j}\sum_{S\in{\cal N}_{j}}|S|\leq O(i_{\max}\cdot d\cdot \omega n) = \tilde{O}(dn^{1+1/\sldd}\log L).
\]

\section{Compiling a Distance Oracle for $G \setminus F$}

While the Thorup-Zwick $(2k - 1)$-approximate distance labeling scheme \cite{TZ05} has labels of size $O(k n^{1 / k})$, it only needs $O(k)$ work to answer a distance query between two vertices $u$ and $v$ given (pointers to) the labels of $u$ and $v$. Our main result in \Cref{thm:main label} does not have this advantage, taking an amount of work polynomial in the label size.

In this section, we show that this is not an issue in a setting where the edge failures are fixed first, followed by multiple time-critical distance queries in $G \setminus F$. Specifically, we show how to extend \emph{any} fault-tolerant distance labeling scheme so that given only the labels of the failed edges $F$, one can ``compile'' an approximate distance oracle that can subsequently be used to answer approximate distance queries on $G \setminus F$ extremely efficiently in time $O(k + \log f)$.
\begin{restatable}{theorem}{fastqueryresult}\label{thm:fast-query-result}
For every $k, f \ge 1$, there is a deterministic labeling scheme that for an undirected $n$-vertex graph with polynomially bounded edge lengths $\ell(e) \in [1, \poly(n)]$ undergoing $f$ edge faults, assigns in polynomial time to each edge a label of size $\tilde{O}(f^{4} n^{1 / k})$ and to each vertex a label of size $\tilde{O}(n^{1 / k})$, such that
\begin{itemize}
    \item \textbf{Compilation.} Given only the edge labels of a set $F \subseteq E$ of failed edges of size $|F| \leq f$, one can compute in time $\tilde{O}(f^7 n^{1 / k})$ a data structure $D$ of size $\tilde{O}(f^2 + fk \cdot n^{1 / k})$.
    \item \textbf{Distance Queries.} Given only the data structure $D$ and vertex labels of two query vertices $\qea, \qeb$, one can compute in time $O(k + \log f)$ an $O(k^5)$-approximation of $\dist_{G \setminus F}(\qea, \qeb)$.
\end{itemize}
\end{restatable}
The extension only requires black-box access to the underlying fault-tolerant distance labeling scheme, and we obtain \Cref{thm:fast-query-result} as a direct corollary of a general template in \Cref{lem:fast-query-extension-template}. %

The idea of the extension is to reduce a query for the approximate distance between $\qea$ and $\qeb$ to a query for the approximate distance between two failed edge endpoints $s_\qea$ and $s_\qeb$, which can then be precomputed in the compilation step. Each step works as follows:
\begin{enumerate}
    \item \textbf{Labels.} For the vertex labels, use Thorup-Zwick vertex labels with parameter $k$. For edge labels, use the fault-tolerant distance labeling scheme's edge labels, appended with for both of the edge's endpoint vertices the fault-tolerant scheme's vertex label of that vertex and the Euler tour ranges of that vertex in each of the cluster-trees in Thorup-Zwick.
    \item \textbf{Compilation.} Given only the edge labels of the failed edges $F$, let $S$ be the set of vertex endpoints of the edges in $F$. Compute an $O(|F|^2)$-size table of all pairwise approximate distances between the vertices of $S$ in $G \setminus F$ using the fault-tolerant distance labeling scheme. Then, store in a hash table for each cluster-tree in the Thorup-Zwick distance oracle containing at least one failed edge endpoint, a data structure containing the Euler tour ranges of every failed edge endpoint $s \in S$ appearing in that tree.
    \item \textbf{Distance Queries.} Afterwards, given a query for the approximate distance between $\qea$ and $\qeb$ in $G \setminus F$, compute the (implicit) Thorup-Zwick path $P$ between $\qea$ and $\qeb$ \emph{in $G$} using the labels of $\qea$ and $\qeb$, and the first and last failed edge endpoints $s_{\qea}, s_{\qeb} \in S$ that appear on the path using the labels of $\qea$ and $\qeb$ and the compiled data structure $D$.
    
    If no failed edge endpoint appears on the path, return $\mathrm{length}(P)$. Otherwise, return $\mathrm{length}(P)$ plus the precomputed approximate distance between $s_\qea$ and $s_\qeb$ in $G \setminus F$.
\end{enumerate}
Note that the prefix of $P$ from $\qea$ until $s_\qea$ and the suffix of $P$ from $s_\qeb$ to $\qeb$ contain no failed edges by definition. The $O(k)$-loss in approximation ratio comes from the fact that the distance in the post-failure graph $G \setminus F$ between $s_\qea$ and $s_\qeb$ might be up to $2k$ times larger than between $\qea$ and $\qeb$:
\begin{equation*}
    \dist_{G \setminus F}(s_\qea, s_\qeb) \leq \mathrm{length}(P) + \dist_{G \setminus F}(\qea, \qeb) \leq (2k - 1) \dist_G(\qea, \qeb) + \dist_{G \setminus F}(\qea, \qeb) \leq 2k \dist_{G \setminus F}(\qea, \qeb).
\end{equation*}

Before proving \Cref{lem:fast-query-extension-template}, we first provide a brief overview of the Thorup-Zwick approximate distance labeling scheme's internals in \Cref{sec:thorup-zwick-prelim}, as we cannot quite black box their scheme. Readers familiar with the scheme may skip to \Cref{sec:fast-queries-formal}.

\subsection{The Thorup-Zwick Distance Labeling Scheme} \label{sec:thorup-zwick-prelim}

Let $G$ be an undirected graph with nonnegative edge lengths. For a positive integer parameter $k \geq 1$, the Thorup-Zwick distance oracle structure \cite{TZ05} consists of a sequence of sets $V = A_0 \supset A_1 \supset \dots \supset A_k = \emptyset$, with the following related definitions:
\begin{itemize}
    \item \makebox[11cm]{$\mathrm{Bunch}_i(u) := \{w \in A_i : \dist(u, w) < \min_{w' \in A_{i + 1}} \dist(u, w')\}$\hfill} for $i \in \{0, 1, \dots, k - 1\}$,
    \item $\mathrm{Bunch}(u) := \bigcup_{i} \mathrm{Bunch}_{i}(u)$,
    \item \makebox[11cm]{$\mathrm{pivot}_i(u) := \arg \min_{w \in A_i \cap \mathrm{Bunch}(u)} \dist(u, w)$\hfill} for $i \in \{0, 1, \dots, k - 1\}$,
    \item $\mathrm{Cluster}(w) := \{u : w \in \mathrm{Bunch}(u)\}$.
\end{itemize}
Thorup and Zwick \cite{TZ05} show the following:
\begin{lemma}[\cite{TZ05}]\label{lem:thorup-zwick-construct}
    There is a deterministic, $\tilde{O}(k m n^{1 / k})$-time algorithm that constructs a sequence of sets $V = A_0 \supset A_1 \supset \dots \supset A_k = \emptyset$ and the corresponding bunches, pivots and clusters of vertices, such that the bunch size of each vertex $u$ is bounded by $|\mathrm{Bunch}(u)| = O(k n^{1 / k} \log n)$.
\end{lemma}
\begin{lemma}[\cite{TZ05}]\label{lem:thorup-zwick-approx}
    For any vertex pair $\qea, \qeb$, let $i$ be the minimum index such that at least one of $\mathrm{pivot}_i(\qea) \in \mathrm{Bunch}(\qeb)$ or $\mathrm{pivot}_i(\qeb) \in \mathrm{Bunch}(\qea)$ holds. WLOG assume the former does, and let $w := \mathrm{pivot}_i(\qea)$. Then,
    \begin{equation*}
        \dist(\qea, \qeb) \leq \dist(\qea, w) + \dist(w, \qeb) \leq (2k - 1) \dist(\qea, \qeb).
    \end{equation*}
\end{lemma}
The Thorup-Zwick distance labeling scheme assigns to each vertex $\qea$ a label containing an array of the vertices $\mathrm{pivot}_i(\qea)$, and a hashmap keyed by the bunch's elements $w \in \mathrm{Bunch}(\qea)$ to $\dist(\qea, w)$. This takes $O(k + |\mathrm{Bunch}(\qea)|) = O(k n^{1 / k} \log n)$ space. To answer a query given the labels of $\qea$ and $\qeb$, one can then iterate over the $k$ options for $i$, and check for each in $O(1)$ time if $\mathrm{pivot}_i(\qea) \in \mathrm{Bunch}(\qeb)$, and if it does, the value of the right hand side of the equation in \Cref{lem:thorup-zwick-approx}. %

The notion of clusters is useful for path reporting. If $\qea \in \mathrm{Cluster}(w)$ and $\qea'$ appears on a shortest $(\qea, w)$-path, then $\qea' \in \mathrm{Cluster}(w)$. Thus, one can take a shortest path tree $T_w$ rooted at $w$ on the vertices of $\mathrm{Cluster}(w)$, i.e. a tree satisfying that each edge in $T_w$ appears in $G$, each vertex in $T_w$ is in $\mathrm{Cluster}(w)$, and the unique path to the root in $T_w$ from any vertex $\qea$ in the tree is a shortest $(\qea, w)$-path. This tree $T_w$ is called the \textit{cluster-tree} of $w$.

We have $\mathrm{pivot}_i(\qea) \in \mathrm{Bunch}(\qea)$ by definition. Thus, if $i$ is the minimum index as in \Cref{lem:thorup-zwick-approx} and WLOG $w = \mathrm{pivot}_i(\qea) \in \mathrm{Bunch}(\qeb)$, then
\begin{equation*}
    \dist_G(\qea, \qeb) \leq \dist_G(\qea, w) + \dist_G(w, \qeb) = \dist_{T_w}(\qea, \qeb) \leq (2k - 1) \dist_G(\qea, \qeb).
\end{equation*}

\subsection{The Extension}\label{sec:fast-queries-formal}

We now show the following:
\begin{lemma}\label{lem:fast-query-extension-template}
    Suppose you have a deterministic fault tolerant distance labeling scheme that for any positive integers $k, f \geq 1$ and any $n$-vertex $m$-edge graph $G$ with polynomially bounded edge lengths $\ell(e) \in [1, \poly(n)]$ undergoing $f$ edge failures has
    \begin{itemize}
        \item edge label sizes $S_{edge} = S_{edge}(n, m, k, f)$ and vertex label sizes $S_{vertex} = S_{vertex}(n, m, k, f)$,
        \item preprocessing time $T_{prep} = T_{prep}(n, m, k, f)$ and query time $T_{query} = T_{query}(n, m, k, f)$,
        \item approximation ratio $\alpha = \alpha(n, m, k, f)$.
    \end{itemize}
    Then, there is a deterministic labeling scheme satisfying the following:
    \begin{enumerate}
        \item \textbf{Labels.} The scheme assigns to each edge a label of size $S_{edge} + 2 S_{vertex} + O(k n^{1 / k} \log n)$ and to each vertex a label of size $O(k n^{1 / k} \log n)$. Computing the labels takes time $T_{prep} + \tilde{O}(k m n^{1 / k})$.
        \item \textbf{Compilation.} Given only the edge labels of a set $F \subseteq E$ of failed edges of size $|F| \leq f$, one can compute in time $\tilde{O}(f^2 \cdot T_{query} + f^2 k n^{1 / k}))$ a data structure $D$ of size $\tilde{O}(f^2 + f k \cdot n^{1 / k})$.
        \item \textbf{Distance Queries.} Given only the data structure $D$ and vertex labels of two query vertices $\qea, \qeb$, one can compute in time $O(k + \log f)$ an $O(\alpha k)$-approximation of $\dist_{G \setminus F}(\qea, \qeb)$.
    \end{enumerate}
\end{lemma}

Before proving \Cref{lem:fast-query-extension-template}, we observe that \Cref{thm:fast-query-result} is an immediate corollary of \Cref{lem:fast-query-extension-template} and \Cref{thm:main label}: the edge label size asymptotically equals the original edge label size, the compilation time asymptotically equals $O(f^2)$ times the original query time, the approximation ratio degrades by a multiplicative $O(k)$, and the preprocessing time remains polynomial. The other qualities in \Cref{lem:fast-query-extension-template} are unaffected by the underlying fault-tolerant distance labeling scheme.

\begin{proof}(of \Cref{lem:fast-query-extension-template}). We go through the three steps in order.\\
        \paragraph{Labels.} By \Cref{lem:thorup-zwick-construct}, we can construct in time $\tilde{O}(k m n^{1 / k})$ a sequence $V = A_0 \supset A_1 \supset \dots \supset A_k = \emptyset$ and the corresponding bunches, pivots, clusters, and cluster-trees $T_w$, such that $|\mathrm{Bunch}(u)| = O(k n^{1 / k} \log n)$ for every vertex $u$. We can additionally construct the vertex and edge labels $\vlabel_{FT}(u)$ and $\elabel_{FT}(e)$ of the fault-tolerant approximate distance labeling scheme in time $T_{prep}$. Now, the labels of the scheme are as follows:
    \begin{itemize}
        \item The vertex label $\vlabel(u)$ of each vertex $u$ will consist of $u$ and
        \begin{itemize}
            \item An array indexed by $i \in \{0, 1, \dots, k - 1\}$ of the vertices $\mathrm{pivot}_i(u)$.
            \item A hashmap indexed by $w \in \mathrm{Bunch}(u)$ of pairs $(\dist_G(u, w), [\ita_{T_w}(u), \itb_{T_w}(u)])$.
        \end{itemize}
        \item The edge label $\elabel(e)$ of each edge $e = \{u, v\}$ will consist of
        \begin{itemize}
            \item The fault-tolerant scheme's edge label $\elabel_{FT}(e)$.
            \item The fault-tolerant scheme's vertex labels $\vlabel_{FT}(u), \vlabel_{FT}(v)$ and the vertex labels $\vlabel(u), \vlabel(v)$ of the edge's endpoints.
        \end{itemize}
    \end{itemize}
    By the size bound $|\mathrm{Bunch}(u)| = O(k n^{1 / k} \log n)$ on bunches, we now clearly satisfy the claimed label size bounds, and the labels can be constructed in the claimed time.
        
    \paragraph{Compilation.} Let $S$ be the set of endpoints of failed edges $e \in F$. Since the edge labels of the scheme contain the fault tolerant scheme's vertex labels for $S$ and edge labels for $F$, we can by making $\binom{|S|}{2} = O(f^2)$ queries to the fault-tolerant distance labeling scheme compute $\alpha$-approximate pairwise distances $\hat{d}(s, s')$ satisfying $\dist_{G \setminus F}(s, s') \leq \hat{d}(s, s') \leq \alpha \cdot \dist_{G \setminus F}(s, s')$ between pairs of vertices $s, s' \in S$. These queries take $T_{query} \cdot O(f^2)$ time in total.

    Consider some cluster-tree $T_w$ in which at least one failed edge endpoint appears (i.e. $w \in \mathrm{Bunch}(s)$ holds). We will want to use the compiled data structure $D$ to find the first and last vertices in $S$ on any $(u, v)$-path in $T_w$, given the Euler tour intervals of $u$ and $v$ in $T_w$. For this, note that it suffices to find the first and last vertices on the $(u, w)$-path, and the first and last vertices on the $(v, w)$-path.

    Recall that a vertex $v$ is on the path to the root from $u$ in $T_w$ if and only if $\ita_{T_w}(v) \leq \ita_{T_w}(u)$ and $\itb_{T_w}(u) \leq \itb_{T_w}(v)$. Since the Euler tour intervals are laminar, it in fact suffices that $\ita_{T_w}(v) \leq \ita_{T_w}(u) \leq \itb_{T_w}(v)$. Thus, a vertex $s \in S$ is on the root path from $u$ if and only if $\ita_{T_w}(s) \leq \ita_{T_w}(u) \leq \itb_{T_w}(s)$, the first such vertex is the one with the shortest interval, and the last such vertex is the one with the longest interval.

    Thus, we simply need a data structure that stores a laminar set of $n' \leq f$ intervals, and supports $O(\log n')$-time queries for the longest and shortest interval containing a given integer point $x$. Note that if we add or subtract $1$ from the query point $x$, the answer can only change when $x \in \{a - 1, a, b, b + 1\}$ for some stored interval $[a, b]$. Thus, we can preprocess all queries at the $O(n')$-many such points, and binary search the closest such point to our query point in time $O(\log n')$%
    , returning its corresponding answer. The preprocessing can trivially be done in $\tilde{O}((n')^2)$ time, and storing the structure takes only $O(f)$ space.

    Now, the data structure $D$ compiled by the scheme will consist of the following:
    \begin{itemize}
        \item A two-dimensional hashmap keyed by pairs $(s, s') \in S^2$ to the approximate distances $\hat{d}(s, s')$.
        \item A hashmap keyed by $w \in \bigcup_{s \in S} \mathrm{Bunch}(s)$ to a data structure storing for $s \in \mathrm{Cluster}(w) \cap S$ intervals $[\ita_{T_w}(s), \itb_{T_w}(s)]$, supporting $O(\log f)$-time queries for the corresponding $s$ of the shortest and longest intervals containing a query point $x$.
    \end{itemize}
    Since $|\mathrm{Bunch}(s)| = \tilde{O}(k n^{1 / k})$ for each of the $O(f)$ points in $S$, the total preprocessing time of the data structures in the hashmap is $\tilde{O}(f^2 k n^{1 / k})$ and their total size is $\tilde{O}(f k n^{1 / k})$. Thus, the compilation takes time $\tilde{O}(f^2 \cdot T_{query} + f^2 k n^{1 / k})$ and $D$ has size $\tilde{O}(f^2 + fk \cdot n^{1 / k})$, as desired. 
        
    \paragraph{Distance Queries.} To answer a distance query, first, using the labels of $\qea$ and $\qeb$, iterate over the $O(k)$ options to find the minimum $i$ such that $\mathrm{pivot}_i(\qea) \in \mathrm{Bunch}(\qeb)$ or the converse holds. Since the pivots of $\qea$ and $\qeb$ are stored in an array and their bunches as hashmaps, with $O(1)$ time access and membership queries, this takes $O(k)$ time.

    Now, assume WLOG that $\mathrm{pivot}_i(\qea) \in \mathrm{Bunch}(\qeb)$ and let $w := \mathrm{pivot}_i(\qea)$. Then, by \Cref{lem:thorup-zwick-approx}, $\dist_{T_w}(\qea, \qeb) \leq (2k - 1) \dist_G(\qea, \qeb)$. Use the entry of the data structure keyed by $w$ to find the first and last failed edge endpoints $s_\qea, s_\qeb \in S$ on the unique $(\qea, \qeb)$-path in $T_w$. Then,
    \begin{itemize}
        \item If no vertex in $S$ appears on the tree path, i.e. $s_\qea$ and $s_\qeb$ are null, return $\dist_{T_w}(\qea, w) + \dist_{T_w}(w, \qeb)$. These two distances can be recoved in time $O(1)$ from the hashmaps for the bunches of $\qea$ and $\qeb$, which store associated with $w$ their respective distances $\dist_G(\qea, w) = \dist_{T_w}(\qea, w)$ and $\dist_G(\qeb, w) = \dist_{T_w}(\qeb, w)$ to $w$.
        
        To bound the approximation ratio in this case, since the $(\qea, \qeb)$-path in $T_w$ contains no failed edge endpoints, it is preserved in $G \setminus F$ and $\dist_{G \setminus F}(\qea, \qeb) \leq \dist_{T_w}(\qea, w) + \dist_{T_w}(w, \qeb)$. On the other hand, for the upper bound
        \begin{equation*}
            \dist_{T_w}(\qea, w) + \dist_{T_w}(w, \qeb) \leq (2k - 1) \dist_G(\qea, \qeb) \leq (2k - 1) \dist_{G \setminus F}(\qea, \qeb).
        \end{equation*}
        \item Otherwise, return $(\dist_{T_w}(\qea, w) + \dist_{T_w}(w, \qeb)) + \hat{d}(s_\qea, s_\qeb)$. These three summands can each be recovered in time $O(1)$. To bound the approximation ratio, first note that
        \begin{align*}
            \dist_{G \setminus F}(\qea, s_\qea) + \dist_{G \setminus F}(s_\qeb, \qeb) &\leq \dist_{T_w}(\qea, w) + \dist_{T_w}(w, \qeb)\\
                &\leq (2k - 1) \dist_{G}(\qea, \qeb)\\
                &\leq (2k - 1) \dist_{G \setminus F}(\qea, \qeb).
        \end{align*}
        Since the paths in $T_w$ from $\qea$ to $s_\qea$ and from $s_\qeb$ to $\qeb$ are preserved in $G \setminus F$. Thus, we have
        \begin{equation*}
            \dist_{G \setminus F}(s_\qea, s_\qeb) \leq \dist_{G \setminus F}(\qea, s_\qea) + \dist_{G \setminus F}(s_\qeb, \qeb) + \dist_{G \setminus F}(\qea, \qeb) \leq 2k \dist_{G \setminus F}(\qea, \qeb)
        \end{equation*}
        And since $\hat{d}(s_\qea, s_\qeb) \leq \alpha \dist_{G \setminus F}(s_\qea, s_\qeb)$, we have $\hat{d}(s_\qea, s_\qeb) \leq \alpha \cdot 2k \cdot \dist_{G \setminus F}(\qea, \qeb)$ and thus
        \begin{equation*}
            (\dist_{T_w}(\qea, w) + \dist_{T_w}(w, \qeb)) + \hat{d}(s_\qea, s_\qeb) \leq (\alpha \cdot 2k + 2k - 1) \cdot \dist_{G \setminus F}(\qea, \qeb) = O(\alpha k) \dist_{G \setminus F}(\qea, \qeb).
        \end{equation*}
        Finally, for the lower bound, we have
        \begin{align*}
            (\dist_{T_w}(\qea, w) + \dist_{T_w}(w, \qeb)) + \hat{d}(s_\qea, s_\qeb) &\geq (\dist_{G \setminus F}(\qea, s_\qea) + \dist_{G \setminus F}(s_\qeb, \qeb)) + \dist_{G \setminus F}(s_\qea, s_\qeb)\\
                &\geq \dist_{G \setminus F}(\qea, \qeb).
        \end{align*}
    \end{itemize}
    The two parts of the query that take non-constant time are finding the index $i$ in $O(k)$ time and finding $s_\qea$ and $s_\qeb$ in $O(\log f)$ time, thus the query takes $O(k + \log f)$ time.
\end{proof}

\section{Sensitivity Oracles}
\label{sect:Oracle}

In this section we state our results in the sensitivity oracle setting. The first result in \Cref{thm:oracle} is a corollary of \Cref{thm:labeling}.

\begin{theorem}\label{thm:oracle}
    Let $G$ be a graph $(V,E,\ell)$ with polynomially bounded edge lengths $\ell(e) \in [1, \poly(n)]$. For given parameters $k, f \geq 1$, for some $s = O(k^4)$, there is a data structure supporting the following operation:
    \begin{itemize}
        \item \textbf{Distance Query in $G \setminus F$}: given two vertices $\qea, \qeb \in V$ and a set $F \subseteq E$ of up to $f$ edge failures, return an $s$-approximation of $\dist_{G \setminus F}(\qea, \qeb)$. This operation takes time $\tilde{O}(f^5 \cdot n^{1 / k})$. %
    \end{itemize}
    The data structure takes $\tilde{O}(f^4 \cdot n^{1 / k})$ space and can be constructed in polynomial time. The construction and operation of the data structure are deterministic.
\end{theorem}
\begin{proof}
    We consider the constructive labeling scheme from \Cref{thm:labeling}. The data structure stores all vertex labels and the \emph{non-trivial} edge labels. We do not need to store the trivial edge labels because by convention, because by convention, the failure edges $F$ in the query are given by their identifiers. There are $\tilde{O}(n^{1 + O(s^{-1/4})})$ non-trivial edge labels, and each individual label has size at most $\tilde{O}(f^4 \cdot n^{O(s^{-1 / 4})})$. Thus, the data structure takes $\tilde{O}(f^4 \cdot n^{1 + O(s^{-1/4})})$ memory total, which is $\tilde{O}(f^4 \cdot n^{1 + 1 / k})$ for sufficiently large $s = O(k^4)$. Finally, the construction time and query time directly follow from \Cref{thm:labeling}.
\end{proof}

The second result is a corollary of \Cref{thm:fast-query-result}.

\begin{restatable}{theorem}{oraclefastquery}\label{thm:oracle-fastquery}
    Let $G$ be a graph $(V,E,\ell)$ with polynomially bounded edge lengths $\ell(e) \in [1, \poly(n)]$. For given parameters $k, f \geq 1$, for some $s = O(k^5)$, there is a data structure supporting the following operations:
    \begin{itemize}
        \item \textbf{Change Failures}: sets the set of failed edges to a given set $F \subseteq E$ of size $|F| \leq f$. This operation takes time $\tilde{O}(f^7 \cdot n^{1 / k})$.
        \item \textbf{Distance Query in $G \setminus F$}: given two vertices $\qea, \qeb \in V$, return an $s$-approximation of $\dist_{G \setminus F}(\qea, \qeb)$ for the current set $F$. This operation takes time $O(k + \log f)$.
    \end{itemize}
    The data structure takes $\tilde{O}(f^4 \cdot n^{1 / k})$ space and can be constructed in polynomial time. The construction and operations of the data structure are deterministic.
\end{restatable}

\begin{proof}
    We consider the constructive labeling scheme from \Cref{thm:fast-query-result}. The edge labels of \Cref{thm:fast-query-result} consist of the edge labels of \Cref{thm:labeling} and the vertex labels of both of the edge's endpoints. The data structure stores all vertex labels and the \emph{non-trivial} edge labels of \Cref{thm:labeling} in the scheme. We do not need to store the trivial edge labels, because by convention, the failure edges $F$ in the query are given by their identifiers.

    In \Cref{thm:labeling}, there are $\tilde{O}(n^{1 + O(s_\text{ft}^{-1/4})})$ non-trivial edge labels, and each individual label has size at most $\tilde{O}(f^4 \cdot n^{O(s_\text{ft}^{-1 / 4})})$. Thus, storing all of these labels takes $\tilde{O}(f^4 \cdot n^{1 + O(s_\text{ft}^{-1/4})})$ memory total, which is $\tilde{O}(f^4 \cdot n^{1 + 1 / k})$ for sufficiently large $s_\text{ft} = O(k^4)$ used for \Cref{thm:labeling} as in \Cref{thm:fast-query-result}. Storing all the additional vertex labels of \Cref{thm:fast-query-result} takes merely $\tilde{O}(n^{1 + 1 / k})$ space.

    Now, on a \textit{change failures} -query, we compile using the labels of the failed edges an auxiliary data structure in time $\tilde{O}(f^7 \cdot n^{1 / k})$ as in \Cref{thm:fast-query-result}. On a \textit{distance in $G \setminus F$} -query, we use the last compiled auxiliary data structure and the labels of the query endpoints to compute an $s = O(k^5)$-approximation of $\dist_{G \setminus F}(\qea, \qeb)$ in time $O(k + \log f)$.
\end{proof}

\bibliographystyle{alpha}
\bibliography{ref,ref_label}

\appendix

\section{Polytime Approximately Sparsest Length-Constrained Cut}\label{sec:polytime-approxcut}

This appendix proves the following result:
\cutorcertify

We claim no novelty for the results of this section; they are obtained through a straightforward adaptation of Section 3.4 of \cite{OrigLCED22} to general length slack parameter ranges and the node weighting setting. Additionally, we point out that \cite{ImprovedLCED24} proves a similar result with worse tradeoffs, but near-linear work.

\begin{definition}[LDD demand]
    Let $G = (V, E)$ be a graph with edge lengths $l$ and capacities $u$, $A$ a node weighting, and $\calN$ a neighborhood cover of $G$ with covering radius $\hcov$, weak cluster diameter $\hdiam$ and width $\omega$. The \textit{LDD demand} $D_{A, \calN} := \frac{1}{\omega} \sum_{S \in \calS \in \calN} D_{A, S}$ is the scaled-down sum of mixing demands $D_{A, S}$ on clusters $S \in \calS \in \calN$, where $D_{A, S}(u, v) := \frac{A(u) A(v)}{2\sum_{v' \in S} A(v')}$ for $u, v \in S$. 
\end{definition}

\begin{lemma}\label{lem:ldd-demand-properties}
    The LDD demand $D_{A, \calN}$ is $A$-respecting and $\hdiam$-length. Furthermore, if $D_{A, \calN}$ can be routed in $G$ with congestion $\gamma$ and length $\hdiam \cdot s'$, then $A$ is $(\hcov, s)$-length $\phi$-expanding in $G$ for $s = 4s' \cdot \frac{\hdiam}{\hcov}$ and $\phi = \frac{1}{4 \gamma \omega}$.
\end{lemma}
\begin{proof}
    The LDD demand is $\hdiam$-length, as each individual demand $D_{A, S}$ is contained in a cluster of diameter $\hdiam$, and $A$-respecting, as each vertex $u$ appears in at most $\omega$ clusters $S$, and
    \begin{equation*}
        \load(D_{A, S})(u) = \sum_{v \in S} 2 D_{A, S}(u, v) = \sum_{v \in S} 2 A(u) A(v) / 2 \sum_{v' \in S} A(v') = A(u).
    \end{equation*}
    Suppose now that the LDD demand can be routed in $G$ with congestion $\gamma$ and length $\hdiam \cdot s'$, and let $D'$ be any $A$-respecting $\hcov$-length demand. We show that $D'$ is routable in $G$ with congestion $2 \gamma \omega \leq \frac{1}{2\phi}$ and length $2 \hdiam \cdot s' \leq (\hcov \cdot s) / 2$. Thus, by the second part of \Cref{thm:lce-routing} and arbitrary choice of $D'$, $A$ must be $(\hcov, s)$-length $\phi$-expanding in $G$.

    Let $F = \sum_{S \in \calS \in \calN} \sum_{u \in S} F_{S, u}$ be a flow routing $D_{A, \calN}$ with congestion $\gamma$ and length $\hdiam \cdot s'$, with $F_{S, u}$ being the subflow from $u \in S$ to other vertices in $S$ (of flow value $|F_{S, u}| = A(u) / \omega$). To show routability of $D'$, consider a vertex pair $(u, v) \in \supp(D')$. We have $\dist_G(u, v) \leq \hcov$, thus there is at least one cluster $S \in \calS \in \calN$ containing both $u$ and $v$. The concatenation of $\frac{\omega}{A(u)} \cdot D'(u, v) F_{S, u}$ with $\frac{\omega}{A(v)} \cdot D'(u, v) F_{S, v}$ with the direction of each flow path reversed is a single-commodity flow from $u$ to $v$ of value $D'(u, v)$. Assign to each vertex pair $(u, v) \in \supp(D')$ one flow $F'_{u, v}$ constructed as such.
    
    Now, the flow $F' = \sum_{u, v} F'_{u, v}$ routes $D'$. The length of this flow is at most two times the length of the flow $F$, thus at most $2 \hdiam \cdot s'$. Since $D'(u, v) \leq A(u)$, any flow $F_{S, u}$ is added with a multiple of at most $2 \omega$ into $F'$. Thus, the congestion of $F'$ is at most $2 \omega$ times the congestion of $F$, thus at most $2 \gamma \omega$. We are done.
\end{proof}

To find an approximately sparsest $(h, s)$-length moving cut with respect to a fixed demand, we blackbox a multicommodity cutmatch algorithm of \cite{LCFlow23}.

\begin{definition}[Multi-Commodity $(h, s)$-length Cutmatch, \cite{LCFlow23}]
    Given a graph $G = (V, E)$ with edge lengths $l$ and capacities $u$, a $h$-length $\phi$-sparse \textit{cutmatch} of congestion $\gamma$ between disjoint and equal-size node weighting pairs $\{(A_i, A_i')\}_{i \in [k]}$ consists of
    \begin{itemize}
        \item For each $i \in [k]$, a partition of the supports of $A_i$ and $A_i'$ into "matched" and "unmatched" parts $M_i \sqcup U_i = \supp(A_i)$ and $M'_i \sqcup U'_i = \supp(A'_i)$.
        \item A $hs$-length flow $F = \sum_{i \in [k]} F_i$ of congestion $\gamma$ where for each $i$, $F_i$ is a flow from $A_i$ to $A'_i$ such that the total flow value of flow paths in $F_i$ from any vertex $v \in \supp(A_i)$ is at most $A_i(v)$ (with equality iff $v \in M_i$), and the total flow value of flow paths in $F_i$ to any vertex $v' \in \supp(A_i')$ is at most $A_i'(v)$ (with equality iff $v' \in M'_i$).
        \item A $hs$-length moving cut $C$ in $G$, such that $\dist_{G - C}(U_i, U_i') \geq h$ for all $i \in [k]$, and the size of $C$ is at most $|C| \leq \phi \cdot \left(\sum_{i \in [k]} |A_i| - |F_i|\right)$.
    \end{itemize}
    For a demand $D$ for which $\supp(D) = \sqcup_{i \in [k]} (v_i, v_i')$, a $(h, s)$-length $\phi$-sparse cutmatch of congestion $\gamma$ is a $(h, s)$-length $\phi$-sparse cutmatch of congestion $\gamma$ between $\{(A_i, A_i')\}_{i \in [k]}$ where $|A_i| = |A_i'| = A_i(v_i) = A_i'(v_i') = D(v_i, v_i')$.
\end{definition}

\begin{theorem}[Theorem A.2 of \cite{LCFlow23}]\label{thm:cutmatch-alg}
There is an algorithm that, given a graph $G = (V, E, u, l)$ with edge lengths $l \geq 1$ and capacities $u \geq 1$, a length constraint $h \geq 1$, a sparsity parameter $\phi \leq 1$ and disjoint and equal-size node weighting pairs $\{(A_i, A'_i)\}_{i \in [k]}$, computes in time $\tilde{O}(m \cdot k \cdot h^{17})$ a $(h, 1)$-length $\phi$-sparse cutmatch of congestion $\gamma = \tilde{O}(1 / \phi)$.
\end{theorem}

As $h$ can be superpolynomial, we obtain polynomial running time (\Cref{cor:cutmatch-rounded-alg}) by a standard rounding trick, at the cost of some length slack.
\begin{corollary}\label{cor:cutmatch-rounded-alg}
    There is an algorithm that, given a graph $G = (V, E, l, u)$ with edge lengths $l \geq 1$ and capacities $u \geq 1$, a length constraint $h \geq 1$, a sparsity parameter $\phi \leq 1$ and disjoint and equal-size node weighting pairs $\{(A_i, A'_i)\}_{i \in [k]}$, computes in polynomial time a $(h, 2)$-length $\phi$-sparse cutmatch of congestion $\gamma = \tilde{O}(1 / \phi)$.
\end{corollary}
\begin{proof}
    Let $G'$ be $G$ except with edge lengths $l'(e) := \lceil n \cdot l(e) / h \rceil$. Now, any vertex-simple path $P$ has $\frac{n}{h} l(P) \leq l'(P) \leq n (l(P) / h + 1)$, thus if $l(P) \leq h$, $l'(P) \leq 2n$, and if $l'(P) \leq 2n$, $l(P) \leq 2h$. Thus, applying \Cref{thm:cutmatch-alg} to $G'$ with length constraint $h' = 2n$ produces a $(2n, 1)$-length cutmatch. This cutmatch in $G$ is $(h, 2)$-length, as desired. The algorithm runs in polynomial time, as $h'$ is polynomial (in fact linear) in $n$, thus $\tilde{O}(m \cdot k \cdot (h')^{17}) = \poly(n)$.
\end{proof}

We now have everything required to prove \Cref{lem:cut-or-certify}.

\cutorcertify*

\begin{proof}
First, note that if $s'$ is not $O(\poly \log n)$, we can set $s' = O(\poly \log n)$, only improving the output. Let $\hcov = h$, $\hdiam = hs'$ and $h_\mathrm{cm} = \frac{s}{8s'} \hdiam = \frac{s}{8} h$. Let $\omega = O(s' n^{O(1 / s')})$ be such that by \Cref{lem:ldd-alg} a neighborhood cover $\mathcal{N}$ of $G$ of covering radius $\hcov$, weak cluster diameter $\hdiam$ and width $\omega$ can be computed in polynomial time.

Let $\gamma = \frac{1}{4 \omega \phi}$ and $\phi' = \tilde{O}(1 / \gamma) = \tilde{O}(\phi \cdot s' n^{O(1 / s')})$ be such that by \Cref{cor:cutmatch-rounded-alg} a $(h_\mathrm{cm}, 2)$-length, $\phi'$-sparse cutmatch of $D_{A, \mathcal{N}}$ of congestion $\gamma$ can be computed in polynomial time.

Note that the flow-part of this cutmatch has length $2 h_\mathrm{cm} = \frac{s}{4s'} \hdiam = \frac{s}{4} h$, and we set $\phi = \frac{1}{4 \omega \gamma}$. Thus, if the cut is empty, by \Cref{lem:ldd-demand-properties}, $A$ is $(h, s)$-length $\phi$-expanding in $G$.

Otherwise, if the cut is nonempty, the algorithm returns the moving cut $4C$ and the demand $D$ where $D(v, v') = D_{A, \mathcal{N}}(v, v')$ if the distance between $v$ and $v'$ in the graph with edge lengths $l + hs \cdot 4C$ is strictly greater than $hs$, and $D(v, v') = 0$ otherwise. For correctness, let $D'$ be the subdemand of $D_{A, \mathcal{N}}$ such that each vertex pair in the support of $D'$ has distance at least $h_\mathrm{cm}$ in the graph with edge lengths $l + 2 h_\mathrm{cm} C$. Since $D_{A, \mathcal{N}}$ is a $hs'$-length demand, each demand pair in $D'$ must have distance strictly greater than $hs$ in the graph with edge lengths $l + hs \cdot 4C$, thus in particular $D'$ is a subdemand of $D$. Additionally, we have $|C| \leq \phi' \cdot (\sum_{i \in [k]} D(v_i, v'_i) - |F_i|) \leq \phi' \cdot |D'| \leq \phi' \cdot |D|$. Thus, $4C$ is $hs$-length $4\phi'$-sparse for $D$, a $hs'$-length $A$-respecting demand.
\end{proof}

\newpage

\section{Union of Sparse Moving Cuts is Small}\label{sect:SmallUnionSparseCut}

In this section, we prove \Cref{thm:SmallUnionSparseCuts}, bounding the size of a union of sparse cuts. We claim no novelty for this proof; it is a straightforward adaptation of the proof of Lemma 5.1 in \cite{OrigLCED22} to general length slack parameter ranges and the node weighting setting.

\SmallUnionSparseCuts*

\begin{proof}
    We define a monotone decreasing potential function $P$ such that the potential $P_i$ of each graph $G_i$ in the sequence satisfies $0 \leq P_i \leq |A| \ln(n)$ and the potential drop applying cut $C_i$ satisfies
    \begin{equation}\label{eq:pot-bound-1}
        P_i - P_{i + 1} \geq \frac{|C_i|}{\phi_i} \cdot n^{-O(1 / s)}.
    \end{equation}
    Summing over this equation for $i \in \{0, 1, \dots, k - 1\}$ then gives the desired result. The specific potential function used is as follows: let $w_i$ for graph $G_i$ be the exponential decay weight function
    \begin{equation*}
    w_i(u,v) := \left\{
    \begin{aligned}
    &n^{-2\dist_{G_i}(u,v)/(hs)} &\text{if }\dist_{G_i}(u,v)\leq hs/2\\
    &0 &\text{if }\dist_{G_i}(u,v)>hs/2
    \end{aligned}
    \right.
    \end{equation*}
    and let $w_i(u) = \sum_{v \in V} w_i(u, v)$. Then, the potential of a vertex $v$ in graph $G_i$ is defined as $P_i(v) := A(v) \ln(w_i(v))$, and the potential of the entire graph is $P_i = \sum_{v \in V} P_i(v)$. Note that $w_i(v, v) = 1$ and $w_i(u, v) \leq 1$, thus $0 \leq \ln(w_i(v)) \leq \ln(n)$, giving the desired bounds on the potential function.

    For all $i \in [k]$, let $D_i$ be a \emph{minimal} $A$-respecting $h$-length demand such that $\spars_{hs}(C_i, D_i) \leq \phi_i$. We show \Cref{eq:pot-bound-1} in two parts, first showing that,
    \begin{equation}\label{eq:pot-bound-2}
        P_i - P_{i + 1} \geq \sum_{u, v \in V} D_i(u, v) \cdot \mathrm{overlap}_i(u, v) 
    \end{equation}
    where
    \begin{equation*}
        \mathrm{overlap}_i(u, v) = \sum_{x \in V} \min\left\{\frac{w_i(u, x)}{w_i(u)}, \frac{w_i(v, x)}{w_i(v)}\right\}
    \end{equation*}
    is a measure of the \textit{overlap} of the weight functions $w_i$ of $u$ and $v$. Then, we show that the overlap of any two nearby vertices must be high, namely that for any $u, v \in V$ such that $\dist_{G_i}(u, v) \leq h$,
    \begin{equation}\label{eq:overlap-bound}
        \mathrm{overlap}_i(u, v) \geq n^{-4/s} / 2.
    \end{equation}
    Then, since each $D_i$ is minimal, we have $|C_i| / \phi_i \leq |D_i|$. Thus,
    \begin{equation*}
        P_i - P_{i + 1} \geq \sum_{u, v \in V} D_i(u, v) \cdot \mathrm{overlap}_i(u, v) \geq \sum_{u, v \in V} D_i(u, v) \cdot \frac{n^{-4 / s}}{2} = |D_i| \cdot \frac{n^{-4 / s}}{2} \geq (|C_i| / |\phi_i|) \cdot \frac{n^{-4 / s}}{2}.
    \end{equation*}

    We now first prove \Cref{eq:pot-bound-2}. Let $\mathrm{drop}_i(v, x) = \frac{1}{w_i(v)} \left(w_i(v, x) - w_{i + 1}(v, x)\right)$ and $\mathrm{drop}_i(v) = \sum_{x \in V} \mathrm{drop}_i(v, x)$. Then, the drop in potential $P_i(v) - P_{i + 1}(v)$ of a vertex $v$ when applying the $i$th cut is at least 
    \begin{equation*}
        P_{i}(v) - P_{i + 1}(v) = A(v) \cdot \ln\left(\frac{w_{i}(v)}{w_{i + 1}(v)}\right) \geq A(v) \cdot \left(1 - \frac{w_{i + 1}(v)}{w_{i}(v)}\right) = A(v)  \cdot \mathrm{drop}_i(v)
    \end{equation*}
    where we used $\ln(x) \geq 1 - \frac{1}{x}$ and $w_{i + 1}(u) = w_i(u) (1 - \mathrm{drop}_i(u))$. Now, take any $(u, v) \in \supp(D_i)$. Note that by minimality of $D_i$, they satisfy $\dist_{G_{i + 1}}(u, v) > hs$. Then, for any vertex $x \in V$, either
    \begin{itemize}
        \item $\dist_{G_{i + 1}}(u, x) > hs / 2$, in which case $w_{i + 1}(u, x) = 0$, thus $\mathrm{drop}_i(u, x) = w_i(u, x) / w_i(u)$.
        \item $\dist_{G_{i + 1}}(v, x) > hs / 2$, in which case $w_{i + 1}(v, x) = 0$, thus $\mathrm{drop}_i(v, x) = w_i(v, x) / w_i(v)$.
    \end{itemize}
    Thus, we have
    \begin{equation*}
        \sum_{x \in V} \mathrm{drop}_i(u, x) + \mathrm{drop}_i(v, x) \geq \sum_{x \in V} \min\left\{\frac{w_i(u, x)}{w_i(u)}, \frac{w_i(v, x)}{w_i(v)}\right\} = \mathrm{overlap}_i(u, v).
    \end{equation*}
    Multiplying by $D_i(u, v)$ and summing over $(u, v) \in \supp(D)$ then gives
    \begin{align*}
        \sum_{u, v \in V} D(u, v) \cdot \mathrm{overlap}_{i}(u, v) &\leq \sum_{u, v \in V} D_i(u, v) \sum_{x \in V} \mathrm{drop}_i(u, x) + \mathrm{drop}_i(v, x)\\
            &\leq \sum_{u \in V} A(u) \sum_{x \in V} \mathrm{drop}_i(u, x)\\
            &= \sum_{u \in V} A(u) \mathrm{drop}_i(u)\\
            &\leq \sum_{u \in V} P_i(u) - P_{i + 1}(u) = P_i - P_{i + 1}.
    \end{align*}
    Finally, \Cref{eq:overlap-bound} is obtained through a straightforward calculation. Let $u, v \in V$ be vertices such that $\dist_{G_i}(u, v) \leq h$, and let $B_u := \mathrm{Ball}_{G_i}(u, hs / 2)$ and $B_v := \mathrm{Ball}_{G_i}(v, hs / 2)$ be closed balls of radius $hs / 2$ around $u$ and $v$ respectively. Then, for all $x \in V$,
    \begin{itemize}
        \item If $x \in B_u$, then $w_i(u,x) \geq n^{-2(h + \dist_{G_i}(v, x)) / (hs)} \geq n^{-2/s} \cdot w_i(v, x)$
        \item If $x \not\in B_u$, then $w_i(v, x) \leq n^{-2(\dist_{G_i}(u, x) - h) / (hs)} \leq n^{-2(hs/2 - h) / (hs)} = n^{2/s - 1}$
    \end{itemize}
    as by the triangle inequality $\dist_{G_i}(v, x) - h \leq \dist_{G_i}(u, x) \leq h + \dist_{G_i}(v, x)$. Thus, we have
    \begin{align*}
        \sum_{x \in V} \min\{w_i(u, x), w_i(v, x)\} &\geq \sum_{x \in B_u \cap B_v} \min\{w_i(u, x), w_i(v, x)\} &\geq n^{-2/s} \cdot \sum_{x \in B_u \cap B_v} w_i(u, x)\\
        \sum_{x \in V} \max\{w_i(u, x), w_i(v, x)\} &\leq \sum_{x \in B_u \cap B_v} n^{2/s} w_i(u, x) + \sum_{x \not\in B_u \cap B_v} n^{2/s - 1}
            &\leq 2n^{2/s} \cdot \sum_{x \in B_u \cap B_v} w_i(u, x)
    \end{align*}
    giving
    \begin{equation*}
    \mathrm{overlap}_i(u, v) \geq \frac{\sum_{x \in V} \min\{w_i(u, x), w_i(v, x)\}}{\sum_{x
     \in V} \max\{w_i(u, x), w_i(v, x)\}} \geq \frac{n^{-2/s} \cdot \sum_{x \in B_u \cap B_v} w_i(u, x)}{2n^{2/s} \cdot \sum_{x \in B_u \cap B_v} w_i(u, x)} = \frac{n^{-4/s}}{2}.
    \end{equation*}
\end{proof}

\section{Proof of \Cref{lemma:derandomized}}\label{sec:derandomized}

This appendix contains the proof of \Cref{lemma:derandomized}. The proof follows the standard strategy of showing that the expected number of violated constraints when sampling $S$ weighted by $w$ is small, then iteratively fixing for each $e \in E$ whether $e \in S$ to minimize this expectation given the fixed prefix.

\derandomized*

We will use the below standard Chernoff bound.

\begin{fact}[Chernoff bound, Theorems 1.10.5 and 1.10.10 in \cite{doerr2018probabilistic}]
\label{fact:Chernoff}
Let $X_{1},...,X_{n}$ be independent random variables taking values in $[0,1]$. Let $X = \sum_{i=1}^{n} X_{i}$. Then for all $0\leq \delta\leq 1$,
\[
\Pr[X\leq (1-\delta)\mathbb{E}(X)]\leq \exp(-\delta^{2}\mathbb{E}[X]/2).
\]
For all $\lambda \geq \mathbb{E}[X]$,
\[
\Pr[X\geq \mathbb{E}[X]+\lambda]\leq \exp(-\lambda/3).
\]
\end{fact}

\begin{proof}[Proof of \Cref{lemma:derandomized}]
    Let $\beta = 100\log m/\tau_{\low}$ and $\alpha = 2\beta\cdot\tau_{\high}$. Consider a set $S'\subseteq e$ constructed by sampling each element $e\in E$ independently with probability 
    \[
    \rho_{e} = \min\{1, w(e)\cdot\beta\}.
    \]
    For each $e\in E$, let $X_{e}$ be the indicator random variable for $e$ being sampled. Moreover, 
    \begin{itemize}
        \item For each $1\leq i\leq \bar{i}$, let $X^p_i := \mathbbm{1}[\sum_{e \in P_i} X_{e} = 0]$ be an indicator random variable for the event that the constraint $|P_i \cap S'| \geq 1$ is violated.
        \item For each $1\leq j\leq \bar{j}$, let $X^q_j := \mathbbm{1}[\sum_{e \in Q_i} X_{e} > \alpha]$ be an indicator random variable for the event that the constraint $|Q_j \cap S'| \leq \alpha$ is violated.
        \item Let $X_{\mathrm{fail}} := \sum_{i} X^p_i + \sum_{j} X^q_j$ be a random variable representing the number of violated constraints.
    \end{itemize}
    \Cref{claim:InitialFailExpect} show that the expectation of $X_{\fail}$ is very small.
    
    \begin{claim}
    \label{claim:InitialFailExpect}
    $\mathbb{E}[X_{\fail}]\leq 0.5$.
    \end{claim}
    \begin{proof}
    By linearity of expectation, it suffices to show that $\Pr[X^{p}_{i}=1]\leq 1/(2m)$ for each $1\leq i\leq \bar{i}$, and $\Pr[X^{q}_{j}=1]\leq 1/(2m)$ for each $1\leq j\leq \bar{j}$.

    Fix an $i$. If there exists $e\in P_{i}$ such that $w(e)\geq 1/\beta$, then we immediately have $\Pr[X^{p}_{i}=1]=0$ since $e$ will be sampled with probability $1$. Now assume each $e\in P_{i}$ has $w(e)\leq 1/\beta$ so $\rho_{e} = w(e)\cdot \beta$. Then $\mu^{p}_{i}:=\mathbb{E}[\sum_{e\in P_{i}}X_{e}] = w(P_{i})\beta$. Since $w(P_{i})\geq \tau_{\low}$, we have $\mu^{p}_{i}\geq 100\log m$. By \Cref{fact:Chernoff}, 
    \[
    \Pr[\sum_{e\in P_{i}}X_{e}\leq 0.5\mu^{p}_{i}]\leq \exp(-\mu^{p}_{i}/8)\leq m^{-10},
    \]
    which means $\Pr[X^{p}_{i} = 1]= \Pr[\sum_{e\in P_{i}}X_{e}=0]\leq m^{-10}$ as desired.
    
    Similarly, fix a $j$, and we define $\mu^{q}_{j}:=\mathbb{E}[\sum_{e\in Q_{j}}X_{e}] = w(Q_{j})\beta$. Since $w(Q_{j})\leq \tau_{\high}$, we have $\mu^{q}_{j}\leq \tau_{\high}\cdot\beta\leq \alpha/2$. By \Cref{fact:Chernoff},
    \[
    \Pr[\sum_{e\in Q_{j}}X_{e}\geq \mu^{q}_{j} + \alpha/2]\leq \exp(-(\alpha/2)/3)\leq m^{-10},
    \]
    In other words, $\Pr[X^{q}_{j}=1]=\Pr[\sum_{e\in Q_{j}}X_{e}> \alpha]\leq m^{-10}$ as desired.
    \end{proof}

    Now, we are ready to describe the deterministic construction of $S$. Let $E = \{v_1, \dots, v_n\}$. Now, we repeatedly fix values $X_{v_t} = x^{*}_{v_t} \in \{0, 1\}$ in increasing order of $t$, maintaining $\mathbb{E}[X_\mathrm{fail}\mid\forall t \leq T, X_{v_t} = x^{*}_{v_t}] \leq 0.5$ for each moment $1\leq T\leq n$. By \Cref{claim:InitialFailExpect}, initially $\mathbb{E}[X_{\fail}] \leq 0.5$. At each moment $1\leq T\leq n$, by linearity of expectation, given that $\mathbb{E}[X_\mathrm{fail}\mid\forall t \leq T-1, X_{v_t} = x^{*}_{v_t}] \leq 0.5$, there exists $x^{*}_{e_{T}}\in\{0,1\}$ such that $\mathbb{E}[X_\mathrm{fail}\mid\forall t \leq T, X_{v_t} = x^{*}_{v_t}] \leq 0.5$ holds. Therefore, we obtain a set $S = \{e_{t}\mid x^{*}_{e_{t}}=1\}$ at the end which will not violate any constraints.

    It remains to show how to compute $\mathbb{E}[X_\mathrm{fail}\mid \forall t \leq T, X_{v_t} = x^{*}_{v_t}]$ given $\{x^{*}_{e_{t}}\mid t\leq T\}$. By linearity of expectation, it suffices to compute $\Pr[\sum_{e\in P_{i}}X_{e} = 0\mid \forall t \leq T, X_{v_t} = x^{*}_{v_t}]$ for each $1\leq i\leq \bar{i}$, and $\Pr[\sum_{e\in Q_{j}}X_{e} \geq \alpha + 1\mid \forall t \leq T, X_{v_t} = x^{*}_{v_t}]$ for each $1\leq j\leq \bar{j}$. Let $V_{\leq T} = \{e_{1},...,e_{T}\}$.
    \begin{itemize}
    \item $\Pr[\sum_{e\in P_{i}}X_{e} = 0\mid \forall t \leq T, X_{v_t} = x^{*}_{v_t}]$ can be computed directly. Concretely, if some $e\in P_{i}\cap V_{\leq T}$ has $x^{*}_{e} = 1$, the probability is $0$. Otherwise, it is $\prod_{e\in P_{i}\setminus V_{\leq T}} (1-\rho_{e})$.
    \item $\Pr[\sum_{e\in Q_{j}}X_{e} \geq \alpha + 1\mid \forall t \leq T, X_{v_t} = x^{*}_{v_t}]$ can be computed via a simple dynamic programming. Concretely, let $\alpha'$ be the number of $e\in Q_{j}\cap V_{\leq T}$ s.t. $x^{*}(e)=1$. We want to compute the probability $\Pr[\sum_{e\in Q_{j}\setminus V_{\leq T}}X_{e}\geq \alpha+1-\alpha']$.
    
    We initialize $f_{0,c} = 0$ for each $1\leq c\leq \alpha+1-\alpha'$ and initialize $f_{0,0}=1$. Say $Q_{j}\setminus V_{\leq T} = \{e'_{1},e'_{2},...,e'_{k}\}$. Then for each $1\leq r\leq k$, compute 
    \begin{equation*}
    f_{r,c} = \left\{
    \begin{aligned}
    &f_{r-1,c-1}\cdot \rho_{e'_{c}} + f_{r-1,c}\cdot (1-\rho_{e'_{c}}),\ \text{for }0\leq c\leq \alpha-\alpha'\\
    &f_{r-1,c-1}\cdot \rho_{e'_{c}} + f_{r-1,c},~~~~~~~~~~~~~~~\text{for }c=\alpha+1-\alpha'
    \end{aligned}
    \right..
    \end{equation*}
    Obviously, the desired probability is $f_{k,\alpha+1-\alpha'}$.
    \end{itemize}
The running time of the whole algorithm is clearly $\poly(n,m,\alpha)$, because we run $O(nm)$ DP and each DP takes $O(n\alpha)$ time.

\end{proof}

\end{document}